%% file: main.tex
\documentclass[11pt, letterpaper]{article}
\input{includes}
\input{defines}
\input{mathshortcut}
\title{\vspace{-1cm} Algorithms for the Diverse-$k$-SAT problem: the geometry of
satisfying assignments}
\author[1]{Per Austrin}
\author[1]{Ioana O. Bercea\footnote{Received funding from Basic Algorithms Research Copenhagen(BARC), supported by VILLUM Foundation Grants 16582 and 54451}}
\author[2]{Mayank Goswami\footnote{Supported by NSF grant CCF-2503086}}
\author[3]{\\Nutan Limaye\footnote{Received funding from the Independent Research Fund Denmark (grant agreement No. 10.46540/3103-00116B) and is also supported by the Basic Algorithms Research Copenhagen (BARC), funded by VILLUM Foundation Grants 16582 and 54451}}
\author[4]{Adarsh Srinivasan \footnote{Supported by the National Science Foundation under grants CCF-2313372 and CCF-2443697 and a grant from the Simons Foundation, Grant Number 825876, Awardee Thu D. Nguyen. Part of this work was done during a visit to ITU Copenhagen and BARC funded by Basic Algorithms Research Copenhagen(BARC), supported by VILLUM Foundation Grants 16582 and 54451}}
\affil[1]{ KTH Royal Institute of Technology \\ \tt{\{austrin,bercea\}@kth.se}}
\affil[2]{Queens College, City University of New York\\
\tt{mayank.goswami@qc.cuny.edu}}
\affil[3]{IT University of Copenhagen\\
\tt{nuli@itu.dk}}
\affil[4]{Rutgers University\\
\tt{adarsh.srinivasan@rutgers.edu}}
\date{\vspace{-8ex}}

\begin{document}
\maketitle 
\thispagestyle{empty}
\begin{abstract}

Given a $k$-CNF formula and an integer $s \geq 2$, we study algorithms that obtain $s$ solutions to the formula that are as dispersed as possible. For $s=2$, this problem of computing the \emph{diameter} of a $k$-CNF formula was initiated by Creszenzi and Rossi, who showed strong hardness results even for $k=2$. The current best upper bound [Angelsmark and Thapper '04] goes to $4^n$ as $k \rightarrow \infty$. As our first result, we show that this quadratic blow up is not necessary by utilizing the Fast-Fourier transform (FFT) to give a $O^*(2^n)$ time exact algorithm for computing the diameter of any $k$-CNF formula. 

For $s>2$, the problem was raised in the SAT community (Nadel '11) and several heuristics have been proposed for it, but no algorithms with theoretical guarantees are known. We give exact algorithms using FFT and clique-finding that run in $O^*(2^{(s-1)n})$ and $O^*(s^2 |\Om|^{\omega \lceil s/3 \rceil})$ respectively, where $|\Om|$ is the size of the solutions space of the formula $\f$ and $\omega$ is the matrix multiplication exponent.

However, current SAT algorithms for \emph{finding one solution} run in time $O^*(2^{\varepsilon_{k}n})$ for $\varepsilon_{k} \approx 1-\Theta(1/k)$, which is much faster than all above run times. \emph{As our main result}, we analyze two popular SAT algorithms - PPZ (Paturi, Pudl\'ak, Zane '97) and Sch\"{o}ning's ('02) algorithms, and show that in time $\text{poly}(s)O^*(2^{\varepsilon_{k}n})$, they can be used to approximate diameter as well as the dispersion ($s>2$) problem. While we need to modify Sch\"{o}ning's original algorithm for technical reasons, we show that the PPZ algorithm, without any modification, samples solutions in a geometric sense. We believe this geometric sampling property of PPZ may be of independent interest.

Finally, we focus on diverse solutions to NP-complete optimization problems, and give bi-approximations running in time $\text{poly}(s)O^*(2^{\varepsilon n})$ with $\varepsilon<1$ for several problems such as  \textsc{Maximum Independent Set}, \textsc{Minimum Vertex Cover}, \textsc{Minimum Hitting Set}, \textsc{Feedback Vertex Set}, \textsc{Multicut on Trees} and \textsc{Interval Vertex Deletion}. For all of these problems, all existing exact methods for finding optimal diverse solutions have a runtime with at least an exponential dependence on the number of solutions $s$. Our methods show that by relaxing to bi-approximations, this dependence on $s$ can be made polynomial.

\end{abstract}

\pagenumbering{arabic} 

\input{introduction}
\input{exact}
\input{PPZ}
\input{schoning-new}
\input{applications-old}
\bibliographystyle{alpha}
\bibliography{ref.bib}
\newpage
\appendix
\input{app_more-sch}
\input{App_Dispersion}
\input{App_DistinctSum}
\input{App_MinOnes}
\input{App_UniformSampling}

\input{App_schcalc}

\end{document}

%% file: includes.tex
\usepackage[utf8]{inputenc}
\usepackage[USenglish]{babel}
\usepackage{amsmath,amssymb,amsfonts,amsthm}
\usepackage{mathabx}
\usepackage{mathtools}
\usepackage[T1]{fontenc}
\usepackage{url}
\usepackage{fullpage}
\usepackage{hyperref}
\usepackage{cleveref}
\usepackage[usenames,dvipsnames,table]{xcolor}
\usepackage{booktabs}
\usepackage{authblk}
\usepackage{todonotes}
\usepackage{comment}
\usepackage{paralist}
\usepackage{enumerate}
\usepackage{enumitem}
\usepackage{authblk}
\usepackage{scalerel}
\usepackage{thm-restate}
\usepackage[title]{appendix}
\usepackage{soul}
\usepackage{nicefrac}
\usepackage{thm-restate}
\usepackage{listings}
\usepackage[multiple]{footmisc}
\usepackage{bbm}
\usepackage[linesnumbered,boxruled,vlined, boxed]{algorithm2e}
\usepackage{soul}
\usepackage{thm-restate}
\usepackage{xspace}
\usepackage{todonotes}

%% file: defines.tex
\newtheorem{theorem}{Theorem}
\newtheorem{lemma}[theorem]{Lemma}
\newtheorem{claim}[theorem]{Claim}

\theoremstyle{definition}
\newtheorem{definition}{Definition}
\newtheorem{remark}{Remark}
\newtheorem{observation}[theorem]{Observation}

\newtheorem{problem}{Problem}

\newcommand{\set}[1]{\left \{ #1 \right \}}

\newcommand{\size}[1]{\ensuremath{\left|#1\right|}}
\newcommand{\alg}{\mathcal{A}}

\newcommand{\Q}[1]{\{0,1\}^{#1}}
\newcommand{\f}{\mathbf{F}}
\newcommand{\Om}{\Omega_{\f}}
\newcommand{\calf}{\mathcal{F}}

\newcommand{\minones}{\textsc{Min-Ones}\xspace}

\newcommand{\LS}{\textsf{LS}}

\newcommand{\phimax}{\Phi\textsc{-Max}}
\newcommand{\phimin}{\Phi\textsc{-Min}}
\newcommand{\divmin}{\textsc{Diverse-}\Phi\textsc{-Min}}
\newcommand{\divmax}{\textsc{Diverse-}\Phi\textsc{-Max}}
\newcommand{\OPT}{\textsf{OPT}}
\newcommand{\ALS}{\textsf{Anchored-}\LS}
\newcommand{\SW}{\textsf{SW}}

\newcommand{\PD}{\textsc{minPD}}
\newcommand{\w}{\textsc{minw}}
\newcommand{\SPD}{\textsc{sumPD}}
\newcommand{\sumd}{\textsc{sum-}d_H}
\newcommand{\mind}{\textsc{min-}d_H}
\newcommand{\opts}{\textsc{Opt-sum}}
\newcommand{\optm}{\textsc{Opt-min}}

\newcommand{\Sopt}{S_{\text{OPT}}}

\newcommand{\D}{\textsc{Diam}}
\newcommand{\p}{\mathcal{P}}

\newcommand{\np}{\mathcal{NP}}

\newcommand{\PPZMod}{\textsf{PPZ-Modify}}

\newcommand{\adarsh}[1]{}
\newcommand{\nutan}[1]{}
\newcommand{\ioana}[1]{}
\newcommand{\per}[1]{}

\newcommand{\review}[1]{}

\newcommand{\sch}{Sch\"{o}ning\xspace}
\newcommand{\Sch}{Sch\"{o}ning\xspace}

\newcommand{\todoin}[1]{}

%% file: mathshortcut.tex
\newcommand{\IR}{\ensuremath{\mathbb{R}}}

\newcommand{\IN}{\ensuremath{\mathbb{N}}}

\newcommand{\ceil}[1]{{\left\lceil{#1}\right\rceil}}
\newcommand{\floor}[1]{{\left\lfloor{#1}\right\rfloor}}

\newcommand{\poly}{\mbox{\rm poly}}

\newenvironment{tbox}{\begin{tcolorbox}[
		enlarge top by=5pt,
		enlarge bottom by=5pt,
		 breakable,
		 boxsep=0pt,
                  left=4pt,
                  right=4pt,
                  top=10pt,
                  arc=0pt,
                  boxrule=1pt,toprule=1pt,
                  colback=white
                  ]
	}
{\end{tcolorbox}}

\DeclarePairedDelimiter\inner{\langle}{\rangle}

%% file: introduction.tex
\section{Introduction}
\label{sec:intro}
In this work, we start by asking a simple question:  what is the complexity of computing the diameter of a $k$-SAT solution space? That is, given a satisfiable $k$-CNF formula, we want to output two satisfying assignments with maximum Hamming distance between them. More generally, what if we want \emph{multiple} satisfying assignments that are maximally far apart? One can also think of this as finding a binary code with optimal rate/distance tradeoff, where each codeword must satisfy the given $k$-CNF formula. We give exact and approximate exponential time algorithms for these problems and show that existing well-known algorithms for finding one solution can be leveraged to output multiple, reasonably far apart, solutions. 

Crescenzi and Rossi~\cite{crescenzi2002hamming} formulated the diameter computation problem for general Constraint Satisfaction Problems (CSPs), under the name \textsc{Maximum Hamming Distance}. They studied the approximability of the problem and gave a complete classification based on Schaefer’s criteria for the satisfiability of CSPs~\cite{schaefer1978complexity}. In particular, they also showed that the diameter problem is NP-hard even for $2$-SAT.\footnote{They in fact show that it is PolyAPX-hard. Moreover, while not explicitly stated, their reduction immediately gives an optimal inapproximability of $O(n^{1-\epsilon})$ for the diameter of a $2$-CNF formula.} On the constructive side, Angelsmark and Thapper~\cite{angelsmark2004algorithms} gave an algorithm that outputs a diameter pair in polynomial space and $(2a_k)^n$ time, whenever there exists an  $(a_k)^n$ time algorithm for finding one satisfying assignment. Under standard complexity assumptions (SETH), $a_{k} \rightarrow 2$ as $k \rightarrow \infty$, so the above approach is unlikely to run in time better than\footnote{We use the $O^*$ notation to hide polynomial factors in $n$.} $O^*(4^n)$.

This already raises the interesting question of the optimal running time needed for finding a diameter pair (i.e., its exponential complexity~\cite{calabro2009exponential}). In the case of graphs, it is known that quadratic blow-up in time is unavoidable, assuming the Orthogonal Vectors Hypothesis~\cite{VWill2018, alman2015probabilistic}. Should we also expect a quadratic blow-up in time for diameter of $k$-SAT? We first show that this is not the case: using a Fourier analytical approach, we show how to compute a diameter pair deterministically in $O^*(2^n)$ time (Theorem~\ref{thm:exactdiam}). 

\noindent\textbf{Dispersion.} The problem of computing $s>2$ diverse satisfying assignments to a $k$-CNF formula was explicitly raised by Nadel~\cite{nadel2011generating}. Generating diverse solutions has many applications~\cite{baste2019fpt,abboud2022improved,bansal2010approximation}, and several other works have focused on finding multiple solutions to either SAT or constraint programming~ \cite{agbaria2010sat, petit2019enriching, hebrard2005finding, plazar2019uniform, kitchen2007stimulus, gomes2006near, arcuri2011formal}. However, all of the above works are heuristic in nature, and we could not find any algorithm for dispersed solutions to $k$-SAT with provable guarantees. Our work provides the first exact and approximate algorithms for computing diverse solutions to a $k$-CNF formula.

There are many different ways to define the \emph{dispersion} for a set of points (see Table 1 in~\cite{indyk2014composable}). We consider two most popular measures of dispersion: minimum pairwise distance and sum of pairwise distances (the latter is equivalent to average pairwise distance). We will use $d_{H}$ to denote the Hamming distance. By the dispersion problem, we mean given a $k$-CNF formula $\f$ and an integer $s \geq 2$, return a set $S$ of $s$ satisfying assignments to $\f$ that maximize $\PD(S):=\min_{z_1,z_2 \in S} d_{H}(z_1,z_2)$ or $\SPD(S):=\frac{1}{2}\sum_{z_1,z_2 \in S} d_{H}(z_1,z_2)$. If the $k$-CNF formula does not have $s$ distinct satisfying assignments, we allow the algorithm to return a multiset. Unless stated otherwise, our results will be for the minimum version of dispersion.

\vspace{1mm}\noindent\textbf{Exact algorithms.} We show that we can extend our Fourier analytical approach for diameter to dispersion, obtaining an exact algorithm in time $O^*(2^{(s-1)n})$ (Theorem~\ref{thm:exactdisp}). Furthermore, for $s \geq 6$ we also get a faster algorithm based on clique finding (Theorem~\ref{thm:exactdisp faster}), that runs in time $O^*(s^2 |\Om|^{\omega \lceil s/3 \rceil})$, where $\Om$ is the set of satisfying assignments of the formula $\f$ and $\omega\leq 2.38$ is the matrix multiplication exponent \cite{williams2024new}.

\vspace{2mm}\noindent\textbf{Faster approximations.} Even with our improvements, the above exact algorithms still run in $O^*(2^{csn})$ time for $c<1$. What if we allow approximations? Two questions arise:
\begin{itemize}
    \item Can one obtain a bound of the form $f(s)O^*(2^n)$? If so, must $f$ have exponential dependence on $s$, or can $f$ be made polynomial in $s$?
    \item The current fastest $k$-SAT algorithms for finding \emph{one solution} run in time $O^*(2^{\varepsilon_{k}n})$ for $\varepsilon_{k} = 1-\Theta(1/k)$. Can one get a bound of the form $f(s)O^*(2^{\varepsilon_{k}n})$? Thus, the best runtime for finding $s$ dispersed solutions that one could hope for is $\text{poly}(s)O^*(2^{\varepsilon_{k}n})$, as this is roughly the time taken to find any set of $s$ solutions. Can we achieve this?
\end{itemize}

\begin{center}
    \textbf{Main result, informal}
\end{center}
\noindent\fbox{
\parbox{\textwidth}{There exist randomized algorithms with a run time of $\text{poly}(s)O^*(2^{\varepsilon_{k}n})$ that, given a $k$-CNF formula $\f$ on $n$ variables and a parameter $s$, return a set $S$ of $s$ many satisfying assignments that approximately maximize $\PD(S)$ and $\SPD(S)$. Moreover, for several \emph{optimization} problems, there exist algorithms with a similar runtime that are bi-approximations, i.e., return approximately-optimal solutions that are also approximately-maximally-diverse.}}

\vspace{2mm} In addition to these results being a step towards bridging the gap between the theory and practice of finding diverse solutions, what is surprising is that the way we arrive at them reveals novel interesting aspects of two extremely well-studied algorithms for finding one solution to a given formula: PPZ and Sch\"{o}ning's algorithm. 

\medskip\noindent
\textbf{PPZ and Sch\"{o}ning's algorithms.} The complexity of the $k$-SAT problem has a long and rich history~\cite{impagliazzo2001complexity,impagliazzo2001problems,calabro2009exponential,fomin2013exact}. In a foundational work, Paturi, Pudl\'ak, and Zane~\cite{PPZ} presented a remarkably simple and elegant randomized algorithm for $k$-SAT. Their algorithm runs in time $O^*\!\left(2^{ (1-1/k)n}\right)$ and outputs a satisfying solution with probability $1-o(1)$ if one exists. A few years after that, \sch~\cite{schoning2002probabilistic} developed another surprisingly simple random walk-based algorithm running in time $O^*\!\left(2^{ (1-1/(k \ln 2))n}\right)$,\footnote{The run-time of \sch's algorithm is normally presented as $O^*\left((2(1-1/k))^n\right)$, which we have rewritten for ease of comparison with PPZ.} which runs faster than the PPZ algorithm for all $k$. 
With time, these approaches have been reanalyzed and sometimes improved in a variety of technically subtle and involved ways~\cite{hofmeister2002probabilistic, baumer2004improving, paturi2005,hertli2010improving,hertli2014breaking,liu2018chain,SchederSteinberger,hansen2019faster, PPZmoreisbetter,scheder2022ppsz}, including the PPSZ algorithm by Paturi, Pudl\'ak, Saks and Zane~\cite{paturi2005}, which is the current fastest algorithm for $k$-SAT. 

In our work, we ask whether PPZ and \sch's can exploit the global geometry of the solution space and go beyond finding just one satisfying assignment. Namely, can they be used to \emph{approximate} the diameter and the dispersion for $k$-SAT? We remark that the main result above is not a black-box result that uses \emph{any SAT solver} - we only know how to use PPZ and Sch\"{o}ning's algorithms for this purpose. To familiarize the reader with these two algorithms, we provide their pseudocodes next. 

\vspace{2mm}\begin{minipage}{0.46\textwidth}
  \begin{algorithm}[H]
  \SetKwFor{RepTimes}{repeat}{times}{end}
    \SetAlFnt{\small}
    \DontPrintSemicolon
    \KwIn{A $k$-CNF formula $\f$ over $n$ variables} 
      \RepTimes{$n^{O(1)} \cdot 2^{(1-1/k)n}$} {
      Sample $\pi \sim S_n$, $y \sim \{0,1\}^n$ u.a.r.\\
      \For{$i \in [n]$}
      {\If{$\f$ contains the unit clause $(x_{\pi(i)})$}{$u_{\pi(i)} \gets 1$}
      \If{$\f$ contains the unit clause $(\Bar{x}_{\pi(i)})$}{$u_{\pi(i)} \gets 0$}
      \Else{$u_{\pi(i)} \gets y_{\pi(i)}$}
      
      $\f \gets \f|_{x_{\pi(i)}=u_{\pi(i)}}$}
      \If{$u$ satisfies $\f$}{\Return{$u=(u_1, u_2, \dots, u_n)$}}
      }
      \Return{``not satisfiable''}
    \caption{\scshape PPZ}
  \end{algorithm}
\end{minipage}
   \hfill
\begin{minipage}{0.46\textwidth}
  \begin{algorithm}[H]
  \SetAlFnt{\small}
  \DontPrintSemicolon
  \SetKwFor{RepTimes}{repeat}{times}{end}
  \KwIn{A $k$-CNF formula $\f$ over $n$ variables} 
    \RepTimes{$n^{O(1)} \cdot 2^{(1-\frac{1}{k \ln 2})n}$} 
    {Sample $y \sim \{0,1\}^n$\\
    \RepTimes{$3n$}{
    
    \uIf{$y$ satisfies $\f$}{\Return{$y$}}
    \Else
    {
    Let $C$ be the first clause in $\f$ not satisfied by $y$, pick one of the $k$ variables in $C$ at random and flip the value that $y$ assigns to that variable }
    }
    }
    \Return{``not satisfiable''}
  \caption{\scshape \sch}
  \end{algorithm}
\end{minipage}

\medskip\noindent\textbf{Farthest Point Oracles} Gonzalez~\cite{gonzalez1985clustering} proposed the farthest-insertion algorithm, and showed that it gives a  1/2 approximation to the minimum version of the dispersion problem: given a metric space of $n$ points, find a set $S$ of $s$ points in it that maximize $\PD(S)$. This was later extended to the sum version by \cite{borodin2012max}. The algorithm builds the set $S$ iteratively; in the $i$th iteration it adds the point $x_i$ that maximizes the minimum (resp. sum of) distance to all the points in the solution so far. Moreover, the factor 1/2 is tight assuming the Exponential Time Hypothesis (ETH), so in a sense, farthest insertion is the best possible (polynomial) algorithm for dispersion~\cite{gao2022obtaining}. 

In our setting, a farthest point oracle takes as input a $k$-CNF formula $\f$ (with a set $\Om$ of satisfying assignments) and a set (or multiset) $S \subseteq \Q{n}$, and outputs a satisfying assignment $z^* \in \Om$ that is ``far away'' from the assignments in $S$. Namely, for $x \in \Q{n}, S \subseteq \Q{n}$, we let $ \mind(S,x)= \min_{y \in S} d_H(x,y)$ and $\sumd(S,x)=\sum_{y \in S} d_H(x,y)$. Then for some $\delta \in [0,1)$, the assignment $z^*$ would either satisfy $$
\mind(z^*, S) \geq (1-\delta) \max_{z' \in \Om} \mind(z', S), \textbf{ or  } \sumd(z^*, S) \geq (1-\delta) \max_{z' \in \Om} \sumd(z', S),$$ for the $\PD(S)$ and the $\SPD(S)$ version, respectively. 

In \Cref{sec:technicallemma}, we describe our main technical lemmas on PPZ and \sch algorithms. This is followed by the algorithms for diameter and dispersion implied by these lemmas (\Cref{intro_results}). As mentioned in the informal result statement, our techniques extend to finding diverse solutions to optimization problems as well. These results are formally described in \Cref{intro_generalize}.

\subsection{Main Technical Lemmata}
\label{sec:technicallemma}

Recall that we are aiming for a runtime of  $\text{poly}(s)O^*(2^{\varepsilon_{k}n})$. The question therefore is: can we implement farthest point insertion in $O^*(2^{\varepsilon_{k}n})$ time? We now state the two main technical lemmas that form the core of our analysis. 

\vspace{2mm}
\noindent\fbox{
\parbox{\textwidth}{

\begin{lemma}[PPZ performs geometric sampling]\label{intro_ppz_brief}

     For any $z_0\in \{0,1\}^n$, with probability at least $\frac{1}{2n}\cdot 2^{-(1-1/k)n}$, each iteration of the PPZ algorithm outputs a satisfying assignment $z^*$, such that $d_H(z_0,z^*) \geq \left(1- \frac{1}{k}\right) \cdot \max_{z' \in \Om} d_H(z_0,z')$. The iteration of PPZ does not depend on $z_0$.

\end{lemma}

\begin{lemma}[Modified Sch\"{o}ning's Algorithm is a farthest point oracle]\label{intro_schoning_brief}
    There exists an algorithm, running in time $O^*\!\left(2^{ (1-1/(k \ln 2))n}\right)$ that takes a $k$-CNF formula $\f$ and $z_0 \in \Q{n}$ as input and outputs a satisfying assignment $z^*$ such that $d_H(z_0, z^*) \geq \left(1-\frac{4(k-1)}{(k-2)^2}\right) \cdot \max_{z' \in \Om} \sumd(S,z')$. Here, $z_0$ is used explicitly inside the iteration.
\end{lemma}
}}

\vspace{1mm}We sketch the proofs in Section~\ref{sec:techniques}. Three remarks are in order.

\begin{remark}
Lemma~\ref{intro_ppz_brief} requires several insights into the behavior of PPZ. PPZ is not a traditional local search algorithm and it falls in the random restriction paradigm~\cite{scheder2022ppsz}. The analysis of PPZ~\cite{PPZ} is local in nature: the authors bound the probability of arriving at a solution $z$ that is $j$-isolated, meaning that exactly $n-j$ neighbors of $z$ are also satisfying solution. This probability is then added over all satisfying assignments, resulting in the PPZ run time bound of $O^*(2^{(1-1/k)n})$. On the other hand, in Lemma~\ref{intro_ppz_brief} we are interested in bounding the probability that PPZ returns a solution that is far away from a given point $z_0$. The fact that PPZ, without any modifications based on $z_0$, returns such far-away solutions automatically was surprising to us.  We leave it as an open question whether the PPZ-based, more involved, state-of-the-art algorithm of Paturi, Pudl\'ak, Saks and Zane (PPSZ) \cite{paturi2005}, can also be shown to exhibit similar behavior.
\end{remark}

\begin{remark}
Unlike PPZ, we could not prove that Sch\"{o}ning's original algorithm works directly as an approximate farthest point oracle. Our modification of Sch\"{o}ning's algorithm controls both the region of starting assignments $x$ and the length of the Sch\"{o}ning walk from $x$. Instead of Sch\"{o}ning's analysis that bounds the probability of finding any solution starting at a random point, we bound the probability that we find a solution far from $z_0$ and close to $x$. As a plus, in addition to giving us a farthest point oracle, this also allows us to obtain a tradeoff between runtime and approximation factors. More details can be found in Section~\ref{sec:techniques}.
\end{remark}

\begin{remark} 
We investigate other promising candidate approaches for $k$-CNF dispersion that do not use PPZ or \Sch's algorithms.  
First, we show that the approach to solve dispersion problem via  \emph{uniform sampling algorithms}~\cite{schmitt2013exploiting} does not necessarily give a good approximation compared to our approach, even for the diameter (Appendix~\ref{app:uniformv2}).  Furthermore, we consider yet another promising approach via the \minones\ problem.  This problem asks for the minimum Hamming weight solution to a SAT formula~\cite{ConicSearch}. While we note that the an algorithm for the \minones\ problem can be used to give a $1/2$ approximation of the diameter(Appendix~\ref{app:minones}), we also observe that this approach is unlikely to be extended to finding more than two diverse solutions, as the reduction to diameter does not generalize.
\end{remark}

Lemma~\ref{intro_ppz_brief} and Lemma~\ref{intro_schoning_brief} give us algorithms for computing a set $S$ with maximum dispersion for both the $\PD(S)$ and the $\SPD(S)$ versions. These are stated formally in Section~\ref{intro_results}. Moreover, we get a variety of applications: diverse solutions to several optimization problems and CSPs, and reanalyzing SAT algorithms when the formula has many diverse assignments. These are presented in Section~\ref{intro_generalize}.

\subsection{Results on Diameter and Dispersion}\label{intro_results}

Throughout the paper, we let  $\f$ denote a $k$-CNF formula on $n$ variables (unless otherwise specified). Given such an $\f$, we let $\Omega_{\f} \subseteq \{0,1\}^n$ denote the set of satisfying assignments of $\f$. We start by formally defining the diameter problem. For a given formula $\f$, let $\D(\f)$ be defined as $\max_{z_1, z_2 \in \Om} \left\{d_H(z_1, z_2)\right\}$, where $\Omega_\f$ is non-empty. Note that when $\f$ has a unique satisfying assignment, then $\D(\f)$ is simply $0$. On the other hand, if $\f$ is not satisfiable, we define $\D(\f) = \perp$. For a set $S\subseteq \{0,1\}^n$, define  $\PD(S):=\min_{z_1,z_2 \in S} d_{H}(z_1,z_2)$ and $\SPD(S):=\frac{1}{2}\sum_{z_1,z_2 \in S} d_{H}(z_1,z_2)$. We then define \opts$(\f,s)$ as the maximum value of $\SPD(S)$ over all multisets $S$ with $s$ satisfying assignments (including multiplicities), and \\$\optm(\f,s) = \max_{S \subseteq \Om, |S|=s} \PD(S)$, i.e., the maximum such distance over all sets of $s$ satisfying assignments. Further, we define \opts$_\neq(\f,s)$ as the maximum value of $\SPD(S)$ over all \emph{sets} $S$ with $s$ distinct satisfying assignments. 
\subsubsection{Computing diameter exactly and approximately}
\label{sec:intro-dia}
\medskip\noindent
\textbf{Computing diameter exactly.} We first study the exponential complexity of computing $\D(\f)$. Specifically, we prove the following theorem. 
\begin{restatable}{theorem}{exactdiam}[Exact Diameter] \label{thm:exactdiam}

 Let $\f$ be a $k$-CNF formula on $n$ variables. 
 There exists a deterministic algorithm that uses $O^*(2^n)$ time and $O^*(2^n)$ space, and outputs a pair of satisfying assignments $z_1, z_2 \in \Omega_\f$ with $d_H(z_1,z_2)= \D(\f)$. 
\end{restatable}
Prior to our work, the best exact algorithm known
was by Angelsmark and Thapper~\cite{AngelsmarkThapperCSP}. Their algorithm runs in time $O((2a_k)^n)$ and space $\poly(n)$, where $O(a_k^n)$ is the running time for solving the $k$-SAT problem. Our result significantly improves the running time of their algorithm (but uses substantially more space than their algorithm). 

Our technique is also different from other techniques in the literature. Namely, this algorithm does not depend on any SAT algorithm. Our main observation is that $\D(\f)$ can be reduced to computing the \emph{convolution} of the Boolean function represented by $\f$ with itself. We then use that such a convolution can be computed within the above stated time and space bounds using the Fast Fourier Transform. 

Our technique for exact diameter is fairly general and does not depend on the fact that the solution space corresponds to a $k$-CNF formula. For any Boolean function $f: \{0,1\}^n \rightarrow \{0,1\}$ such that for a given $x \in \{0,1\}^n$, there is a polynomial time oracle to compute $f(x)$, our algorithm can be used to exactly compute the diameter of $f$ with the above performance guarantees. 

\medskip\noindent
\textbf{Approximating the diameter.} Next, we give algorithms for approximating $\D(\f)$\footnote{All approximation algorithms we present here use $\poly(n)$ space.}. As a warm-up, here is a simple way to approximate $\D(\f)$. We can start by using the best known algorithm to find a single satisfying assignment for $\f$. Suppose that assignment is $\alpha$. We can then (in polynomial time) change $\f$ to $\f'_\alpha$ by negating some of the variables such that $1^n$ becomes the satisfying assignment of $\f'_\alpha$. 
One can then use the best known algorithm for the \minones problem to find a satisfying assignment for $\f'_\alpha$, which finds a satisfying assignment with minimum $1$s in it, say $\beta$. It is easy to see that the Hamming distance between $\alpha, \beta$ gives a $0.5$-approximation to the diameter of $\f$. For more details on this reduction, we refer the reader to~\Cref{app:minones}. By using the best known algorithms for $k$-SAT~by Paturi, Pudl\'ak, Saks, and Zane~\cite{paturi2005} and for \minones~by Fomin, Gaspers, Lokshtanov and Saurabh~\cite{ConicSearch}, it is easy to see that we can obtain $(\alpha, \beta)$ in time $O^*((2-\frac{1}{k})^{n})=O^*\left(2^{ (1-\frac{1}{(2 \ln 2) \cdot k})n}\right)$\footnote{Note that, $ O^*((2-\frac{1}{k})^{n}) = O^*(2(1-\frac{1}{2k}))^{n} \sim O^*(2^n \cdot e^{-\frac{n}{2k}}) = 2^{n (1-\frac{1}{(2 \ln 2) \cdot k})}$.}.

Here, we obtain better running time for $\D(\f)$ for $k\geq 3$ with a small loss in the approximation factor. From here on, we assume that $k \geq 3$ unless stated otherwise.

\begin{restatable}[PPZ approximating $\D(\f)$]{theorem}{ppzdiam}
    \label{thm:ppz-for-dia}
    Let $\f$ be a $k$-CNF formula on $n$ variables. There exists a randomized algorithm running in time $O^*\left(2^{(1-1/k)n}\right)$ that takes $\f$ as input and if $\f$ is satisfiable, outputs $z_1^*, z_2^* \in \Om$ with $d_H(z_1^*, z_2^*) \geq \frac{1}{2}\cdot\left(1-\frac{1}{k}\right) \D(\f)$ with probability $1-o(1)$.
\end{restatable}

The running time of the algorithm here is exactly the same as the running time of the algorithm achieved in~\cite{PPZ}, which solves the $k$-SAT problem. Our result demonstrates that the diameter can be approximated in the same time used to compute a single satisfying assignment. In fact, the way we achieve this running time is by repeatedly invoking the PPZ algorithm. At the heart of the analysis of the PPZ algorithm lies the Satisfiability Coding Lemma from~\cite{PPZ}. Informally speaking, the Satisfiability Coding Lemma says that if the solutions of a $k$-CNF instance are \emph{well-separated} then they have a small description. In our proof, we generalise this lemma. We discuss our proof idea in detail in \Cref{sec:techniques}. 

Next, we show how to approximate the diameter within the running time guarantees of \sch's algorithm for $k$-SAT. Specifically, we prove the following theorem. 
\begin{theorem}[\sch approximating $\D(\f)$.]
    \label{thm:sch-for-dia-fixedapprox}
    Let $\f$ be a $k$-CNF formula on $n$ variables. There exists a randomized algorithm running in time $O^*\!\left(2^{ (1-\frac{1}{k \ln 2})n}\right)$ that takes $\f$ as input and if $\f$ is satisfiable, outputs $z_1^*, z_2^* \in \Om$ with $d_H(z_1^*, z_2^*) \geq \frac{1}{2}\left(1-\frac{4(k-1)}{(k-2)^2}\right) \cdot \D(\f)$ with probability $1-o(1)$.
\end{theorem}

In fact, \Cref{thm:sch-for-dia-fixedapprox} is one instance of a smooth tradeoff between the approximation factor and the running time. We present the full tradeoff in \Cref{thm:sch-for-dia} \Cref{sec:sch}.
Notice that the running time obtained here is better than the running time obtained using \Cref{thm:ppz-for-dia}, which in turn is faster than the naive algorithm that uses \minones. We incur some loss in the approximation factors to obtain these speedups. As stated, the result gives non-trivial approximation factors when $k\geq 7$. \Cref{thm:sch-for-dia} generalizes \Cref{thm:sch-for-dia-fixedapprox} to get non-trivial approximation factors for any $k$. In \Cref{thm:sch-for-dia-two} \Cref{app:more-sch}, we present another \sch-type algorithm to approximate the diameter that outperforms the algorithm in \Cref{thm:sch-for-dia-two} for small values of $k$ and some regimes of the approximation factor. 

\subsubsection{Computing dispersion exactly and approximately}

We extend all the algorithms from Section~\ref{sec:intro-dia} and obtain bounds for the dispersion problem.

\medskip\noindent
\textbf{Exact algorithms for dispersion.}
We start with the problem of exactly computing $\opts(\f,s)$, $\optm(\f,s)$ and $\opts_{\neq}(\f,s)$. The obvious algorithm for computing all these quantities would be to do a brute force search over all $z_1, z_2, \dots, z_s \in \set{0,1}^n$, which would require $O^*(2^{sn})$ time. We observe that we can extend the Fourier analytical approach we used in \Cref{thm:exactdiam} to do this in $O^*(2^{(s-1)n})$ time and $O^*(2^n)$ space.  We state and prove the formal statement in \Cref{sec:exact}. We also provide an alternate  algorithm for dispersion in \Cref{thm:exactdisp faster} in~\Cref{sec:exactdisp faster}. The algorithm, based on clique-finding, runs in time $O(s^2 \cdot \size{\Om}^{\omega \lceil s/3 \rceil})$ and uses space  $O(\size{\Om}^{2 \lceil s/3 \rceil})$, where $\omega \le 2.38$ denotes the matrix multiplication exponent~\cite{williams2024new}. As such, it is faster than the Fourier analysis-based algorithm for any $s\geq 6$, and can be much faster when the size of the solution set is less than $2^{n}$.

\medskip\noindent
\paragraph{Approximating dispersion.} We now turn to approximation algorithms for dispersion. Our goal is to come up with approximation algorithms for all the versions of the dispersion problem as in the case of approximation algorithms for computing the diameter. We saw that \minones can be used to give a 0.5 approximation to $\D(\f)$. However, it is not clear how we can use it to approximate the dispersion. More about this in \Cref{sec:techniques}. 

\paragraph{Approximating $\opts(\f, s)$. } We show that PPZ as well as \sch's algorithms can be modified to compute $\opts(\f,s)$. Formally, 
\begin{theorem}[PPZ approximating $\opts(\f,s)$] \label{thm:ppz-for-sumdisp-easy}
    Let $\f$ be a $k$-CNF formula on $n$ variables. There exists a randomized algorithm running in time $O^*\left(s^4 \cdot 2^{(1-1/k)n}\right)$ that takes $\f$ and an integer $s \geq 1$ as input and if $\f$ is satisfiable, with probability at least $1-o(1)$, outputs a multiset $S \subseteq \Om$ of size $s$ such that $$\SPD(S) \geq \left(1-\frac{4}{k-3}\right)\left(1-\frac{2}{s+2}\right)\cdot \opts(\f,s) \; .$$
\end{theorem}
\begin{remark}
    When $k \leq 6$, this algorithm achieves a better approximation ratios for smaller values of $s$ than stated above. Note that as $k$ and $s$ become large, the approximation factor tends to $1$. For more details, we refer to the reader to the full version of this theorem (\Cref{thm:ppz-for-sumdisp}) in \Cref{sec:PPZ}). 
\end{remark}

For $\opts_\neq(\f,s)$, we can obtain exactly the same approximation factors as in \Cref{thm:ppz-for-sumdisp} for certain parameter regimes of $s$ (see \Cref{app:distinctsum} for more details). 

\medskip\noindent
\textbf{Approximating $\optm(\f,s)$.} Next, we show that our techniques can be used to approximate \optm\ as well. Formally,

\begin{restatable}[PPZ approximating $\optm(\f,s)$]{theorem}{ppzmindisp} \label{thm:ppz-for-mindisp}
    Let $\f$ be a $k$-CNF formula on $n$ variables. There exists a randomized algorithm running in time $O^*\left(s^3 \cdot 2^{(1-1/k)n}\right)$ that takes $\f$ and an integer $s \geq 1$ as input and if $\f$ is satisfiable and $|\Om| \geq s$, with probability at least $1-o(1)$, outputs a set $S$ of size $s$ such that $\PD(S) \geq \frac{1}{2} \left(1-\frac{1}{kH^{-1}(1-1/k)}\right) \cdot \optm(\f,s)$ \footnote{The function $H^{-1}(\cdot)$ denotes the inverse of the binary entropy function $H(x)= -x \log(x)-(1-x) \log(1-x) $ restricted to the domain $[0,1/2]$. The domain of $H^{-1}$ is $[0,1]$ and its range is $[0,1/2]$.} 
    
\end{restatable}
Note that, in the above statement, the approximation factor is non-trivial ($>0$) only for $k \geq 5$. We note that we can also obtain \sch-type running time bounds for dispersion for $k \geq 2$. We achieve this by extending \Cref{thm:sch-for-dia-fixedapprox}. The statements of our results and their proofs appear in \Cref{sec:sch}. 

\medskip\noindent
\textbf{Approximating $\optm(\f,s)$  for heavy-weight solutions.} We now consider a heavy-weight variant of  $\optm(\f,s)$. Formally, for a $k$-CNF formula $\f$, we let $\Omega_{\f,\geq W}$ denote the set of satisfying assignments to $\f$ with Hamming weight at least $W$. We then define $$\optm(\f,s,\geq W) = \max_{\substack{S \subseteq \Omega_{\f,\geq W}\\|S|=s}} \PD(S), \optm(\f,s,\leq W) = \max_{\substack{S \subseteq \Omega_{\f,\leq W}\\|S|=s}} \PD(S) \; .$$
and let $\w(S)$ denote the minimum Hamming weight of assignments in $S$. We show that the approach developed for approximating \optm\ via \sch's algorithm can also be used to return dispersed satisfying assignments of heavy weight. 

\begin{theorem}[\sch for weighted dispersion] \label{thm:schheavyeasy}
    Let $\f$ be a $k$-CNF formula on $n$ variables, $W \in [n]$ and $s \in \IN$. Let $\delta = \frac{4(k-1)}{(k-2)^2}$. There exist algorithms that take $\f,s,W$ as input and output with probability $1-o(1)$ in time $O^*\left(s^3 \cdot 2^{n(1-\frac{1}{k \ln 2})}\right)$:
    \begin{enumerate}
        \item $S^* \subseteq \Omega_{\f, \geq (1-\delta) W}$ of size $s$ such that 
        $\PD(S^*) \geq \frac{1}{2}\left(1-\delta \right) \optm(\f,s, \geq W)$ if $\f$ is satisfiable and $|\Omega_{\f, \geq W}| \geq s$. 
        \item  $S^* \subseteq \Omega_{\f, \leq (1+\delta) W}$ of size $s$ such that 
        $\PD(S^*) \geq \frac{1}{2}\left(1-\delta \right) \optm(\f,s, \leq W)$ if $\f$ is satisfiable and $|\Omega_{\f, \leq W}| \geq s$,

    \end{enumerate}
\end{theorem}
\begin{remark}
    We note that when $W=0$, this just reduces to an algorithm for approximating $\optm(\f,s)$. The approximation factors in \Cref{thm:schheavyeasy} are non-trivial only for $k \geq 7$. However, just like the case of \Cref{thm:sch-for-dia-fixedapprox}, \Cref{thm:schheavyeasy} can be generalized, obtaining running time bounds for any $k$ and for a larger range of approximation factors (\Cref{thm:sch-heavy-full}). Further, we can prove that an analogous result exists for the sum of distances dispersion measure. We refer the reader to \Cref{sec:sch} for the complete theorem statements and proofs. 
\end{remark}

\subsection{Generalizations and applications.} \label{intro_generalize}
\medskip\noindent
\textbf{1. Isometric Reductions. }Dispersion has also been studied when the space is induced by solutions to some NP-complete optimization problem~\cite{baste2019fpt,baste2022diversity}. To address this optimization aspect, we first generalize our techniques to give dispersed solutions of high (or low) Hamming weight\footnote{In a recent work, Gurumukhani, Paturi, Pudl\'{a}k, Saks, and Talebanfard~\cite{gurumukhani2024local} consider the problem of enumerating satisfying assignments with Hamming weight at least $W$ for a given $k$-CNF formula (assuming that satisfying assignments of smaller weight do not exist). They show that this problem has interesting connections to circuit lower bounds.}.
Namely, given $W \in [n]$, all of our solutions will have Hamming weight at least (or at most) approximately $W$, and their dispersion will be close to that of an optimally dispersed set wherein all solutions have weight at least (or at most) $W$. We then formalize a set of reductions, that preserve the size of the solution set and the distances between solutions. We call such reductions \emph{isometric}. As a result, we can approximate dispersion for problems such as \textsc{Maximum Independent Set}, \textsc{Minimum Vertex Cover} and \textsc{Minimum Hitting Set}. 

\medskip\noindent
\textbf{2. Using the monotone local search framework for diverse solutions.} Our second application allows us to compute diverse solutions to optimization problems that perhaps do not allow isometric reductions to SAT. In this case, we show how to use the \emph{monotone local search} framework by Fomin, Gaspers, Lokshtanov and Saurabh~\cite{ConicSearch}. This allows us to extend our results to a variety of problems, including \textsc{Feedback Vertex Set}, \textsc{Multicut on Trees}, and \textsc{Minimum $d$-Hitting Set} (see \Cref{table: table} for a sample of the results that can be obtained using this technique\footnote{The table provides the running time guarantees to obtain $3/2$-approx. optimal, $1/4$-approx. maximally diverse solutions, by plugging in $\delta=1/2$ into the run-time bounds in \Cref{thm:PLFS}}).

For all of these problems, any existing exact methods for finding a set of optimal, maximally diverse solutions has a runtime with at least an exponential dependence on the number of solutions $s$ \cite{baste2019fpt,baste2022diversity}. Our methods show that by relaxing to bi-approximations, this dependence on $s$ can be made polynomial.  

\begin{table}[h]
\centering
\begin{tabular}{|l|l|l|l|}
\hline
Optimization Problem  & One optimal solution & Multiple approximately optimal, \\
 & \cite{ConicSearch}  & approximately dispersed solutions \\
\hline
\textsc{$d$-Hitting Set $(d \geq 3)$} & $(2 - \frac{1}{d})^n$ & \Cref{thm:isometricreduction} \\
\textsc{Vertex cover } & $1.5^n$ & $s^3 \cdot 1.5486^n$ \\
\textsc{Maximum independent Set}  & $1.5^n$ & $s^3 \cdot1.5486^n$ \\
\hline
\textsc{Feedback Vertex Set}  & $1.7217^n$ &  $s^3 \cdot1.6420^n$\\
\textsc{Subset Feedback Vertex Set}  & $1.7500^n$ & $s^3 \cdot1.6598^n$ \\
\textsc{Feedback Vertex Set in Tournaments}  & $1.3820^n$ &  $s^3 \cdot1.5162^n$\\
\textsc{Group Feedback Vertex Set} &  $1.7500^n$ &  $s^3 \cdot1.6598^n$\\
\textsc{Node Unique Label Cover} & $(2 - \frac{1}{|\Sigma|^2})^n$ & \Cref{thm:PLFS} \\
\textsc{Vertex $(r,\ell)$-Partization $(r,\ell \leq 2)$}  & $1.6984^n$ &  $s^3 \cdot1.6289^n$ \\
\textsc{Interval Vertex Deletion}  & $1.8750^n$ &  $s^3 \cdot1.7789^n$\\
\textsc{Proper Interval Vertex Deletion} & $1.8334^n$ & $s^3 \cdot1.7284^n$ \\
\textsc{Block Graph Vertex Deletion} &  $1.7500^n$ & $s^3 \cdot1.6598^n$ \\
\textsc{Cluster Vertex Deletion}  & $1.4765^n$ & $s^3 \cdot1.5415^n$ \\
\textsc{Thread Graph Vertex Deletion}  & $1.8750^n$ &  $s^3 \cdot1.7789^n$\\
\textsc{Multicut on Trees}  & $1.3565^n$ & $s^3 \cdot1.51^n$ \\
\textsc{3-Hitting Set}  & $1.5182^n$ & $s^3 \cdot1.5544^n$ \\
\textsc{4-Hitting Set}  & $1.6750^n$ & $s^3 \cdot1.6167^n$ \\
\textsc{$d$-Hitting Set $(d \geq 3)$} & $(2 - \frac{1}{d-0.9245})^n$ & \Cref{thm:PLFS} \\
\textsc{Min-Ones 3-SAT} & $s^3 \cdot1.6097^n$ & \Cref{{thm:sch-heavy-full}} \\
\textsc{Min-Ones $d$-SAT $(d \geq 4)$}  & $(2 - \frac{1}{d})^n$ & \Cref{{thm:sch-heavy-full}} \\
\textsc{Weighted $d$-SAT $(d \geq 3)$}& $(2 - \frac{1}{d})^n$ & \Cref{{thm:sch-heavy-full}} \\
\textsc{Weighted Feedback Vertex Set}  & $1.7237^n$ & $s^3 \cdot1.6432^n$ \\
\textsc{Weighted 3-Hitting Set}  & $1.5388^n$ &  $s^3 \cdot1.5612^n$\\
\textsc{Weighted $d$-Hitting Set $(d \geq 4)$}  & $(2 - \frac{1}{d-0.832})^n$ & \Cref{thm:PLFS} \\ \hline
\end{tabular}
\caption{The second column contains the time taken to obtain one exact solution using methods in~\cite{ConicSearch}. The third column contains the time taken to obtain $3/2$-approx. optimal, $1/4$-approx. maximally diverse solutions (except for Maximum Independent Set, where we obtain $(1/2,1/4)$-bi-approx.)  
}

\label{table: table}
\end{table}

\medskip\noindent
\textbf{3. On faster SAT algorithms.} Another compelling reason to study diversity of the solution space of a $k$-CNF formula is that the existence of far apart solutions might be used to study the computational complexity of $k$-SAT and its variants. Indeed, the geometry of the  solution space has been studied extensively, both to obtain faster SAT solvers (parameterised by the number of solutions, such as in  Hirsch~\cite{hirsch1998} and Kane and Watanabe~\cite{kane2016short}) and in the random SAT setting, e.g., the diameter by Feige, Flaxman and Vilenchik~\cite{feige2011diameter} and the giant connected component by Chen, Mani, Moitra~\cite{moitra}). 

Consider a formula $\f$ with $|\Om|=2^{\delta n}$ for some $\delta>0$. For such a formula, it is known that PPZ scales optimally, i.e., it finds one solution in time $2^{(1-1/k)(1-\delta)n}$~\cite{calabro2008complexity}. Cardinal, Nummenpalo and Welzl~\cite{cardinal2017solving} proved a weaker result for Sch\"{o}ning, but nevertheless, both PPZ and Sch\"{o}ning run faster if the solution space is large. In fact, the same is true for PPSZ~\cite{PPZmoreisbetter}.

Taking this idea a step further, we investigate the runtime of PPZ and Sch\"{o}ning's algorithms when $\Om$ contains many well-dispersed solutions. For example, if $\Om$ contains a Hamming code that achieves the Gilbert Varshmov bound, we can show an exponential improvement in the runtime of Sch\"{o}ning's algorithm (\Cref{sec:fast}). Similarly, using the geometric sampling property of PPZ in Lemma~\ref{intro_ppz_brief}, we obtain an improved runtime in this setting. In this sense, \emph{if having more (solutions) is better~\cite{PPZmoreisbetter}, then our results formalize the intuition that more dispersed solutions are even better}.

\medskip\noindent
\textbf{4. Relation to coding theory.} We mention a connection that might be of independent interest. The dispersion problem can be restated in the language of coding theory, namely, we are looking for codewords that also satisfy a given $k$-CNF formula. If $\f(x) = 1$ for all $x \in \{0,1\}^n$, then it is known that a uniformly random code achieves the Gilbert-Varshamov bound~\cite{Roth_2006}. When $\f$ is not trivial, the algorithms presented in this work provide such a code. Moreover, our result says that the code can be found in time proportional to the running times of PPZ and \Sch (when the size of the code is small).  Additionally, in practice, one also wants codes that have succinct representations, e.g. linear codes~\cite{guruswami2010list,grigorescu2012succinct}. While our codes do not exhibit this property, it would indeed be interesting to extend our algorithms in this direction.

\medskip\noindent
\textbf{5. CSPs.} Finally, since \sch's algorithm for finding one solution generalizes to CSPs, we also give algorithms obtaining diverse solutions to CSPs (Section~\ref{sec:csp}).

\input{technical}

\subsection{Organization of the paper.}
In \Cref{sec:exact}, we present and analyse our algorithms for exact diameter and dispersion (\Cref{thm:exactdiam}, \Cref{thm:exactdisp}, and \Cref{thm:exactdisp faster}). In \Cref{sec:PPZ}, we present our PPZ-based algorithms for approximately computing diameter and dispersion (\Cref{thm:ppz-for-dia}, \Cref{thm:ppz-for-sumdisp} and \Cref{thm:ppz-for-mindisp}). In \Cref{sec:sch}, we present our \sch-based algorithms for diameter, dispersion and weighted dispersion (\Cref{thm:sch-for-dia}, \Cref{thm:sch-for-sumdisp}, \Cref{thm:sch-heavy-full}). In \Cref{sec:applications}, we present our results on diversity preserving reductions and applications of parameterized local feasibility search and prove the results presented in Table 2. 

%% file: technical.tex
\subsection{Technical Overview: Proof sketches for \Cref{intro_ppz_brief} and \Cref{intro_schoning_brief}} \label{sec:techniques}

In this section we outlines the main techniques behind \Cref{intro_ppz_brief} and~\Cref{intro_schoning_brief}, that show that PPZ and Sch\"{o}ning algorithms can be employed as approximate farthest point oracles. Because of this approximation, slightly more work needs to be done in order to bound the overall approximation factors for dispersion. We include the technical details for this part of our analysis in~\Cref{sec:approxfarthest}. There, we also show how to adapt Cevallos, Eisenbrand, and Zenklusen's local search algorithm~\cite{cevallos2019improved} for our setting.

\paragraph{Lemma~\ref{intro_ppz_brief}: PPZ samples geometrically} The PPZ algorithm consists of repeating the following procedure $O^*(2^{(1-1/k)n})$ times: sample an assignment $y \in \{0,1\}^n $ and a permutation $\pi \in S_n$ uniformly and independently at random. Then call a deterministic subroutine $\PPZMod(\f,y,\pi)$ that runs in $n^{O(1)}$ time and outputs another assignment $u$. The algorithm stops once $u \in \Om$.

The analysis is based on bounding the probability that, for a randomly chosen $y$ and $\pi$,  $\PPZMod(\f,y,\pi)$ leads to some satisfying assignment $z \in \Om$. For any $z \in \Om$, let $\tau(\f,z)$ denote the probability that an iteration outputs $z$ and for any set $A \subseteq \Om$, let $\tau(\f,A)=\sum_{z' \in A} \tau(\f, z')$ denote the probability that an iteration outputs a satisfying assignment in $A$. 

The lower bound that PPZ gives on  $\tau(\f,z)$ uses the the \emph{local} geometry of $\Om$ around $z$ in the following sense: we say that $z$ is $j$\emph{-isolated} if, out of the $n$ neighboring assignments to $z$ in the Boolean hypercube, at least $j$ of them are not satisfying. The key observation in the analysis of the PPZ algorithm, called the \emph{Satisfiability Coding Lemma}~\cite{PPZ} states that for every $j$-isolated satisfying assignment $z$, it holds that $\tau(\f,z) \geq 2^{-n+j/k}$. Intuitively, the more isolated a solution $z$ is, the more choices of $y$ and $\pi$ would lead to it through  $\PPZMod(\f,y,\pi)$.

Our renewed analysis of PPZ shows that, for any fixed assignment $z_0\in \{0,1\}^n$,  $\PPZMod(\f,y,\pi)$ is also likely to output satisfying assignments that are far away from it. We state Lemma~\ref{intro_ppz_brief} formally in \Cref{lem:anchor:diam} that shows that with probability at least $\frac{1}{2n}\cdot 2^{-n+n/k}$, each iteration of the PPZ algorithm outputs a satisfying assignment $z^*$, such that $$d_H(z_0,z^*) \geq \left(1- \frac{1}{k}\right) \cdot \max_{z' \in \Om} d_H(z_0,z') \;.$$

Thus, we get that PPZ is also an approximate farthest point oracle. More interestingly, the run of PPZ does not depend on $z_0$, and therefore we say that PPZ samples geometrically. We note that the original analysis does not take into account distances between solutions, i.e., the probability of finding a solution only depends on the number of its \emph{immediate} neighbors that are non-solutions. This in itself is a local feature that does not capture global properties like the diameter/dispersion of the solution space. Indeed, our analysis differs 
from the original PPZ analysis in precisely the fact that it exploits this global information (which is needed for diameter/diversity, but not needed if we just want to find one solution).

In order to exploit global geometric properties of the solution space, we view $\Om$ as a subgraph $G_{\f}$ of the $n$-dimensional Hypercube graph. We then divide the vertices in $G_\f$ into $n$ layers, where layer $V_{j}$ consists of all the vertices at distance $j$ from $z_0$ (in $G_\f$). We also define $U_{j}=\bigcup_{j' \geq j} V_{j'}$. Now, we want to show that assignments in higher layers will be reached by $\PPZMod(\f,y,\pi)$ with good probability. We do this by proving that for large enough $j$, either $|U_j|$ is large or the number of cut edges between $U_j$ and $\Om \setminus U_j$ is small in $G_\f$. 

We then use the original Satisfiability Coding Lemma and the fact that an assignment is $j$-isolated if and only if its degree in $G_\f$ is $n-j$, to  show that, for any subset $A$ of the vertices in $G_\f$, it holds that
 $$\tau(\f,A) \geq 2^{-n(1-1/k)}|A|2^{-\left(  \frac{2|E(A)|}{k|A|}+\frac{|S|}{k|A|}\right)}\;,$$ where $E(A)$ denotes the edges in $G_\f$ between vertices in $A$ and $S$ denotes the edges in $G_\f$ between $A$ and $\Om \setminus A$ (\Cref{lem:separator}). 
 We then use the edge isoperimetric lemma for subgraphs of the hypercube which upper bounds the number of edges in the subgraph by a function of the number of vertices in the subgraph. To complete the proof of \Cref{lem:anchor:diam}, we lower bound the probability $\tau(\f,A)$, where $A$ are the assignments in $\Om$ that are far away from $z_0$.

We also show that the above analysis can be extended to prove that for any subset $S \subseteq \Q{n}$, with probability at least $\frac{1}{2n}\cdot 2^{-n+n/k}$, each iteration of the PPZ algorithm outputs a satisfying assignment $z^*$, such that $$\sumd(S,z^*) \geq \left(1- \frac{2}{k+1}\right) \cdot \max_{z' \in \Om} \sumd(S,z') \;.$$ 

This directly implies the existence of a $\left(1-\frac{2}{k+1}\right)$-approximate farthest point oracle that runs in the same time as the PPZ algorithm (\Cref{lem:anchor:sum}). However, we were not able to show a similar lower bound with respect to the $\mind$ distance from $S$. Instead, we can use \Cref{lem:anchor:diam} to show that for every satisfying assignment $z \in \Om$, each iteration of the PPZ algorithm outputs a satisfying assignment within Hamming distance $\frac{n}{k}$ from $z$ (invoke \Cref{lem:anchor:diam} on the antipode of $z$). We can also assume that we have a lower bound on $\max_{z' \in \Om} \mind(S, z')$ on the order of $n/\Theta(1)$ (just exhaustively search all the balls around assignments in $S$ until you hit PPZ running time). Thus, we get an approximate farthest point oracle running in the same time as the PPZ algorithm for the min-dispersion problem as well.

\paragraph{Lemma~\ref{intro_schoning_brief}: Modified \sch's algorithm is a farthest point oracle.} Our second approach for designing farthest point oracles uses \sch's algorithm~\cite{schoning2002probabilistic}. At its core, \sch's algorithm is a local search algorithm that does a random walk from some starting assignment $z_0$. The main subroutine takes as input $z_0$ and, as long as there is a clause that is unsatisfied, picks one of its $k$ literals at random an flips its value. \sch showed that, if there exists a satisfying assignment within Hamming distance $t$ from $z_0$, then within $3t$ steps, the above random walk outputs a satisfying assignment with probability at least $1/(k-1)^t$. By picking the starting point $z_0$ uniformly at random from $\{0,1\}^n$ and letting the random walk go for $3n$ steps, one can then show that the subroutine suceeds with probability at least $((1/2\cdot (1 +1/(k-1)))^n$.

We modify \sch's algorithm by picking the starting point $z_0$ and then setting the length of the random walk more carefully. Suppose we are promised that there exists a satisfying assignment $z^*$ that is distance $r$ (in max-sum or max-min) from some set $S$ of assignments. We then restrict our starting points to be sampled such that they are also guaranteed to be approximately at distance $r$ from $S$. From there, we perform a random walk of small length such that any satisfying assignment we find is also guaranteed to be far away from $S$. The probability that we succeed depends on bounding the set of good starting points: those that are close to the promised $z^*$ (not just far from $S$), since these are the ones most likely to find a satisfying assignment within the length of the random walk. This is the most technically involved step of our analysis. We thus get a farthest point oracle for diameter and all versions of dispersion. Moreover, the \sch strategy can also find heavy-weight assignments. This is done by artificially adding $0^n$ as part of the set $S$ (thus, an assignment that is far from $S$ in Hamming distance will also have a large weight).

%% file: exact.tex
\section{Exact algorithms for diameter and dispersion}
\label{sec:exact}
In this section, we present our algorithm for diameter (\Cref{thm:exactdiam}) and two algorithms for dispersion (\Cref{thm:exactdisp} and \Cref{thm:exactdisp faster}). The problem of computing $\D(\f)$ has been studied by Angelsmark and Thapper \cite{angelsmark2004algorithms}. They give an algorithm that runs in $O^*((2a_k)^n)$ time and $n^{O(1)}$ space, where $O^*(a_k^n)$ is the run-time of a $k$-SAT solver. Note that the strong exponential hypothesis implies that $\lim_{k \to \infty} a_k=2$. We observe that there exists an algorithm to compute $\D(\f)$ exactly, using $O^*(2^n)$ time and $O^*(2^n)$ space. Then, we give two algorithms that compute $\optm(\f,s)$ and $\opts(\f,s)$ in time $O^*(2^{(s-1)n})$ and $O^*(2^{n\omega\ceil{ s/3}})$, where $\omega \leq 2.38$ is the matrix multiplication exponent. In fact, these algorithms do not use the fact that $\f$ is a $k$-CNF formula. We formally define the setup below. 

\paragraph{Preliminaries.} Let $f: \{0,1\}^n \to \{0,1\}$ be a Boolean function computable by an oracle.  Our algorithms use Fourier analysis of Boolean functions, and we briefly recall some facts first. 
\begin{definition}[Fourier Transform]
    Given any function $f: \Q{n} \to \IR$, the Fourier transform of $f$ is defined as follows. 
    $$ \hat{f}(y):= \sum_{x \in \Q{n}} (-1)^{\inner{x,y}} f(x)  \; ,$$ where $\inner{x,y} = \sum_{i=1}^n x_i y_i$. 
\end{definition}
\begin{definition}[Convolution]
    Given two functions $f,g: \Q{n} \to \IR$, we define their convolution to be
    $$ (f*g)(y):=\sum_{x \in \Q{n}}f(x) g(x \oplus y) \; ,$$
\end{definition}
\noindent
where $\oplus$ represents bit-wise addition, modulo $2$. Any function $f :\Q{n} \to \IR$ can be represented as a column vector $f \in \IR^{2^n}$, by indexing the columns using $\{0, 1, \dots, 2^n-1\}$. It can be shown that $\hat{f}= H_{2^n} \cdot f$, where $H_{2^n}$ is the $2^n \times 2^n$ Walsh-Hadamard matrix, which is inductively defined as follows:
$$ H_1= \begin{bmatrix} 1 \end{bmatrix}, H_{2^{m+1}}=\begin{bmatrix} H_{2^m} & H_{2^m}\\
H_{2^m} & -H_{2^m}\end{bmatrix} \text{ for all } m \geq 1 \; .$$
Given the vector $f$, the vector $\hat{f}$ can be computed by a divide and conquer algorithm called the fast Walsh-Hadamard transform that uses  $O(n \cdot 2^n)$ operations. Also, note that $f(x)= \frac{1}{2^n} \sum_{y \in \Q{n}} (-1)^{\inner{x,y}} \hat{f}(y)$. Further, for any two functions $f,g : \Q{n} \to \IR$, $\widehat{f*g}(x)= \hat{f}(x) \hat{g}(x)$, for every $x \in \Q{n}$. This implies that given the vectors $f, g \in \IR^{2^n}$, the vector $f*g \in \IR^{2^n}$ can be computed in $O(n \cdot 2^n)$ time. For more details and proofs of the above facts, we refer the reader to~\cite{odonnell2021analysis}.

\subsection{Computing the diameter of Boolean functions: the proof of \Cref{thm:exactdiam}}
To define the exact diameter of $f$, we slightly abuse notation and define $$\D(f)= \max_{z_1, z_2 \in f^{-1}(1)} d_H(z_1, z_2) \;.$$
We relate computing $\D(f)$ to evaluating the vector $(f*f)$. 

\begin{lemma} \label{lem:convolutiontodiam}
    For any $y \in \Q{n}$, there exist $z_1, z_2 \in f^{-1}(1)$ with $z_1 \oplus z_2 =y$ if and only if $(f*f)(y) > 0$. 
\end{lemma}

\begin{proof}
    Suppose there exist $z_1, z_2 \in f^{-1}(1)$ with $z_1 \oplus z_2 =y$. All the terms in the summation $\sum_{x \in \Q{n}} f(x) f(x \oplus y)$ are either $0$ or $1$, and $f(z_1)f(z_2)=1$ appears in it, implying that $(f*f)(y)>0$. On the other hand, if $(f*f)(y)>0$, this implies that at least one of the terms in the summation is $1$ which implies that there exist $z_1, z_2 \in f^{-1}(1)$ with $z_1 \oplus z_2 =y$.
\end{proof}

\medskip\noindent
The above lemma motivates the following algorithm:

\medskip
\begin{algorithm}[H]
 \label{alg:diam:exact}
 \KwIn{A black box computing a Boolean function $f:\Q{n} \to \{0,1\}$}
 \KwOut{$z_1, z_2 \in f^{-1}(1)$ such that $d_H(z_1, z_2)=\D(f)$ if $f^{-1}(1)\neq \emptyset$, $\perp$ if $f^{-1}(1)=\emptyset$}
 Compute the vector $f \in \IR^{2^n}$ of values of $f$. \\
    Using the fast Walsh-Hadamard transform, compute the vector $\hat{f}= H_{2^n} \cdot f \in \IR^{2^n}$. Multiply each element of this vector with itself to obtain the vector $\hat{f}^2 \in \IR^{2^n}$. \\
    Compute the vector $(f*f) = \frac{1}{2^n} H_{2^n} \cdot \hat{f}^2$ using the fast Walsh-Hadamard transform. Let $z \in \{0,1\}^n$ be any of the vectors with largest Hamming weight such that $(f*f)(z)>0$. Output $\perp$ if there is no such $z$, abort.\\
    Find any $x \in \Q{n}$ such that $f(x)=f(x \oplus z)=1$ and output $x, x \oplus z$.
 \caption{Exact diameter using Fourier transform}
 \end{algorithm}

\medskip\noindent
Each step of this algorithm uses $O^*(2^n)$ time and $O^*(2^n)$ space, which proves \Cref{thm:exactdiam}. 
\subsection{Exact algorithms for dispersion using Fourier transforms}

We now generalize the above algorithm for diameter to dispersion, where our objectives are defined over the $f^{-1}(1)$ (similarly as in the diameter case). In the following section we present another algorithm with faster running time, but that algorithm works for $s\geq 6$. Our algorithm presented below can be used for all values of $s$.

\begin{restatable}{theorem}{exactdispersion}\label{thm:exactdisp}
	Let $f: \Q{n} \to \{0,1\}$ be a function computable by a black box and let $s$ be a given parameter. Then, there exist deterministic algorithms $\mathcal{A}_1, \mathcal{A}_2, \mathcal{A}_3$ that make $2^n$ oracle calls to $f$ and in addition to that, use $O^*(2^{(s-1)n})$ time and $O^*(2^n)$ space provide the following guarantees.
 \sloppy
 \begin{enumerate}
     \item The output of $\mathcal{A}_1$ is a multiset $\{z_1, z_2, \dots, z_s\} \subseteq f^{-1}(1)$ such that $\SPD(z_1, z_2, \dots, z_s) = \opts(f,s)$. 
     \item The output of $\mathcal{A}_2$ is a set $\{z_1, z_2, \dots, z_s \} \subseteq f^{-1}(1)$ such that $\PD(z_1, z_2, \dots, z_s) = \optm(f,s)$. 
     \item If $|f^{-1}(1)| \geq s$, the output of $\mathcal{A}_3$ is a set $\{z_1, z_2, \dots, z_s \} \subseteq f^{-1}(1)$ such that $\SPD(z_1, z_2, \dots, z_s)$ $ =$ $\opts_{\neq}(f,s)$. 
 \end{enumerate}
\end{restatable}

We now prove \Cref{thm:exactdisp}. We begin by observing that for every $(z_0, z_1, \dots, z_{s-1}) \in \Q{sn}$, $\SPD(z_0, z_1, \dots, z_{s-1})=\SPD(0, z_1 \oplus z_0, \dots, z_{s-1} \oplus z_0)$ and $\PD(z_0, z_1, \dots, z_{s-1})=\PD(0, z_1 \oplus z_0, \dots, z_{s-1} \oplus z_0)$. Hence, the value of $\SPD(z_0, z_1, \dots, z_{s-1})$ and $\PD(z_0, z_1, \dots, z_{s-1})$ are determined entirely by $y_1, y_2, \dots, y_{s-1}$, where $y_j=z_0 \oplus z_j$ for each $j \in \{1,2,\dots, s-1\}$. Next, we prove the following generalization of \Cref{lem:convolutiontodiam}. 

\begin{lemma} \label{lem:convolutiondiversity}
    For every $w_1, w_2, \dots, w_{s-2} \in \Q{n}$, define the function $g_{(w_1, w_2, \cdots, w_{s-2})}(x):=f(x) f(x \oplus w_1) f(x \oplus w_2) \dots f(x \oplus w_{s-2})$. For every $y, \in \Q{n}$ and $w_1, w_2, \dots, w_{s-2} \in \Q{n}$, there exist $z_0, z_1, \dots, z_{s-1} \in f^{-1}(1)$, such that $z_1=z_0 \oplus y$ and $z_j=z_0 \oplus w_{j-1} \oplus y$ for each $j \in \{2,3, \dots, s-1\}$ if and only if $f*g_{(w_1, w_2, \cdots, w_{s-2})}(y) >0$.
\end{lemma}
\begin{proof}
    Suppose that there exist $z_0, z_1, \dots, z_s \in f^{-1}(1)$, such that $z_1=z_0 \oplus y$ and $z_j=z_0 \oplus w_{j-1} \oplus y$ for $j \in \{2,3, \dots, s-1\}$. This implies that $f(z_0) f(z_1) \dots f(z_s)=f(z_0) f(z_0 \oplus y) f(z_0 \oplus w_1 \oplus y) \dots f(z_0 \oplus w_{s-2} \oplus y)=1$. Because $f*g_{(w_1, w_2, \cdots, w_{s-2})}(y)= \sum_{x \in \Q{n}} f(x) g_{(w_1, w_2, \cdots, w_{s-2})}(x \oplus y)=\sum_{x \in \Q{n}} f(x) f(x \oplus y) f(x \oplus y \oplus w_1) f(x \oplus y \oplus w_2) \dots f(x \oplus y \oplus w_{s-2})$, this summation has a non-zero term ($x=z_0$) which implies that $f*g_{(w_1, w_2, \cdots, w_{s-2})}(y) >0$. 
    
    On the other hand, if $f*g_{(w_1, w_2, \cdots, w_{s-2})}(y) >0$, this implies that at least one of the terms in this summation is $1$. This implies that there exists $x \in \Q{n}$ such that $f(x) = f(x \oplus y) = f(x \oplus y \oplus w_1) = f(x \oplus y \oplus w_2) \dots =f(x \oplus y \oplus w_{s-2})=1$. Now let $z_0=x$, $z_1=x \oplus y$, and $z_j=x \oplus y \oplus w_j-1$, thus proving that there exists $z_0, z_1, \dots, z_{s-1} \in f^{-1}(1)$, such that $z_1=z_0 \oplus y$ and $z_j=z_0 \oplus w_{j-1} \oplus y$ for each $j \in \{2,3, \dots, s-1\}$ if and only if $f*g_{(w_1, w_2, \cdots, w_{s-2})}(y) >0$.
\end{proof}

\medskip
Hence, for each $w_1, w_2, \dots, w_{s-2} \in \Q{n}$, we can run the following procedure to compute an array containing the values of $f*g_{(w_1, w_2, \cdots, w_{s-2})}(y)$ for every $y \in \Q{n}$. 

\medskip
\begin{algorithm}[H]
 \label{alg:disp:exact:intermediate}
 \KwIn{A black box computing $f:\Q{n} \to \{0,1\}$, $w_1, w_2, \dots, w_{s-2} \in \Q{n}$.}
 \KwOut{An array $f*g_{(w_1, w_2, \cdots, w_{s-2})} \in \IR^{2^n}$ containing the values of $f*g_{(w_1, w_2, \cdots, w_{s-2})}(y)$ for every $y \in \Q{n}$. }
 Compute the vectors $f,g_{(w_1, w_2, \cdots, w_{s-2})} \in \IR^{2^n}$ with the values of $f(x)$ and $g_{(w_1, w_2, \cdots, w_{s-2})}(x)$ for each $x \in \Q{n}$. \\
    Compute the vectors $\hat{f} = H_{2^n} \cdot f , \hat{g}_{(w_1, w_2, \cdots, w_{s-2})}= H_{2^n} \cdot g_{(w_1, w_2, \cdots, w_{s-2})}$ using the fast Walsh-Hadamard transform. \\
    Compute the vector $\hat{f} \cdot \hat{g}_{(w_1, w_2, \cdots, w_{s-2})} \in \IR^{2^n}$ by multiplying the elements of $\hat{f}$ and $\hat{g}_{(w_1, w_2, \cdots, w_{s-2})}$ element-wise. \\
    Compute the vector $f * g_{(w_1, w_2, \cdots, w_{s-2})} = \frac{1}{2^n} H_{2^n} \cdot \left( \hat{f} \cdot \hat{g}_{(w_1, w_2, \cdots, w_{s-2})} \right)$ using the fast Walsh-Hadamard transform. 
    \caption{Algorithm to compute convolution of $f$ and $g_{w_1, w_2, \dots, w_{s-2}}$.}
 \end{algorithm}
\vspace{2mm} \noindent 

This implies that by iterating over all $(w_1, w_2, \cdots, w_{s-2}) \in \Q{(s-2)n}$, we can compute $\opts(f,s)$ and $\optm(f,s)$ using $O^*(2^{(s-1)n})$ time and $O^*(2^n)$ space. We formally define the algorithm below. Note that we have defined it to compute $\opts(f,s)$, but the same algorithm with minor modifications can be used to compute $\optm(f,s)$ and $\opts_{\neq}(f,s)$.

\begin{algorithm}[H]
 \label{alg:disp:exact}
 \KwIn{A black box computing a Boolean function $f:\Q{n} \to \{0,1\}$}
 \KwOut{$z_1, z_2, \dots z_s \in f^{-1}(1)$ such that $\sumd(z_1, z_2, \dots, z_s)=\opts(f,s)$ if $f^{-1}(1)\neq \emptyset$, $\perp$ if $f^{-1}(1)=\emptyset$}
Initialize $\mathcal{M}=\perp, y_1, y_2, \dots, y_{s-1}=\perp$.\\
\For{$(w_1, w_2, \dots, w_{s-2}) \in \Q{(s-2)n}$}{
    Compute an array containing the values of $f * g_{(w_1, w_2, \cdots, w_{s-2})}(y)$ for each $y \in \Q{n}$ using Algorithm~\ref{alg:disp:exact:intermediate}.\\
    \For{$y \in \Q{n}$}{
    \If{$f * g_{(w_1, w_2, \cdots, w_{s-2})}(y)>0$ and $\SPD(0, y, y \oplus w_1, y \oplus w_2, \dots, y \oplus w_{s-2}) > \mathcal{M}$}{set $\mathcal{M}:=\SPD(0, y, y \oplus w_1, y \oplus w_2, \dots, y \oplus w_{s-2})$ , $y_1=y, y_2=y\oplus w_1, \dots, y_{s-1}=y \oplus w_{s-2}$.}
    
    }

 }
 \If{$\mathcal{M}=\perp$}{
  output $\perp$
 }
 \Else{
    If there exists $x \in \Q{n}$ such that $f(x)=f(x \oplus y_1)= \dots, f(x \oplus y_{s-1})=1$, output $z_0=x, z_1= x \oplus y_1, z_2= x \oplus y_1, \dots, z_{s-1}=x \oplus y_{s-1}$

 }
 
    \caption{Algorithm for exact dispersion using Fourier transforms}
 \end{algorithm}
\begin{remark}
    To design an algorithm for $\optm(f,s)$, we replace the comparison in line 5 of the algorithm with one using $\PD$ instead of $\SPD$. An algorithm to compute $\opts_{\neq}(f,s)$ would be identical, except that we would iterate over $w_1, w_2, \dots, w_{s-2}$ such that they are all different, and in the inner loop, we would iterate over all $y \neq \mathbf{0}$.
\end{remark}

\paragraph{Proof of correctness:} Define the $n$-dimensional subspace $V \subseteq \Q{(s-1)n}$ to be $\{(x,x,\dots, x) \mid x \in \Q{n}\}$, which partitions $\Q{(s-1)n}$ into the $2^{(s-2)n}$ cosets $V_{(w_1, w_2, \cdots, w_{s-2})}=\{(x, x\oplus w_1, x \oplus w_2, \dots, x \oplus w_{s-2}) \mid x \in \Q{n}\}$ for each $s$-tuple $(w_1, w_2, \cdots, w_{s-2}) \in \Q{(s-2)n}$. \Cref{lem:convolutiondiversity} implies that for each $(y_1, y_2, \dots, y_{s-1})=(y, y \oplus w_1, y \oplus w_2, \dots, y \oplus w_{s-2}) \in V_{(w_1, w_2, \cdots, w_{s-2})}$, there exists $z_0, z_1, \dots, z_s \in f^{-1}(1)$ with $z_j=z_0 \oplus y_j$ for $j \in \{1,2,\dots, s-1\}$ if and only if $f*g_{(w_1, w_2, \cdots, w_{s-2})}(y) >0$. This completes the proof of \Cref{thm:exactdisp}.

\input{exact-dispersion}

%% file: exact-dispersion.tex
\subsection{Exact Algorithms for Dispersion Using Clique-Finding}\label{sec:exactdisp faster}

 In this section, we discuss an alternate technique for exactly computing dispersion. The running time and space of the algorithm depend on the size of the solution space $\Om$. For any $s\geq 6$, the algorithm runs faster than the one in~\Cref{sec:exact}, but at the cost of potentially higher space.

 We now formulate our results to work for dispersion over an arbitrary subset $X$ of the hypercube, of size $M$. We thus slightly abuse notation and define $\opts(X,s)$, $\optm(X,s)$ and $\opts_{\neq}(X,s)$. In what follows, $\omega \le 2.38$ denotes the matrix multiplication exponent~\cite{williams2024new}.

\begin{restatable}{theorem}{exactdispersion-faster}\label{thm:exactdisp faster}
	There exist deterministic algorithms $\mathcal{A}_1, \mathcal{A}_2, \mathcal{A}_3$ that given as input a non-empty set $X \subseteq \Q{n}$ of size $M$ and parameter $s$, runs in $O(\poly(n, s) \cdot M^{\omega \lceil s/3 \rceil})$ time, uses $O(M^{2 \lceil s/3 \rceil})$ space, and have the following behaviour.
 \sloppy
 \begin{enumerate}
     \item The output of $\mathcal{A}_1$ is $z_1, z_2, \dots, z_s \in X$ such that $\SPD(z_1, z_2, \dots, z_s) = \opts(X,s)$. 
     \item The output of $\mathcal{A}_2$ is $z_1, z_2, \dots, z_s \in X$ such that $\PD(z_1, z_2, \dots, z_s) = \optm(X,s)$. 
     \item And, as long as $|S| \geq s$, the output of $\mathcal{A}_3$ is a set $\{z_1, z_2, \dots, z_s\} \in X$ such that $\SPD(z_1, z_2, \dots, z_s)$ $ =$ $\opts_{\neq}(X,s)$. 
 \end{enumerate}
\end{restatable}

Note that when applied with $X$ being the set of satisfying assignments to a formula $\f$, the running time is at worst $O(2^{\omega \lceil s/3 \rceil n})$ but in general much faster depending on the number of satisfying assignments.  Furthermore, these algorithms do not rely on the underlying space being $\Q{n}$; they can be used on any $M$-point metric space.

The algorithms use the same idea as $O(n^{c_s s})$ time algorithms for finding a clique of size $s$ in a graph, where $c_s \approx \omega/3$ with variations depending on $s \bmod 3$ \cite{EisenbrandG2004Clique}. In particular the $\optm$ problem immediately reduces to the $s$-clique problem by creating a graph on $X$ where $x, y \in X$ are connected by an edge if their distance is at least $d$ (for some guess $d \in [0,n]$ for the value of $\optm(X, s)$, which we can then binary search over).  Similarly for the $\opts$ objective function, the problem reduces to finding an $s$-clique of maximum weight in an edge-weighted graph, which can be solved by similar methods. Similar ideas have been used before in for example~\cite{williams2005new}.

Let us describe the algorithms in more detail, starting with the case of $\optm$ since it is easier.  While in this case the reduction to $s$-clique described above could be used directly, let us still take a slightly longer route and reduce to triangle-finding, in order to provide a warm-up for the $\opts$ algorithm where this is needed.

To simplify notation we assume that $s$ is divisible by $3$.  Given a guess $d \in [0,n]$ for the value of $\optm(X, s)$, define a graph $G_d$ where the vertex set is
\[
V(G_d) = \left\{ (x_1, \ldots, x_{s/3}) \in X^{s/3} \,|\,d_H(x_i, x_j) \ge d \text{ for all } 1 \le i < j \le s/3 \right\}.
\]
Two vertices $(x_1, \ldots, x_{s/3})$ and $(y_1, \ldots, y_{s/3})$ are connected by an edge if $d_H(x_i, y_j) \ge d$ for all $i$ and $j$.  Note that $G_d$ has $O(M^{s/3})$ vertices and $O(M^{2s/3})$ edges, and can be constructed in $O(s^2 \cdot M^{2s/3})$ time.

\begin{claim}
    Three tuples $(x_1, \ldots, x_{s/3})$, $(y_1, \ldots, y_{s/3})$, and $(z_1, \ldots, z_{s/3})$ form a triangle in $G_d$ if and only if $\PD(x_1, \ldots, x_{s/3}, y_1, \ldots, y_{s/3}, z_1, \ldots, z_{s/3}) \ge d$.
\end{claim}

This immediately gives us the algorithm $\mathcal{A}_2$ for $\optm(X, s)$: try all possible values of $d$, construct the graph $G_d$, and then search for a triangle in $G_d$, which can be done in $O(|V(G_d)|^{\omega}) = O(M^{\omega s / 3})$ time~\cite{itai1977finding}.

Moving on to the $\opts$ objective function, we change the above algorithm as follows.  Given six values $\vec{d} = (d_1, d_2, d_3, d_{12}, d_{23}, d_{13}) \in [0, sn]^6$, we define the tri-partite graph $G_{\vec{d}}$ with vertex sets $V_1, V_2, V_3$ defined by
\[
V_k(G_{\vec d}) = \left\{ (x_1, \ldots, x_{s/3}) \in X^{s/3} \,|\, \frac{1}{2} \sum_{i,j} d_H(x_i, x_j) \ge d_k \right\}.
\]
Two vertices $(x_1, \ldots, x_{s/3}) \in V_{k}$ and $(y_1, \ldots, y_{s/3}) \in V_{k'}$ are connected by an edge if 
\[
\sum_{i,j} d_H(x_i, y_j) \ge d_{k,k'}.
\]
We then have the following claim, which yields the algorithm $\mathcal{A}_1$ (by enumerating all $O((ns)^6)$ possible values of $\vec{d}$).\footnote{Note that, if we reduced $\opts$ to an $s$-clique problem instead of triangle finding, there would be ${s \choose 2}$ distances to guess, which would lead to an extra runtime factor of roughly $n^{s^2/2}$.  This is why we reduce to triangle-finding instead.}
\begin{claim}
    If three vertices $(x_1, \ldots, x_{s/3}) \in V_1$, $(y_1, \ldots, y_{s/3}) \in V_2$, and $(z_1, \ldots, z_{s/3})$ form a triangle in $G_{\vec{d}}$ then $\SPD(x_1, \ldots, x_{s/3}, y_1, \ldots, y_{s/3}, z_1, \ldots, z_{s/3}) \ge d_1 + d_2 + d_3 + d_{12} + d_{23} + d_{13}$.
    Conversely, there exists a $\vec{d}$ such that $d_1 + d_2 + d_3 + d_{12} + d_{23} + d_{13} \ge \opts(S, s)$ and $G_{\vec{d}}$ has a triangle.
\end{claim}

Finally, to get the algorithm $\mathcal{A}_3$ for $\opts_{\neq}(X, s)$, we simply change the definition of the vertices and edges of $G_{\vec{d}}$ to exclude any tuples with repeated strings.

%% file: PPZ.tex
\section{The PPZ algorithm performs geometry-based sampling}\label{sec:PPZ}
This section is devoted to proving \Cref{thm:ppz-for-dia}, \Cref{thm:ppz-for-sumdisp-easy} and \Cref{thm:ppz-for-mindisp}, which we restate below. In fact, we prove a slightly stronger version of \Cref{thm:ppz-for-sumdisp-easy}, which is stated here.

\ppzdiam*

We now state the full version of \Cref{thm:ppz-for-sumdisp-easy}.
\begin{restatable}[PPZ approximating $\opts(\f,s)$]{theorem}{ppzsumdisp} \label{thm:ppz-for-sumdisp}
    Let $\f$ be a $k$-CNF formula on $n$ variables. There exists a randomized algorithm running in time $O^*\left(s^4 \cdot 2^{n-n/k}\right)$ that takes $\f$ and an integer $s \geq 1$ as input and if $\f$ is satisfiable, with probability at least $1-o(1)$, outputs a multiset $S^* \subseteq \Om$ of size $s$ such that:
    \begin{enumerate}
        \item $\SPD(S^*) \geq \frac{1}{2}\cdot \left(1-\frac{2}{k+1}\right) \cdot \opts(\f,s)$ if $s \leq 3 + \floor{\frac{4}{k-1}}$.
        \item $\SPD(S^*) \geq \frac{k-1}{k+3}\left(\frac{1-\frac{1}{s}}{1+\frac{k-1}{(k+3)}\cdot \frac{1}{s} }\right) \cdot \opts(\f,s)$ if $s \geq 3 + \ceil{\frac{4}{k-1}}$. 
    \end{enumerate}
\end{restatable}

\ppzmindisp*

\paragraph{Proof organization: } We prove the above three theorems in parallel using the following five step procedure. 
\begin{enumerate}
    \item In \Cref{sec:ppzdetails}, we summarize the PPZ algorithm and state the satisfiability coding lemma.
    \item In \Cref{sec:separator}, we prove the \emph{separator lemma}, that generalizes the satisfiablity coding lemma. 
    \item In \Cref{sec:geometric}, we prove \emph{geometric sampling properties} of PPZ, with respect to $\D$ and $\optm$ in \Cref{lem:anchor:diam}, and $\opts$ in \Cref{lem:anchor:sum}. 
    \item In \Cref{sec:alg}, we use these geometric properties to develop farthest point oracles for $\optm$ and $\opts$. 
    \item In \Cref{sec:finalppz}, we describe our algorithms for finding dispersed solutions with respect to $\opts$ and $\optm$. These algorithms use the farthest point oracles in the well known algorithms for dispersion studied by Gonzales~\cite{gonzalez1985clustering} and Cevallos, Eisenbrand and Zenklusen~\cite{cevallos2019improved}.  
\end{enumerate}

\paragraph{Notation.} We use a graph theoretical framework to analyze the PPZ algorithm. Let $G_\f$ be the subgraph of the $n$-dimensional boolean hypercube induced by the set of satisfying assignments of $\f$. That is, the vertex set of $G_\f$ is $\Om$, and $z,z' \in \Om$ are connected in $G_\f$ if $d_H(z,z')=1$. For any $(z,z')$ connected in $G_\f$, $z' = z \oplus e_k$ for some $k \in [n]$, where $e_k \in \{0,1\}^n$ is the $k$-th standard basis vector. For any $z \in \Om$, we use $\deg(z)$ to denote its degree in the graph $G_\f$. 

\subsection{The PPZ algorithm} \label{sec:ppzdetails}

In this section, we formally define the subroutine used in the PPZ algorithm and recall its analysis.

\noindent

\PPZMod. This subroutine takes as input a $k$-CNF formula $\f$, a string $y \in \{0,1\}^n$, and a permutation $\pi \in \mathcal{S}_n$ of length $n$. It iteratively computes a string $u \in \{0,1\}^n$ in $n$ steps. 

\noindent
Let $\f_0 = \f$. In each step $i$, the algorithm computes $u_{\pi(i)}$ and updates the formula $\f_{i-1}$ to $\f_i$ as follows: if $\f_{i-1}$ has a clause $C = (x_{\pi(i)})$, then it sets $u_{\pi(i)}$ to $1$; if it has a clause $C = (\overline{x_{\pi(i)}})$ then it sets $u_{\pi(i)}$ to $0$, and if there is no such clause, i.e., any clause containing the variable $x_{\pi(i)}$ has two or more variables, then it sets $u_{\pi(i)}$ equal to $y_{\pi(i)}$. It updates $\f_{i-1}$ to $\f_i$ by setting all instances of the variable $x_{\pi(i)}$ as per $u_{\pi(i)}$ and simplifying the formula as needed (i.e., removing satisfied clauses and eliminating $0$-valued literals from all clauses).  After $n$ steps, the algorithm outputs $u \in \{0,1\}^n$ as computed above.

For any $z \in \{0,1\}^n$, let $\tau(\f,z)$ denote the probability that $\PPZMod(\f,y,\pi)$ outputs $z$ when $y$ and $\pi$ are chosen independently and uniformly at random from $\{0,1\}^n$ and $\mathcal{S}_n$, respectively. For any subset $A \subseteq \{0,1\}^n$, we use $\tau(\f,A)$ to denote the probability that $\PPZMod(\f,y,\pi)$ outputs an assignment in $A$ over $y$ and $\pi$ chosen independently and uniformly at random. For any fixed $\pi,y$, the procedure $\PPZMod$ outputs a fixed assignment that only depends on $\pi$ and $y$, which implies that
\[\tau(\f, A)  =\sum_{z \in A} \tau(\f, z) \;.\]
In their paper~\cite{PPZ}, Paturi, Pudl{\'a}k and Zane proved the satisfiability coding lemma, which states that for a satisfying assignment $z$, $\tau(\f,z)$ depends on how \emph{isolated} $z$ is (i.e, its degree in $G_\f$).
\begin{lemma}[Satisfiability Coding Lemma (Paturi, Pudl{\'a}k, Zane~\cite{PPZ})] 
\label{lem:ppz-j-isolated}
Let $\f$ be a $k$-CNF formula on $n$ variables. Let $y$ be chosen uniformly at random from $\{0,1\}^n$ and $\pi$ be chosen uniformly at random from $\mathcal{S}_n$. Let $z$ be a satisfying assignment of $\f$ such that $\deg(z) = n-j$ for some $j \in [n]$. Then, the probability that \PPZMod$(\f,y,\pi)$ outputs $z$ is at least $2^{-n+j/k}$. 
\end{lemma} 
\noindent
If $\Om$ is non-empty (i.e, $\f$ is satisfiable), they show that $\sum_{z \in \Om} 2^{-\deg(z)/k} \geq 1$, which implies the following lower bound on the probability that $\PPZMod$ outputs any satisfying assignment to $\f$.
$$\tau(\f,\Om)=\sum_{z \in \Om}\tau(\f,z)=2^{-n+n/k} \sum_{z \in \Om} 2^{-\deg(z)/k} \geq 2^{-n+n/k}\;.$$
This implies that repeating $\PPZMod$ $O^*\left(2^{n(1-1/k)}\right)$ times is enough to output a satisfying assignment to $\f$ with probability $1-o(1)$, if one exists.

\subsection{The separator lemma}
\label{sec:separator}

We first generalize \Cref{lem:ppz-j-isolated} to lower bound $\tau(\f,A)$ for arbitrary sets $A$ of satisfying assignments.
\begin{lemma}[Separator Lemma] \label{lem:separator}
    Let $A \subseteq \Om$, let $S$ be the set of edges of $G_\f$ with one endpoint in $A$ and the other endpoint in $\Om \setminus A$. Further, let $E(A)$ be the  edges of $G_\f$ with both endpoints in $A$. Then, 
\begin{align}
    \tau(\f,A) & \geq 2^{-n(1-1/k)}|A|2^{-\left(  \frac{2|E(A)|}{k|A|}+\frac{|S|}{k|A|}\right)} \label{eq:sep1} \\
    & \geq 2^{-n(1-1/k)}|A|^{1-1/k}2^{-\frac{|S|}{k|A|}} \label{eq:sep2}
\end{align}    
    
\end{lemma}

\begin{proof}
\begin{align}
\tau(\f,A) & = \sum_{z \in A} \tau(\f, z) \geq \sum_{z \in A} 2^{-n + (n - \deg(z) / k} & \text{by~\Cref{lem:ppz-j-isolated}} \notag \\
& = \sum_{z \in A} 2^{-n(1-1/k)  - \deg(z) / k} & \notag \\
& = 2^{-n(1-1/k)} \cdot \sum_{z \in A}   2^{- \deg(z) / k} & \notag \\
& \geq 2^{-n(1-1/k)} \cdot |A| \cdot  2^{-\frac{\sum_{z \in A} \deg(z)}{k \cdot |A|}} & \text{By AM-GM inequality} \notag \\
 & \geq 2^{-n(1-1/k)} \cdot |A| \cdot  2^{- \frac{2|E(A)| + |S|}{k \cdot |A|}} & \text{By the handshake lemma} \notag \\
 & = 2^{-n(1-1/k)} \cdot 2^{-\frac{|S|}{k \cdot |A|}} \cdot \left(|A| \cdot 2^{-\frac{2 |E(A)|}{k \cdot |A|}}\right) & \notag \\
 & \geq 2^{-n(1-1/k)} \cdot 2^{-\frac{|S|}{k \cdot |A|}} \cdot \left(|A| \cdot 2^{-\frac{\log (|A|)}{k}}\right) & \text{By the edge isoperimetric inequality}\notag \\
 & \geq 2^{-n(1-1/k)} \cdot 2^{-\frac{|S|}{k \cdot |A|}} \cdot \left(|A|^{1-1/k} \right) &\notag 
\end{align}
For completeness, recall the edge-isoperimetric inequality for subgraphs of a hypercube~\cite{bollobas1986combinatorics}, that states that for any subset $A \subseteq \{0,1\}^n$, $|E(A)| \leq (|A| \log(|A|))/2$. 
\end{proof}

\subsection{Geometric sampling properties of $\PPZMod$} \label{sec:geometric}
In this section, we prove the dispersion properties of the $\PPZMod$ subroutine. The goal is to show that $\PPZMod$ is acts like an approximate farthest oracle: if a satisfying assignment exists that is ``far away'' from a set of already chosen solutions, then $\PPZMod$ will output an approximately ``far away'' satisfying assignment with good probability. 

\medskip\noindent
In particular, let $z_0 \in \Q{n}$ be any (not necessarily satisfying) assignment to $\f$. Let $r$ denote the maximum distance from $z_0$ to any satisfying assignment in $\Omega_\f$. We show that $\PPZMod$ will output, with probability at least $n^{-O(1)} \cdot 2^{-n+n/k}$, a satisfying assignment $z$ such that $d_H(z,z_0) \geq \left(1-\frac{1}{k}\right)r$. As a corollary, this implies that for any satisfying assignment $z$, $\PPZMod$ outputs a satisfying assignment to $\f$ within distance $n/k$ of $z$ with probability at least $n^{-O(1)} \cdot 2^{-n+n/k}$. Formally, we show that:
\begin{lemma}\label{lem:anchor:diam}
    Let $\f$ be a satisfiable $k$-CNF formula, $z_0 \in \{0,1\}^n$, and $r= \max_{z \in \Om} d_H (z, z_0)$. Let $y$ and $\pi$ be chosen uniformly at random and independently from $\{0,1\}^n$ and $\mathcal{S}_n$ respectively. The probability that $\PPZMod(\f, y, \pi)$ outputs $z^* \in \Om$ with $d_H(z^*,z_0) \geq \left(1-1/k\right) \cdot r$ is at least $\frac{1}{2n} \cdot 2^{-n+n/k}$
\end{lemma}
\begin{proof}
    We partition the vertices of $G_\f$ based on the value of $d_H(\cdot, z_0)$. For $ 0\leq i \leq n$, we define $V_i = \{ z \in \Om \mid d_H(z,z_0)= i\}$. We define $U_i = \bigcup_{j \geq i} V_j$ for $0 \leq i \leq n$. For any $z \in V_i$, the neighbours of $z$ are either in $V_{i+1}$ or $V_{i-1}$. For each $V_i$, let $S_i$ denote the set of edges between $V_i$ and $V_{i-1}$. 

    \medskip \noindent
    Let $i^*= \lceil \alpha \cdot r \rceil$, where $\alpha = \left(1-1/k\right)$. We will show that $\tau(\f, U_{i^*}) \geq \frac{1}{2n} \cdot 2^{-n+n/k}$. Note that for any $i$, the edges that have one end point in $U_{i}$ and the other in $\Om \setminus U_{i}$, is the set of edges between $V_{i}$ and $V_{i-1}$, i.e. $S_i$. Hence, from Lemma~\ref{lem:separator} inequality~\ref{eq:sep2}, we get that \[\tau(\f, U_{i^*}) \geq 2^{-n+n/k} \cdot |U_{i^*}|^{1-1/k} \cdot 2^{-\frac{|S_{i^*}|}{k|U_{i^*}|}}.\]
    \paragraph{Upper bounding $|S_{i^*}|/|U_{i^*}|$:} For any $z \in V_{i^*}$, consider any vertex $z' \in V_{i^*-1}$ that is connected to $z$. Because $z$ and $z'$ are connected, we have that there exists $m \in [n]$ such that $z'=z\oplus e_m$. Also, because $d_H(z', z_0)=d_H(z, z_0)-1$,  $m$ must be in the support of the vector amongst $z_0 \oplus z$. Hence, there are at most $i^*$ possible choices for $m$ to take. Therefore, $z$ is connected to at most $i^*$ vertices in $V_{i^*-1}$, $\frac{|S_{i^*}|}{|U_{i^*}|}$ is upper bounded by $i^*$, and
    \[ \tau(\f, U_{i^*}) \geq 2^{-n+n/k} \cdot |U_{i^*}|^{1-1/k} \cdot 2^{ -i^*/k} \;.\]
    Now, the task is to lower bound $|U_{i^*}|^{1-1/k} \cdot 2^{ -i^*/k}$ by $\frac{1}{2n}$.

    \paragraph{Lower bounding $|U_{i^*}|$:} In what follows, we will show that either $|U_{i^*}|^{1-1/k} \cdot 2^{ -i^*/k} \geq \frac{1}{2n}$, or $\tau(\f, U_j) \geq 2^{-n+n/k}$ for some $j > i^*$. As $U_j \subseteq U_{i^*}$ for any $j \geq i^*$, this would imply that $\tau(\f, U_{i^*}) \geq 2^{-n+n/k}$. 
    
    \medskip\noindent
    Assume that $\tau(\f, U_j) < 2^{-n+n/k}$, for every $i^*<j \leq r$. Lemma~\ref{lem:separator} inequality~\ref{eq:sep1} implies that
    \begin{equation}
          2^{-n+n/k}\cdot |U_{j}| \cdot 2^{-\left(  \frac{2|E(U_{j})|}{k|U_{j}|}+\frac{|S_{j}|}{k|U_{j}|}\right)} \leq \tau(\f, U_j) < 2^{-n+n/k} \text{ for all } i^*<j \leq r
    \end{equation}
    This implies that
    \begin{equation}
        |U_j| \leq 2^{\left(  \frac{2|E(U_{j})|}{k|U_{j}|}+\frac{|S_{j}|}{k|U_{j}|}\right)} \text{ for all } i^*<j \leq r
    \end{equation}
    Further, note that $\frac{2|E(U_{j})|}{k|U_{j}|}+\frac{|S_{j}|}{k|U_{j}|} \leq \frac{2|E(U_{j})|+ 2|S_{j}|}{k|U_{j}|}=\frac{2|E(U_{j-1})|}{k|U_{j}|}$. We now use the edge isoperimetric inequality in the hypercube which implies that $|E(U_{j-1})| \leq  \frac{|U_{j-1}| \log(|U_{j-1}|)}{2}$. This implies that
    \begin{equation}
        |U_{j-1}| \log(U_{j-1}) \geq k \cdot |U_j| \log(|U_j|) \text{ for all } i^*<j \leq r
    \end{equation}
    The set $U_{r}$ is non-empty, and because $\tau( \f, U_{r}) < 2^{-n+n/k}$, $S_{r}$ is non-empty by the satisfiability coding lemma. This implies that $|V_{r-1}| \geq 1$, $|U_{r-1}| \geq 2$, and $|U_{r-1}|\log(|U_{r-1}|) \geq 2$. This in turn implies $|U_{i^*}| \log(|U_{i^*}|) \geq 2k^{r-i^*-1}$ by combining the inequalities for all $i^*+1\leq j \leq r-1$. 
    
    \medskip \noindent
    For $k \geq 3$, this implies that $|U_{i^*}| \geq 2^{r-i^*}$. Because $i^*=\ceil { \frac{(k-1)r}{k} }$, this implies that $r \leq \frac{k i ^*}{k-1}+1$, implying that $|U_{i^*}| \geq \frac{1}{2} \cdot 2^{\frac{i^*}{k-1}}$. Because $\tau(\f, i^*) \geq 2^{-n+n/k} |U_{i^*}|^{1-1/k}2^{-i^*/k}$, this implies that $\tau(\f, i^*) \geq \frac{1}{2}\cdot 2^{-n+n/k}$. 
    
    \medskip \noindent
    We now consider the case that $k=2$. As $\log(|U_{i^*}|) \leq n$, 
    $$ |U_{i^*}| \geq \frac{1}{n} \cdot 2^{r-i^*} \geq \frac{1}{2n} \cdot 2^{\frac{i^*}{k-1}} \; ,$$
    which proves that $\tau(\f, U_{i^*}) \geq \frac{1}{2n} \cdot 2^{-n+n/k}$.
\end{proof}

The above lemma proves geometric sampling properties of PPZ for $\D$ and $\optm$. Now we consider $\opts$: there exists a multi-set of assignments $T$, and our goal is to find a satisfying assignment $z^*$ that maximises the sum of distances from the assignments in $T$, denoted as $\sumd(z^*,T)$. We show that with probability at least $\frac{1}{2n} \cdot 2^{-n+n/k}$, $\PPZMod$ outputs such a satisfying assignment, with an approximation factor of $\left(1-\frac{2}{k+1}\right)$. We employ the same strategy as in the proof of Lemma~\ref{lem:anchor:diam}, dividing the vertex set of $G_\f$ into levels based on $\sumd(\cdot, T)$. However, in this case, we can no longer argue that a vertex $z \in V_{i^*}$ neighbors in only $V_{i^*-1}$ and $V_{i^*+1}$. This is because changing one coordinate in $z$ does not necessarily decrease the objective function $\sumd(z, T)$ by just one. Hence, bounding the size of the separator $S_{i^*}$, where $S_{i^*}$ is the set of edges between $U_{i^*}$ and $G_{\f} \setminus U_{i^*}$ is more involved.

\begin{lemma}\label{lem:anchor:sum}
    Let $\f$ be a satisfiable $k$-CNF formula, $T \subseteq \{0,1\}^n$ be a multiset of size $t$, and $r_{\text{sum}}= \max_{z \in \Om}\sumd (z,T)$. Let $y$ and $\pi$ be chosen uniformly at random from $\{0,1\}^n$ and $\mathcal{S}_n$ respectively. The probability that $\PPZMod(\f, y, \pi)$ outputs $z^* \in \Om$ with $\sumd(z^*,T) \geq \frac{k-1}{k+1} \cdot r_{\text{sum}}$ is at least $ \frac{1}{2n} \cdot 2^{-n+n/k}$. 
\end{lemma}
\begin{proof}
    We partition the vertices of $G_\f$ based on the value of $\sumd(\cdot, T)$. For $ 0\leq i \leq tn$, we define $V_i = \{ z \in \Om \mid \sumd(z,T) = i\}$. We define $U_i = \bigcup_{j \geq i} V_j$ and $\overline{U}_{i}  = \Om \setminus U_i$ for $0 \leq i \leq n$. It is easy to see that for any vertex $z \in V_i$, it's neighbors are in $V_j$ for $i-t \leq j \leq i+t$. This is because adding a unit vector $e_k$ to any vector in $\{0,1\}^n$ can increase or decrease its Hamming distance to any other vector by at most $1$, and can hence adding it to $z$ can increase or decrease the value of $\sumd(z, T)$ by at most $t$. We will use $S_i$ to denote the set of edges with exactly one endpoint in $U_i$ (with the other endpoint being in $U_j$ for $i-t \leq j \leq i$). 

    Let $i^* = \lceil \alpha \cdot r_{\text{sum}} \rceil$, where $\alpha=(k-1)/(k+1)$. We will show that $\tau(\f, U_{i^*}) \geq \frac{1}{n} \cdot 2^{-n+n/k}$. From Lemma~\ref{lem:separator} inequality~\ref{eq:sep2}, we get that \[\tau(\f, U_{i^*}) \geq 2^{-n+n/k} \cdot |U_{i^*}|^{1-1/k} \cdot 2^{-\frac{|S_{i^*}|}{k|U_{i^*}|}}.\]
    
    \paragraph{Upper bounding $|S_{i^*}|/|U_{i^*}|$: }The next step is upper bound $|S_{i^*}|$. To do so, for any vertex in $U_{i^*}$, we upper bound the number of vertices in $\overline{U_{i^*}}$ it is adjacent to. We need to only consider vertices in the sets $V_{i^*}, V_{i^*+1}, \cdots, V_{i^*+t-1}$. Consider a vertex $z \in V_{i^*+l-1}$, for $1 \leq l \leq t$. For each $z' \in  \overline{U}_{i^*}$ that $z$ is adjacent to, there exists $m \in [n]$ such that $z' = z \oplus e_m$. Because $\sum_{y \in T} |e_m \oplus z \oplus y|=\sumd(z', T) \leq i^*-1$, and $\sum_{y \in T} |z \oplus y|=\sumd(z,T)=i^*+l-1$, as $z \in V_{i^*+l-1}$, we obtain the following condition on $m$. 

\begin{equation} \label{eqn:lowerbound}
    \sum_{y \in T}|e_m \oplus z \oplus y| \leq \left(\sum_{y \in T} |z \oplus y|\right) -l
\end{equation}

\noindent
Let $T' \subseteq T$ be the subset of $T$ of size $t'$ defined to be $\{y \in T \mid |e_m \oplus z \oplus y| = |z \oplus y| -1\}= \{y \in T \mid (z \oplus y)_m=1\}$. Hence, 
\begin{equation*}
    \sum_{y \in T}|e_m \oplus z \oplus y|= i^*+l-1+t-2t'\leq i^*-1 \; ,
\end{equation*}
which implies that $t' \geq (t+l)/2$. As $(z \oplus y)_m = 1$ for every $y \in T'$, and $\sum_{y \in T'}|z \oplus y| \leq \sum_{y \in T}|z \oplus y|=i^*+l-1$, this implies that there are at most $\frac{2(i^*+l-1)}{t+l}$ possible values $m$ can take in $[n]$, which implies that the number of vertices in $\overline{U_{i^*}}$ that $z$ is adjacent to is at most $\frac{2(i^*+l-1)}{t+l}$, which means we can upper bound $S_{i^*}$. 
\begin{equation*}
    |S_{i^*}| \leq \sum_{l=1}^{t} \frac{2(i^*+l-1)}{t+l} |V_{i^*+l-1}|
\end{equation*}
This implies the following upper bound on $|S_{i^*}|/|U_{i^*}|$. 
\begin{equation}
    \frac{|S_{i^*}|}{|U_{i^*}|} \leq \frac{\sum_{l=1}^{t}\frac{2(i^*+l-1)}{t+l} |V_{i^*+l-1}|}{|U_{i^*}|} \leq \frac{\sum_{l=1}^{t}\frac{2(i^*+l-1)}{t+l} |V_{i^*+l-1}|}{\sum_{l=1}^{t}|V_{i^*+l-1}|} \leq 2 \cdot \max\left\{ \frac{i^*}{t+1}, \frac{i^*+t-1}{2t} \right\}
\end{equation}
\noindent
If the maximum is achieved by $\frac{i^*+t-1}{2t}$, this implies that $i^* \leq t+1$, and $2\cdot \max\left\{ \frac{i^*}{t+1}, \frac{i^*+t-1}{2t} \right\} \leq 2$. This in turn, implies that $$\tau(\f, U_{i^*}) \geq 2^{-n+n/k}|U_{i^*}|^{1-1/k}2^{-\frac{|S_{i^*}|}{k |U_{i^*}|}} \geq 2^{-n+n/k}|U_{i^*}|^{1-1/k} 2^{-2/k} \geq \frac{1}{2} \cdot 2^{-n+n/k}\;.$$  On the other hand, if the max is achieved by $\frac{i^*}{t+1}$, this implies that 

\begin{equation}
    \tau(\f, U_{i^*}) \geq 2^{-n+n/k} \cdot |U_{i^*}|^{1-1/k}2^{-\frac{2 i^*}{k(t+1)}}
\end{equation}
\paragraph{Lower bounding $|U_{i^*}|$: } We now show that either $|U_{i^*}|^{1-1/k}2^{-\frac{2 i^*}{k(t+1)}} \geq \frac{1}{2n}$, or $\tau(\f, U_j) \geq 2^{-n+n/k}$, for some $i^* < j \leq r_{\text{sum}}$. This implies that $\tau(\f, U_{i^*}) \geq 2^{-n+n/k}$. Assume that
\begin{equation}
    \tau(\f, U_{j}) < 2^{-n(1-1/k)} \text{ for all } i^* < j \leq r_{\text{sum}}
\end{equation}
We show that this implies that $|U_{i^*}|^{1-1/k}2^{-\frac{2 i^*}{k(t+1)}} \geq \frac{1}{2n}$. Using Lemma~\ref{lem:separator} inequality~(\ref{eq:sep1}), 
\begin{equation} \label{eqn:sum:contradiction}
    2^{-n(1-1/k)}\cdot |U_{j}| \cdot 2^{-\left(  \frac{2|E(U_{j})|}{k|U_{j}|}+\frac{|S_{j}|}{k|U_{j}|}\right)}\leq \tau(\f,U_{j}) <2^{-n(1-1/k)}
\end{equation}
\noindent
This implies that
\[|U_{j}| < 2^{ \frac{2|E(U_{j})| + |S_{j}|}{k|U_{j}|}} \]

\noindent
Now note that the set $S_j$ consists of edges with one edge in the set $U_j$. Hence, and because the edges cross at most $t$ levels, this implies that $\frac{2|E(U_{j})| + |S_{j}|}{k|U_{j}|} \leq \frac{2|E(U_{j-t})|}{k|U_{j-t}|}$. Further, we can use the edge isoperimetric inequality in hypercubes which states that $|E(U_{j-t})| \leq \frac{1}{2} \cdot |U_{j-t}|\log(|U_{j-t}|)$ to show that
\begin{equation*}
    |U_{j-t}| \log(|U_{j-t}|) > k \cdot |U_{j}| \log(|U_{j}|) \text{ for all }i^* < j \leq r_{\text{sum}}
\end{equation*}
As $U_{r_{\text{sum}}}$ is non-empty, and $S_{r_{\text{sum}}}$ is also non-empty (if not, $\tau(\f, U_{r_{\text{sum}}}) \geq 2^{n-n/k}$), this implies that $|U_{r_{\text{sum}}-t}| \log(|U_{r_{\text{sum}}-t}|) \geq 2$. As $i^* \geq r_{\text{sum}} - \lfloor \frac{r_{\text{sum}}-i^*}{t} \rfloor \cdot t$, this implies that
\[ |U_{i^*}| \log(|U_{i^*}|) > 2 k^{\lfloor\frac{r_{\text{sum}}-i^*}{t}\rfloor-1} \]
\noindent
This implies that for $k\geq 3$, $|U_{i^*}| \geq 2^{\lfloor\frac{r_{\text{sum}}-i^*}{t}\rfloor}$ and because $\log(|U_{i^*}|) \leq n$, for $k=2$, $|U_{i^*}| \geq \frac{1}{n} \cdot 2^{\lfloor\frac{r_{\text{sum}}-i^*}{t}\rfloor}$. This implies that
\[ |U_{i^*}|^{1-1/k}2^{-\frac{2 i^*}{k(t+1)}}\geq \frac{1}{n} \cdot 2^{\left(1-1/k\right)\lfloor\frac{r_{\text{sum}}-i^*}{t}\rfloor-\frac{2 i^*}{k(t+1)}} \geq 2^{-\left(1-1/k\right)} \cdot \frac{1}{n} \cdot 2^{\frac{(k-1)r_{\text{sum}}-(k+1)i^*}{kt}} \geq \frac{1}{2n}\]
\end{proof}
\noindent

\subsection{Algorithmic Implications: farthest point oracles} \label{sec:alg}

We now use these dispersion properties to define farthest point oracles for the $\D$, $\PD$ and $\SPD$ problems. To begin with, we show that we can use the PPZ algorithm to design an \emph{approximate farthest point oracle}. An approximate farthest point oracle takes as input a $k$-CNF formula $\f$, an assignment $z$, and outputs a satisfying assignment $z^*$ that is approximately the farthest satisfying assignment for $\f$ from $z$. 
\begin{lemma} \label{lem:PPZFarthest}
     Let $\f$ be a $k$-CNF formula over $n$ variables and $n^{O(1)}$ clauses and $z \in \Q{n}$ be any assignment to $\f$. If $\f$ is satisfiable, there exists an algorithm that in time $O^*(2^{n-n/k})$ that outputs $z^* \in \Om$, with $d_H(z,z^*) \geq \left(1-\frac{1}{k}\right) \max_{z' \in \Om} d_H(z, z')$ with probability at least $1-2^{-2n}$.
 \end{lemma}
 \begin{proof}
    Consider the following algorithm:

    \begin{algorithm}[H] \label{alg:PPZFarthest}
 \KwIn{A $k$-CNF formula $\f$,$z \in \Q{n}$}
 \KwOut{$z^* \in \Om^s$ with $d_H(z,z^*) \geq \left(1 - \frac{1}{k}\right) \max_{z' \in \Om} d_H(z, z')$ if $\f$ is satisfiable, $\perp$ otherwise}
 Set $z^*=\perp, D=0$. \\
 
 \SetKwFor{RepeatTimes}{repeat}{times:}{endfor}
    \RepeatTimes{$4n^2\cdot 2^{n-n/k}$}{
    Sample $y \in \Q{n}, \pi \in \mathcal{S}_n$ independently and uniformly at random\; 
    $u:=\PPZMod(\f, y, \pi)$ \;
    \If{$u$ satisfies $\f$ and $d_H(z,u) > D$}{
        $z^* \gets u, D \gets d_H(z,u)$.
    }
    }
    Output $z^*$
 \caption{\textsf{PPZ-Farthest}}
 \end{algorithm}
 \Cref{lem:anchor:sum} implies that with probability at least $\frac{1}{2n} 2^{n-n/k}$, $\PPZMod(\f, y,\pi)$ outputs $z^* \in \Om$ with $d_H(z, z^*) \geq \left(1-\frac{1}{k}\right) \max_{z' \in \Om}d_H(z,z')$. Hence, the probability that in $2n^2\cdot 2^{n-n/k}$ iterations of $\PPZMod$, the algorithm outputs such a $z^*$ is at least $1-\left(1- \frac{1}{2n} 2^{n-n/k} \right)^{4n^2 \cdot 2^{-n+n/k} }\geq 1-e^{-2n}$

 \end{proof}
 
Next, we can define a farthest point oracle for $\sumd$. 
 \begin{lemma} \label{lem:PPZFarthestSum}
     Let $\f$ be a $k$-CNF formula over $n$ variables and $n^{O(1)}$ clauses and $S\subseteq \{0,1\}^n$ be a multiset of size $s$. There exists an algorithm running in time $O^*(s \cdot 2^{n-n/k})$ that, if $\f$ is satisfiable, outputs $z^* \in \Om$, with $\sumd(S,z^*) \geq \left(\frac{k-1}{k+1}\right) \max_{z \in \Om} \sumd(z, S)$ with probability $1-2^{-2n}$.
 \end{lemma}

 \begin{proof}
    Consider the following algorithm:
    
\begin{algorithm}[H] \label{alg:PPZFarthestSum}
 \KwIn{A $k$-CNF formula $\f$,$S \subseteq \{0,1\}^n, |S|=s$}
 \KwOut{$z^* \in \Om^s$ with $\sumd(S,z^*) \geq \left(\frac{k-1}{k+1}\right) \max_{z \in \Om} \sumd(z, S)$ if $\f$ is satisfiable, $\perp$ otherwise}
 Set $z^*=\perp, D=0$. \\
 \SetKwFor{RepeatTimes}{repeat}{times:}{endfor}
    \RepeatTimes{$4n^2\cdot 2^{n-n/k}$}{
    Sample $y \in \Q{n}, \pi \in \mathcal{S}_n$ independently and uniformly at random\; 
    $u:=\PPZMod(\f, y, \pi)$ \;
    \If{$u$ satisfies $\f$ and $\sumd(S,u) \geq D$}{
        $z^* \gets u, D \gets \sumd(S,u)$.
    }
    }
    Output $z^*$
 \caption{\textsf{PPZ-Farthest-Sum}}
 \end{algorithm}

    In the $i$-th iteration in the loop of the algorithm, let $y_i, \pi_i$ be the sampled assignment and permutation respectively and let $u_i:= \PPZMod(\f, y_i, \pi_i)$. By \Cref{lem:anchor:sum}, for each $i$, $u_i \in \Om$, and $\sumd(u_i, S) \geq \frac{k-1}{k+1} \cdot \max_{z \in \Om} \sumd(z, S)$ with probability at least $ \frac{1}{2n} \cdot 2^{-n+n/k}$. Because $y, \pi$ in each iteration are sampled independently, the probability that there exists $i \in [4n^2\cdot 2^{n-n/k}]$ such that $\sumd(u_i, S) \geq \frac{k-1}{k+1} \cdot \max_{z \in \Om} \sumd(z, S)$ is at least $1-\left(1-\frac{1}{2n} \cdot 2^{-n+n/k}\right)^{4n^2\cdot 2^{n-n/k}} \geq 1- e^{-4n^2\cdot 2^{n-n/k} \cdot \frac{1}{2n} \cdot 2^{-n+n/k}} = 1-e^{-2n}$. Hence, with probability at least $1-e^{-2n}$, \textsf{PPZ-Farthest-Sum}$(\f,S)$ outputs $z^* \in \Om$, with $\sumd(S,z^*) \geq \left(\frac{k-1}{k+1}\right) \max_{z \in \Om} \sumd(z, S)$. The running time bound follows from the fact that the algorithm contains $4n^2 \cdot 2^{n-n/k}$ iterations, and each iteration takes $s \cdot n^{O(1)}$ time (to compute $\sumd$ and to run \PPZMod).   
 \end{proof}
 \medskip 
 Next, we give a farthest point oracle for $\mind$.
 \begin{lemma} \label{lem:PPZFarthestMin}
     Let $\f$ be a $k$-CNF formula over $n$ variables and $n^{O(1)}$ clauses and $S \subseteq \set{0,1}^n$ be a set of size $s$. There exists an algorithm running in time $O^*(s^2 \cdot 2^{n-n/k})$ that, if $\f$ is satisfiable, there $\textsf{PPZ-Farthest-Min}(\f,s)$ outputs $z^* \in \Om$, with $\mind(S,z^*) \geq \left(1-\frac{1}{kH^{-1}(1-1/k)}\right) \max_{z \in \Om} \mind(z, S)$ with probability at least $1-2^{-2n}$. 
 \end{lemma}
 
 \begin{proof}
    Consider the following algorithm: 
    
\begin{algorithm}[H]
 \label{alg:PPZFarthestMin}
 \KwIn{A $k$-CNF formula $\f$,$S \subseteq \{0,1\}^n, |S|=s, r \in [n]$}
 \KwOut{$z^* \in \Om$, with $\mind(S,z^*) \geq \left(1-\frac{1}{kH^{-1}(1-1/k)}\right) \max_{z \in \Om} \mind(z, S)$ if $\f$ is satisfiable, $\perp$ otherwise.}
 \SetKwFor{RepeatTimes}{repeat}{times:}{endfor}
 Set $z^*=\perp, D=0$. \\
 Let $R$ be the largest $r \in [n]$ such that $\sum_{i = 0}^r\binom{n}{i} \leq 2^{n-n/k}$. \\
 \For{$z \in S$}{
    \For{$u \in \{0,1\}^n: d_H(u,z) \leq R$}{
    \If{$\mind(u,S) > D$ and $u$ satisfies $\f$}{
        $z^* \gets u, D \gets \mind(u,S)$
    }
    }
 }
    \RepeatTimes{$4n^2\cdot 2^{n-n/k}$}{
    Sample $y \in \Q{n}, \pi \in \mathcal{S}_n$ independently and uniformly at random\; 
    $u:=\PPZMod(\f, y, \pi)$ \;
    \If{$u$ satisfies $\f$ and $\mind(S,u) > D$}{
        $z^* \gets u, D \gets \mind(u,S)$
    }
    }
 \caption{\textsf{PPZ-Farthest-Min}}
 \end{algorithm}
For $0 \leq x \leq \frac{1}{2}$, let $H(x):= - x \log(x) - (1-x) \log(x)$. And for $0 \leq y \leq 1$, we define $H^{-1}(y)$ to be the unique $0\leq x \leq \frac{1}{2}$ such that $H(x)=y$. It is known that for any $r$, $\sum_{j=0}^r \binom{n}{j} \leq 2^{n H(r/n)}$.

     Suppose that there exists $z_0 \in \Om$ such that  $\mind(z_0,S)=r$. This implies that there exists $z \in S$ such that $d_H(z,z_0)=r$. If $r \leq R$, this implies that the exhaustive search in the hamming sphere of radius $r$ around each $z$ will find $z_0$. 

     Next, we consider the case that $r \geq R+1$. Firstly, because $R$ is the largest $r \in [n]$ such that $\sum_{i = 0}^r\binom{n}{i} \leq 2^{n-n/k}$, this implies that $\sum_{i = 0}^{R+1}\binom{n}{i} > 2^{n-n/k}$. Using the fact that $\sum_{i = 0}^{R+1}\binom{n}{i} \leq 2^{n H((R+1)/n)}$, and the definition of $H^{-1}$, this implies that $R+1 \geq n \cdot H^{-1}(1-1/k)$. \Cref{lem:anchor:diam} implies that with $y, \pi$ chosen uniformly at random and independently from $\{0,1\}^n$ and $\mathcal{S}_n$, $\PPZMod(\f, y, \pi)$ outputs $z^* \in \Om$ with $d_H(z_0,z^*) \leq \frac{n}{k}$ with probability at least $\frac{1}{2n} \cdot 2^{-n+n/k}$. The triangle inequality then implies that $\mind(z^*,S) \geq r-n/k$. Further, because $r \geq R+1 \geq H^{-1}(1-1/k) \cdot n$, this implies that $n \leq \frac{r}{H^{-1}(1-1/k)}$. Hence, $\mind(z^*,S) \geq \left(1-\frac{1}{kH^{-1}(1-1/k)}\right) r$ . Hence, repeating this $4n^2 \cdot 2^{n-n/k}$ times ensures that with probability $1-2^{-2n}$, the algorithm outputs $z^* \in \Om$ such that $\mind(z^*, S) \geq \left(1-\frac{1}{kH^{-1}(1-1/k)}\right) \cdot \max_{z \in \Om} \mind(z, S)$. The running time bound follows from the fact that the algorithm uses the $\PPZMod$ subroutine $4n^2 \cdot 2^{n-n/k}$ times, and computes the function $\mind(\cdot, \cdot)$ at most $O^*\left(s \cdot 2^{n-n/k}\right)$ times.
 \end{proof}
 \subsection{PPZ-based algorithms for dispersion: Proofs of \Cref{thm:ppz-for-dia}, \Cref{thm:ppz-for-sumdisp} and \Cref{thm:ppz-for-mindisp}} \label{sec:finalppz}  
 
 \subsubsection*{Proof of \Cref{thm:ppz-for-dia}}\Cref{lem:PPZFarthest} implies that the algorithm $\textsf{PPZ-Farthest}$ behaves like a $(1-1/k)$-approximate farthest point oracle for $k$-SAT that runs in time $O^*(2^{n-n/k})$. That is, it takes as input a $k$-CNF formula $\f$ and $z \in \Q{n}$, and with probability $1-2^{-2n}$, outputs $z^* \in \Om$ such that $d_H(z, z^*) \geq (1-1/k) \cdot \max_{z' \in \Om} d_H(z,z')$. Hence, we can use the following procedure to output a $\frac{1}{2} \left(1-1/k\right)$ approximation to $\f$: Use the PPZ algorithm to find one satisfying assignment $z_1^*$ to $\f$, and then output $z_2^*=\textsf{PPZ-Farthest}(\f, z_1^*)$. The triangle inequality then implies that $z_1^*$ and $z_2^*$, will satisfy $d_H(z_1^*, z_2^*) \geq \frac{1}{2} (1-1/k) \cdot \D(\f)$.
 \subsubsection*{Proof of \Cref{thm:ppz-for-sumdisp}}
\begin{restatable}{lemma}{sumdispersion} \label{lem:sumdispersion}
    Suppose there exists a $1-\delta$-approximate farthest point oracle, $\mathcal{O}$ that takes a $k$-CNF formula $\f$ and a multi-set $S \subseteq \set{0,1}^n$ and with probability $1-2^{-2n}$, outputs $z^* \in \Om$ such that $\sumd(S, z^*) \geq (1-\delta) \cdot \max_{z' \in \Om} \sumd(S, z')$. Then, there exists an algorithm taking $\f$ and $s$ as input that uses $s^3 n$ calls to $\mathcal{O}$ (and an additional $s^4 n^{O(1)}$ overhead) that outputs a multi-set $S^* \subseteq \Om$ with $\SPD(S^*) \geq \max \{\frac{1}{2}(1-\delta) ,\frac{(1-\delta)(s-1)}{(1+\delta)s+(1-\delta)} \}\cdot \opts(\f,s)$ with probability $1-o(1)$.
\end{restatable}
\begin{proof}
    We defer the proof to \Cref{sec:dispersion}. 
\end{proof}

We note that \Cref{lem:PPZFarthestSum} implies that the algorithm $\textsf{PPZ-Farthest-Sum}$ is a $1-\delta$ approximate farthest point oracle, as defined in \Cref{lem:sumdispersion}, for $\delta=\frac{2}{k+1}$. Hence, we can use $\textsf{PPZ-Farthest-Sum}$ as a black box in the algorithm defined by \Cref{lem:sumdispersion}. This completes the proof of \Cref{thm:ppz-for-sumdisp}.
\subsubsection*{Proof of \Cref{thm:ppz-for-mindisp}} 
\begin{restatable}{lemma}{mindispersion}
\label{lem:mindispersion}
    Suppose there exists a $1-\delta$-approximate farthest point oracle, $\mathcal{O}$ that takes a $k$-CNF formula $\f$ and a set $S \subseteq \set{0,1}^n$ as input and with probability $1-2^{-2n}$, outputs $z^* \in \Om$ such that $\mind(S, z^*) \geq (1-\delta) \cdot \max_{z' \in \Om} \mind(S, z')$. Then, there exists an algorithm taking $\f$ and $s$ as input that uses $s$ calls to $\mathcal{O}$ (and an additional $sn^{O(1)}$ overhead) that outputs a set $S^* \subseteq \Om$ with $\PD(S^*) \geq \frac{1}{2} (1-\delta)\cdot \optm(\f,s)$ with probability $1-o(1)$.
\end{restatable}
\begin{proof}
    We defer the proof to \Cref{sec:dispersion}.
\end{proof}
We note that \Cref{lem:PPZFarthestMin} implies that the algorithm $\textsf{PPZ-Farthest-Min}$ is a $1-\delta$ approximate farthest point oracle as defined in~\Cref{lem:mindispersion}, for $\delta=\frac{1}{k H^{-1}(1-1/k)}$. Hence, we can use $\textsf{PPZ-Farthest-Min}$ as a black box in the algorithm defined by~\Cref{lem:mindispersion}. This completes the proof of \Cref{thm:ppz-for-mindisp}. 

%% file: schoning-new.tex
\section{From approximate local search to dispersion -- \sch's algorithm}\label{sec:sch}

In this section we prove a generalization of \Cref{thm:sch-for-dia-fixedapprox} and we state and prove theorems with the same running time guarantees (up to a factor polynomial in $s$) to approximate $\optm(\f,s),\optm(\f,s, \geq W), \optm(\f,s, \leq W)$ as well as $\opts(\f,s)$. We note that the algorithm for $\optm(\f,s)$ follows as special cases of the algorithms for $\optm(\f,s, \geq W), \optm(\f,s, \leq W)$. 

\medskip \noindent
To start with, we define the quantity $\tau(\delta,k, n)$ to be $\frac{2^n (k-1)^{R}}{\binom{n}{R}} \text{ where } R=\floor{\frac{\delta n}{2 \left(2+\delta +\frac{2}{k-2}\right)}}$, for each $\delta \in \left( 0, \min \left\{1, \frac{4(k-1)}{(k-2)^2} \right\}\right]$. From now on, we assume that $k \geq 3$ unless stated otherwise. 

\begin{restatable}[\sch for \D: Generalization of \Cref{thm:sch-for-dia-fixedapprox}]{theorem}{schdiam}
    \label{thm:sch-for-dia}
    Let $\f$ be a $k$-CNF formula on $n$ variables. For each $0 < \delta \leq \min\{1, \frac{4(k-1)}{(k-2)^2} \}$, there exists an algorithm taking $\f$ as input and running in time $O^*\left( \tau(\delta, k, n)\right)$ that outputs $z_1^*, z_2^* \in \Om$ such that $d_H(z_1^*, z_2^*) \geq \frac{1}{2}\cdot \left(1-\delta\right) \D(\f)$, if $\f$ is satisfiable.

\end{restatable}

To make the above result more concrete, we first observe that we can define $a_{k,\delta}$, such that $\tau(\delta,k,n)=O^*(a_{k,\delta}^n)$. Now, for $k=7$ and $k=4$, we plot $a_{k,\delta}$ as a function of $\delta$ and compare it with what the PPZ algorithm achieves. Hence, this algorithm provides a smooth trade-off between the approximation factor (i.e., $(1-\delta)$) and running time. We note that for $k=7$, we can achieve the \sch running time for a non-trivial value of $\delta$, but for $k=4$, we cannot do so, even for $\delta$ very close to $1$. Note that this algorithm still achieves non-trivial savings over a brute force search for all values of $\delta$, and in particular, it can be faster than the PPZ algorithm (albeit with a worse approximation factor). 
\begin{figure}[ht]
    \centering
    \includegraphics[scale=0.5]{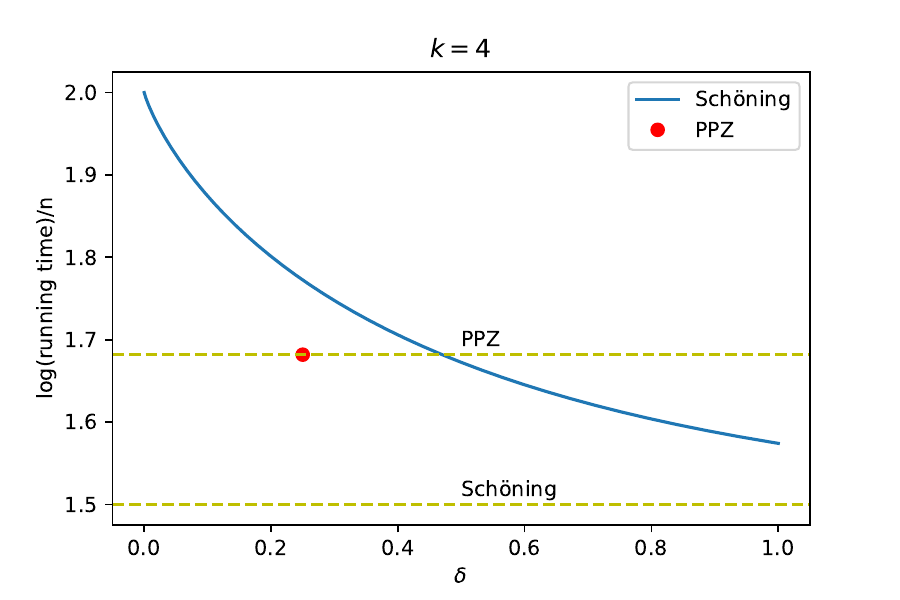}
    \includegraphics[scale=0.5]{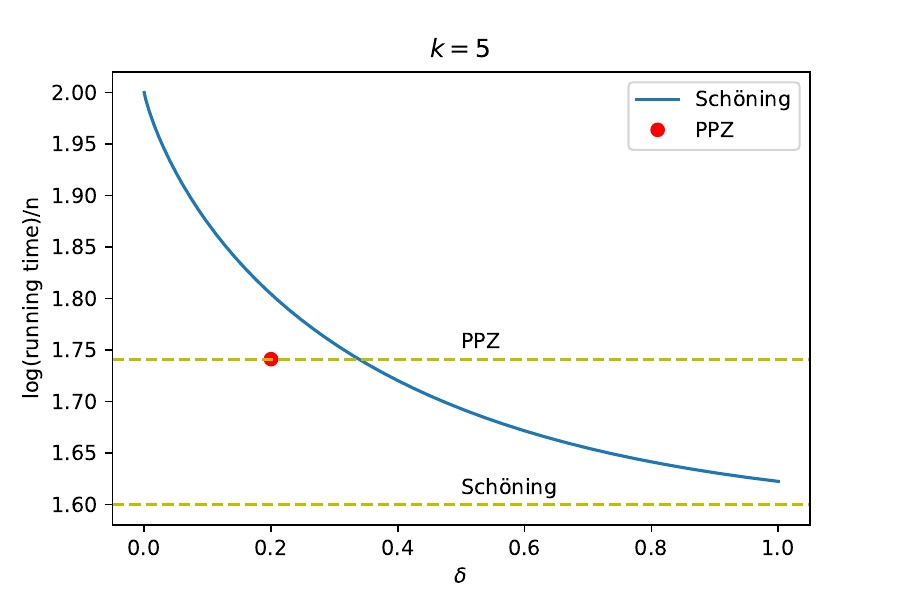}
    \includegraphics[scale=0.5]{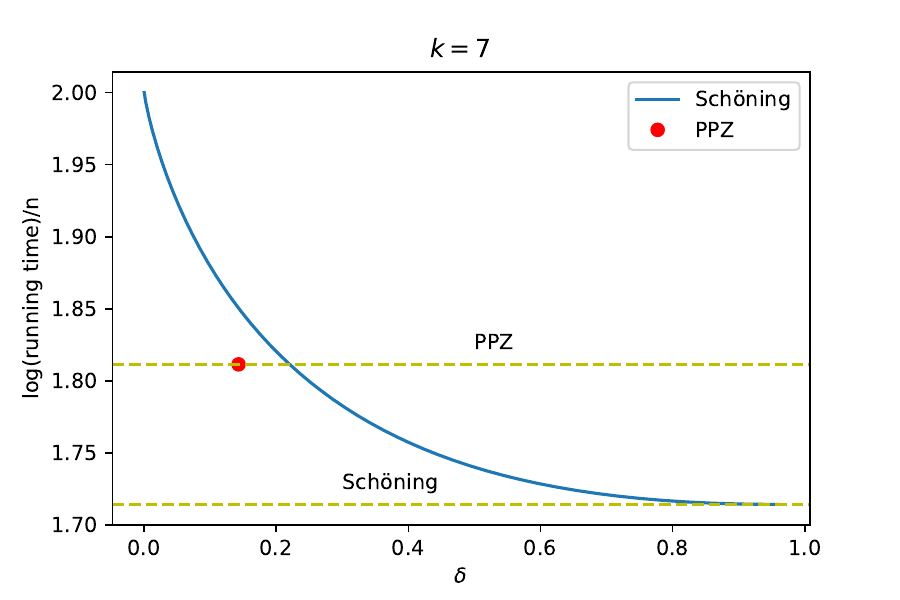}
    \includegraphics[scale=0.5]{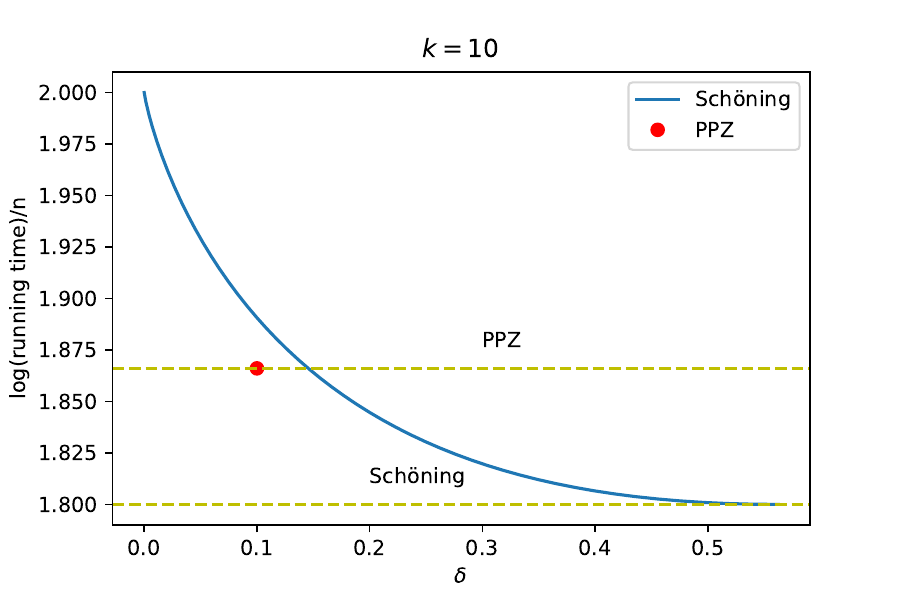}
    \caption{Plot of $a_{k,\delta}$ with respect to $\delta$, with the PPZ running time and approximation factor and \sch running time for comparison, for different values of $k$.}
    \label{fig:enter-label}
\end{figure}

\begin{remark}\label{rem:sch:2}
    When $k \geq 7$, we can use $\delta=\frac{4(k-1)}{(k-2)^2}$ in this algorithm to get a running time of $O^*\left(\left( 2 - \frac{2}{k} \right)^n\right)$ that matches the run-time of \sch's algorithm for finding one satisfying assignment. Thus, \Cref{thm:sch-for-dia} is a generalization of \Cref{thm:sch-for-dia-fixedapprox}. For smaller $k$, while we cannot match the running time of \sch's algorithm, we can still get better than brute force algorithms for the diameter and dispersion problems. 
\end{remark}

\paragraph{Weighted dispersion:} For a $k$-CNF formula $\f$, let $\Omega_{\f,= W}$,$\Omega_{\f,\geq W}$, $\Omega_{\f,\leq W}$ denote the set of satisfying assignments to $\f$ with Hamming weight $W$, at least $W$ and at most $W$ respectively. Let $\optm(\f,s,\geq W)=\max_{S \subseteq \Omega_{\f,\geq W}, |S|=s} \PD(S)$, and \\ $\optm(\f,s,\leq W)=\max_{S \subseteq \Omega_{\f,\leq W}, |S|=s} \PD(S)$.

\begin{theorem}[Weighted dispersion- Full version of \Cref{thm:schheavyeasy}] \label{thm:sch-heavy-full}
    Let $\f$ be a $k$-CNF formula on $n$ variables, $W \in [n]$ and $s \in \IN$. 
    \begin{enumerate}

        \item For each $0 < \delta \leq \min \left \{1, \frac{4(k-1)}{(k-2)^2} \right \}$, there exists an algorithm that takes $\f,W,s$ as input and runs in time $O^*\left(s^3 \cdot \tau(\delta, k, n)\right)$ and outputs a set $S^* \subseteq \Omega_{\f, \geq (1-\delta) W}$ of size $s$ such that $\PD(S^*) \geq \frac{1}{2}\left(1-\delta \right) \optm(\f,s, \geq W)$ with probability $1-o(1)$.

        \item For each $0 < \delta \leq \min \left \{1, \frac{4(k-1)}{(k-2)^2} \right \}$, there exists an algorithm that takes $\f,W,s$ as input and runs in time $O^*\left(s^3 \cdot \tau(\delta, k, n)\right)$ and outputs a set $S^* \subseteq \Omega_{\f, \geq (1+\delta) W}$ of size $s$ such that $\PD(S^*) \geq \frac{1}{2}\left(1-\delta \right) \optm(\f,s, \geq W)$ with probability $1-o(1)$.

    \end{enumerate}
    
\end{theorem}

Note that as a special case, this theorem leads to an algorithm for $\optm(\f,s)$ with the same time bounds and approximation factors. In addition, we show that a slight modification of this algorithm can also be used for $\opts(\f,s)$. 
\begin{restatable}{theorem} {schopts}[\sch approximating $\opts(\f,s)$]
\label{thm:sch-for-sumdisp}
    Let $\f$ be a $k$-CNF formula on $n$ variables and $s \in \IN$. For each $0 < \delta \leq \min \left \{1, \frac{4(k-1)}{(k-2)^2} \right \}$, there exists an algorithm that takes $\f,s$ as input and runs in time $O^*\left(s^3 \cdot \tau(\delta, k, n)\right)$ that outputs, with probability $1-o(1)$, a multi-set $S^* \subseteq \Omega_{\f, \geq (1-\delta) W}$ of size $s$ such that $$\SPD(S^*) \geq \begin{cases}
            \frac{1}{2}\left(1-\delta \right) \opts(\f,s) \text{ if }  s \leq 3+ \floor{ \frac{2 \delta}{1-\delta} } \\
            \frac{1-\delta}{1+\delta}\left(\frac{1-\frac{1}{s}}{1+\frac{1-\delta}{1+\delta} \cdot \frac{1}{s}}\right) \opts(\f,s) \text{ if } s > 3+ \floor{ \frac{2 \delta}{1-\delta} }
        \end{cases}$$ 
    
\end{restatable}
\paragraph{The case of $2$-SAT and other small $k$:} We design different algorithms to handle the case of $2$-SAT, which also outperform the algorithms presented here in some regimes of $\delta$ for larger $k$. For example, for $k=3$, it outperforms the algorithm in~\Cref{thm:sch-for-dia} for all values of $\delta$, and for $k \geq 4$, it outperforms~\Cref{thm:sch-for-dia} for smaller values of $\delta$. This is presented in~\Cref{app:more-sch}.

\paragraph{Proof organization:} We prove the above three theorems in parallel using the following three step procedure. 
\begin{enumerate}
    \item In \Cref{sec:schprelims}, we recall \sch's algorithm and the key observations used to analyse it. 
    \item In \Cref{sec:diamfarthest}, we develop and analyze farthest point oracles for $\D$, $\sumd$ and $\mind$ using \sch's algorithm.
    \item In \Cref{sec:schfinal}, we describe and analyse our algorithms for finding dispersed solutions with respect to $\optm$, completing the proofs of \Cref{thm:sch-for-dia} and \Cref{thm:sch-heavy-full}. Just like for PPZ, these algorithms use the farthest point oracles in the algorithms for dispersion studied by Gonzales~\cite{gonzalez1985clustering}. In \Cref{app:schsum}, we describe and analyse an algorithm for finding dispersed solutions with respect to $\opts$, completing the proof of \Cref{thm:sch-for-sumdisp}. 
    \item In \Cref{app:more-sch}, we describe another algorithm that handles the case of $2$-SAT and $3$-SAT and also outperforms the algorithms described in this section for some regimes of $\delta$ for larger values of $k$.
\end{enumerate}

\subsection{Parameterized local search} \label{sec:schprelims}
\textbf{The Sch\"oning walk.} Sch\"oning's algorithm consists of repeatedly invoking the following procedure, which we call a \emph{Sch\"oning walk}. Formally, a Sch\"oning walk of length one, denoted by $\SW_1(\f,z)$, takes as input a formula $\f$ and an assignment $z\in\set{0,1}^n$, and returns another assignment $z'\in\set{0,1}^n$ constructed as follows: if $z$ is a satisfying assignment, then $z'=z$. Otherwise, let $C$ be a clause in $\f$ that is not satisfied by $z$. Pick one of its $k$ literals uniformly at random and flip its value in $z$, thus obtaining $z'$. For $t\geq 2$, a Sch\"oning walk of length $t$ can be recursively defined as $\SW_t(\f,z) = \SW_1(\f, 
\SW_{t-1}(\f,z))$. We refer to $z$ as the starting point of the Sch\"oning walk of length $t$. 

\medskip \noindent
We note the following key observation about the \sch walk. We refer the reader to \sch's original paper for a proof~\cite{schoning1999probabilistic}. 
\begin{restatable}{lemma}{obssch} \label{obs:schmain}
    For any starting assignment $z \in \set{0,1}^n$, if there exists a satisfying assignment $z^* \in \set{0,1}^n$ such that $d_H(z,z^*) \leq t$, then $\SW_t(\f,z)$ outputs a satisfying assignment with probability at least $k^{-t}$. Furthermore, $\SW_{\ceil{\left(1+2/(k-2)\right)t}}(\f,z)$ outputs a satisfying assignment with probability at least $(k-1)^{-t}$.
\end{restatable}
\begin{remark}
    In \sch's original paper, the statement proved is that $\SW_{3t}(\f,z)$ outputs a satisfying assignment with probability at least $(k-1)^{-t}$. However, looking at the analysis more carefully, we can prove that a shorter \sch walk of length $\left(1+2/(k-2)\right)t$ suffices (for $k=3$, these two quantities are equal). This fact is irrelevant to the performance of the original algorithm, but is helpful for our purpose of finding dispersed satisfying assignments to $\f$. 
\end{remark}

\paragraph{\sch's local search: }\Cref{obs:schmain} gives a \emph{parameterized local search} algorithm for $k$-SAT. Formally for some values $\alpha \geq 1,c > 1$, a local search procedure $\LS_{\alpha,c}$ takes as input a $k$-CNF formula $\f$, a starting assignment $z \in \Q{n}$, and $t \in [n]$, such that if there exists a satisfying assignment $z_0$, with $d_H(z,z_0) \leq t$, then, in time $n^{O(1)} c^t$, $\LS_{\alpha,c}$ outputs a satisfying assignment $z^* \in \Om$, with $d_H(z, z^*) \leq \ceil{\alpha t}$. \footnote{We have defined an ``approximate'' version of local search. The traditional definition does not use $\alpha$}. Hence, there exist two versions of parameterized local search for $k$-SAT. 

\begin{enumerate}
    \item $\LS_{1,k}$: This involves repeating the \sch walk starting at $z$ for $t$ steps $n^{O(1)} \cdot k^t$ times.
    \item $\LS_{\left(1+2/(k-2)\right),k-1}$: This involves repeating the \sch walk of $\ceil{\left(1+2/(k-2)\right)t}$ steps starting at $z$ $n^{O(1)} \cdot (k-1)^t$ times.
\end{enumerate}

Consider the following algorithm for solving $k$-SAT. Given a local search procedure $\LS_{\alpha, c}$, set $t=\floor{\frac{n}{c+1}}$, sample $z \in \Q{n}$ uniformly at random, and run $\LS_{\alpha, c}$ with $z$ and $t$ as input. If there exists a satisfying assignment $z_0$, $z$ will be within distance $t$ of $z_0$ with probability at least $\frac{\binom{n}{t}}{2^n}$. To succeed in finding a satisfying assignment with probability $1-o(1)$, it is sufficient to repeat this procedure $n^{O(1)} \cdot \frac{2^n}{\binom{n}{t}}$ times. The entire algorithm runs in time $n^{O(1)} \cdot \frac{2^n}{\binom{n}{t}c^{-t}}= O^*\left( \left( \frac{2}{1+1/c}\right)^n \right)$. \sch uses $\LS_{\left(1+2/(k-1)\right)t, k-1}$, which gives a running time of $O^*\left( \left(  2\left(1-\frac{1}{k} \right) \right)^n\right)$. We refer the reader to \Cref{schcalc} for a proof of this statement. 

\paragraph{The case of $2$-SAT and other small $k$:} Our algorithms for approximating dispersion use the procedure $\LS_{\left(1+2/(k-2)\right),k-1}$. For the case of small $k$ and small $\delta$, it is useful to use the local search procedure $\LS_{1,k}$ instead. It turns out that this algorithm gives a better trade-off with $\delta$. We present more details in~\Cref{app:more-sch}.\footnote{Before \sch's algorithm for $k$-SAT was discovered, a very similar (polynomial time) algorithm was developed for $2$-SAT by Papadimitriou~\cite{papadimitriou1991selecting}. It picks a starting assignment $z \in \{0,1\}^n$ at random and performs a \sch walk for $O(n^2)$ steps starting at $z$. If the $2$-CNF formula is indeed satisfiabile, this algorithm finds a satisfying assignment with probability $1-o(1)$. \sch's main innovation in extending this algorithm to get a better than brute force algorithm was in restarting the local search process with a new randomly chosen starting assignment after $3n$ steps. However, computing the diameter of a $2$-CNF formula is an NP-complete problem, and \sch's paradigm is useful here as well.}

\subsection{Anchored local search and farthest point oracles}
\label{sec:diamfarthest}

Next, we show that we can carefully control the length of the \sch walk to come up with farthest point oracles. We call this procedure ``anchoring''. This technique is general and can be used with any $\LS_{\alpha, c}$ procedure for a ``subset problem''. We will see more examples in 
\Cref{sec:applications}. 

\begin{lemma} \label{lem:anchor:sch}
     Consider a local search algorithm $\LS_{\alpha, c}$. Then, for every $0 < \delta \leq \frac{2(1+\alpha)}{c-1}$, there exists an algorithm  running in time $\frac{2^n c^{R}}{\binom{n}{R} }$, where $R=\floor{\frac{\delta n}{2(1+\alpha+\delta)}}$, that takes as input $\f$ and $z \in \{0,1\}^n$, and if $\f$ is satisfiable, outputs $z^* \in \Om$ such that $d_H(z^*, z) \geq \left(1-\delta\right) \cdot \max_{z' \in \Om} d_H(z,z')$ with probability at least $1-2^{-n}$. 
\end{lemma}
\begin{proof}
    Consider the following procedure. 
    
\begin{algorithm}[H] \label{alg:ALS}
 \KwIn{A $k$-CNF formula $\f$ over $n$ variables, $z \in \{0,1\}^n$, $r \in [n]$}
 Let $t:=\min\left\{ \floor{\frac{\delta r}{1+\alpha}}, R\right\}$. \\
 Sample a starting point $y$ uniformly~\footnote{We note that it is possible to uniformly sample from $A_{r-t , r+t} (z)$ in polynomial time. First, we pick a radius $x \in \set{r-t,\ldots, \min \set{{r+t,n}}}$ proportional to  the ratio ${n \choose x} / \size{A_{r-t , r+t} (z)}$. We then choose a random permutation in $\pi \in \mathcal{S}_n$ and let $A \subseteq [n]$ be the first $x$ elements of $\pi$. $y$ is obtained by setting $y_i=1$ if and only if $i \in A$. } 
 at random from $ A_{r-t,r+t} (z)$, where $A_{r-t, r+t}(z):=\{ x \in \{0,1\}^n \mid r-t \leq d_H(x,z) \leq r+t \}$. \\
Output $\LS_{\alpha,c}(\f,y,t)$ 
 \caption{$\ALS_{\alpha, c , \delta}$}
\end{algorithm}
\vspace{2mm}
Suppose there exists a satisfying assignment $z_0 \in \Om$ such that $d_H(z_0, z)=r$. Let $y \in A_{r-t , r+t} (z)$ be the starting point sampled by $\ALS_{\alpha,\delta}(\f,z,r)$. Consider any starting point $y$ that is within distance $t$ from $z_0$. Because $\LS_{\alpha}(\f,y,t)$ outputs a satisfying assignment $z^*$ whose distance is at most $\alpha t $ from $y$, the triangle inequality implies that $d_H(z^*, z_0) \leq (1+ \alpha)t \leq \delta r$. Because $y$ is sampled uniformly at random from $A_{r-t , r+t} (z)$ and $t \leq \frac{\delta r}{1+\alpha}$, the probability that $y$ is within distance $\leq t$ from $z_0$ is at least $\frac{\binom{n}{t}}{|A_{r-t , r+t}(z)|}$, which implies that with probability at least $\frac{\binom{n}{t} (1-2^{-n})}{|A_{r-t , r+t}(z)|}$, $\ALS_{\alpha, c ,\delta}(z,r)$ outputs $z^* \in \Om$ such that $d_H(z_0, z^*) \leq \delta r$, and using the triangle inequality again,  implies that $d_H(z, z^*) \geq \left(1-\delta \right) \cdot  r$

This implies that repeating the procedure $\ALS_{\alpha, c ,\delta}(\f,z,r)$ $n^{O(1)} \cdot  \frac{|A_{r-t , r+t}(z)|}{\binom{n}{t} }$ times is enough to output a satisfying assignment $z^*$ such that $d_H(z, z^*) \geq \left(1-\delta \right) \cdot  r$ with probability at least $1-2^{-n}$. Iterating over all $r \in [n]$ (and returning the $z^* \in \Om$ found with maximum Hamming distance from $z$) implies the existence of an algorithm that outputs $z^* \in \Om$ with $d_H(z,z^*) \geq (1-\delta) \max_{z' \in \Om} d_H(z,z')$ in time $$n^{O(1)} \cdot \sum_{r \in [n]} \frac{|A_{r-t , r+t}(z)| c^t}{\binom{n}{t}  } = n^{O(1)} \max_{r \in [n]} \frac{|A_{r-t , r+t}(z)| c^t}{\binom{n}{t}  } \; .$$

\medskip \noindent
     The next step is to upper bound the quantity $\tau(r,n):=\frac{|A_{r-t , r+t}(z)| c^t}{\binom{n}{t}}$. We start by upper bounding $|A_{r-t , r+t}(z)|$. Because $A_{r-t , r+t}(z)$ is a union of Hamming spheres around $z$, we can upper bound it as follows. 
     \begin{equation*}
         |A_{r-t , r+t}(z)| \leq \begin{cases}
             n \cdot \binom{n}{r+t}, \text{ if } r+t < \frac{n}{2} \\
             2^n, \text{ if } r-t \leq \frac{n}{2} \leq r+t \\
             n \cdot \binom{n}{r-t}, \text{ if } r-t > \frac{n}{2}
         \end{cases}
     \end{equation*}
    Recall that $t=\min\left\{ \floor{\frac{\delta r}{1+\alpha}}, R\right\}$. This implies that when $r \leq \frac{n}{2\left(1+\frac{\delta}{1+\alpha}\right)}$, the corresponding value of $t$ is $\floor{\frac{\delta r}{1+\alpha}}$, and when $r \geq \frac{n}{2\left(1+\frac{\delta}{1+\alpha}\right)}$, the corresponding value of $t$ is $R$. We can now upper bound $\tau(r,n)$ in each regime as follows:
     \begin{equation*}
         \tau(r,n) \leq \begin{cases}
             \frac{\binom{n}{r+\floor{\frac{\delta r}{1+\alpha}}}c^{\floor{\frac{\delta r}{1+\alpha}}}}{\binom{n}{\floor{\frac{\delta r}{1+\alpha}}} }  \text{ if } r < \frac{n}{2\left(1+\frac{\delta}{1+\alpha}\right)} \\
             \noalign{\vskip9pt}
             \frac{2^n c^{R}}{\binom{n}{R} } \text{ if }\frac{(1+\alpha)n}{2(1+\alpha +\delta)} \leq  r \leq \frac{n}{2}+R\\
             \noalign{\vskip9pt}
             \frac{\binom{n}{r-R}c^{R}}{\binom{n}{R} } \text{ if } r > \frac{n}{2}+R
         \end{cases}
     \end{equation*}
     Notice that when $r \geq \frac{(1+\alpha)n}{2(1+\alpha +\delta)}$, $\tau(r,n)$ is at most $\frac{2^n c^{R}}{\binom{n}{R} }$.

    \medskip \noindent
     We can now substitute $\beta=\frac{\delta}{1+\alpha}$ in \Cref{lem:derivative} (stated and proved below) to show that when $0 \leq r \leq \frac{n}{2\left(1+\frac{\delta}{1+\alpha}\right)}$, $\tau(r,n)$ is upper bounded by $\frac{2^n c^{R}}{\binom{n}{R} }$, completing the proof.
\end{proof} 
\begin{lemma} \label{lem:derivative}
    Let $0 < \beta \leq \frac{2}{c-1}$. Then,
    \[ \max_{r \in \{0,1,\dots, \floor{\frac{n}{2(1+\beta)}}\}} \frac{ \binom{n}{r+\floor{\beta r}} c^{\floor{\beta r}}}{    \binom{n}{\floor{\beta r}}} \leq n^{O(1)} \cdot  \frac{2^n}{\binom{n}{\floor{ \frac{\beta n}{2(1+1 \beta)}}} c^{-\floor{\frac{\beta n}{2+2 \beta}}}}\]
\end{lemma}

In order to prove \Cref{lem:derivative}, we need some observations.

\begin{observation} \label{obs:binomone}
    For integers $n$ and $m \leq n/2$, $\frac{2}{n} \cdot \binom{n}{m+1} \leq \binom{n}{m} \leq \binom{n}{m+1} \; .$
\end{observation}

\begin{observation}[\cite{macwilliams1977theory}] \label{obs:binomapprox}
\begin{equation*}
   \frac{1}{n^{O(1)}} \cdot \left(\mu^{-\mu} (1-\mu)^{\mu-1} \right)^n \leq \binom{n}{\mu n} \leq \left(\mu^{-\mu}  (1-\mu)^{\mu-1} \right)^n
\end{equation*}
\end{observation}

\begin{observation} \label{obs:binomderivative}
    The derivative of the function $f(\mu)=\mu^{-\mu} (1-\mu)^{\mu-1} $ with respect to $\mu$ is $f'(\mu)=f(\mu) \left( \ln{(1-\mu)} - \ln{\mu} \right)$.
\end{observation}

\noindent \textbf{Proof of \Cref{lem:derivative}}
    Let $r= \mu n$. let $f(\mu)=\mu^{-\mu} (1-\mu)^{\mu-1}$ We use \Cref{obs:binomapprox} to show that 
    \[\frac{\binom{n}{r+\floor{\beta r}}}{\binom{n}{\floor{\beta r}} c^{-\floor{\beta r}}} = n^{O(1)} \cdot \left(g(\mu)\right)^n \; , \]
    where $g(\mu)=\frac{f\left((1+\beta) \mu)\right)}{c^{-\beta \mu} f\left(\beta \mu\right)}$. We next show that $g$ is an increasing function of $\mu$, which means that the maximum value of $g(\mu)$ is obtained at $\mu=\frac{1}{2(1+\beta)}$. Using the quotient, product and chain rules for differentiation and \Cref{obs:binomderivative}, we can show that
    \begin{align*}
        g'(\mu)=& \frac{\beta \ln(c) c^{\beta \mu} f(\beta \mu) f((1+\beta) \mu) + (1+\beta) c^{- \beta \mu } f(\beta \mu) f'((1+\beta) \mu)-\beta c^{\beta \mu} f'(\beta \mu) f((1+\beta) \mu)}{f(\beta \mu)^2} \\
        \noalign{\vskip9pt}
        =& \frac{ \splitfrac{\beta \ln(c) c^{\beta \mu} f(\beta \mu) f((1+\beta) \mu) - \beta \ln\left(\frac{1-\beta \mu}{\beta \mu}\right)c^{\beta \mu } f(\beta \mu)f((1+\beta) \mu)}{
        +(1+\beta) \left( \ln \left(\frac{1-(1+\beta)\mu}{(1+\beta)\mu}\right) \right)c^{-\beta \mu} f(\beta \mu) f\left( (1+\beta) \mu \right)} }{f(\beta\mu)^2}\\
        \noalign{\vskip9pt}
        =& g(\mu) \left(\beta \ln(c) - \beta \ln\left(\frac{1-\beta \mu}{\beta \mu}\right) +(1+\beta) \ln \left(\frac{1-(1+\beta)\mu}{(1+\beta)\mu}\right)\right)
    \end{align*}
    Let $h(\mu)=\frac{g'(\mu)}{g(\mu)}$. If we show that the $h(\mu)$ is a decreasing function of $\mu$, when $0 \leq \mu \leq \frac{1}{2(1+\beta)}$, that is enough to show that $h(\mu)\geq 0$ for all $0 \leq \mu \leq \frac{1}{2(1+\beta)}$. We now compute $h'(\mu)$.
    \begin{equation*}
        h'(\mu)= \frac{\beta}{(1-\beta \mu) \mu} - \frac{1+\beta}{(1-(1+\beta)\mu)\mu} \; ,
    \end{equation*}
    which is negative for all $0 < \mu \leq \frac{1}{2(1+\beta)}$. Hence, $g\left(\frac{1}{2(1+\beta)}\right)$ is an upper bound for all $g(\mu)$ for $0 \leq \mu \leq \frac{1}{2(1+\beta)}$.
    
We now generalize this approach to come up with a farthest point oracle for the $\mind$ dispersion measure.
 
\paragraph{Heavy and low weight dispersion:} We now show that this approach can also be used to return dispersed satisfying assignments of large or small Hamming weight. For a $k$-CNF formula $\f$, recall that $\Omega_{\f,= W}$,$\Omega_{\f,\geq W}$, $\Omega_{\f,\leq W}$ denote the set of satisfying assignments to $\f$ with Hamming weight $W$, at least $W$ and at most $W$ respectively. Let $\optm(\f,s,\geq W)=\max_{S \subseteq \Omega_{\f,\geq W}, |S|=s} \PD(S)$, and $\optm(\f,s,\leq W)=\max_{S \subseteq \Omega_{\f,\leq W}, |S|=s} \PD(S)$. 

 \begin{lemma} [Farthest Point Oracle]
 \label{lem:schFarthestHeavy}
    Consider a local search algorithm $\LS_{\alpha, c}$. Then, for every $0 < \delta \leq \frac{2(1+\alpha)}{c-1}$, there exists an algorithm that takes as input a $k$-CNF formula $\f$, a set $S \subseteq \{0,1\}^n$ of size $s$ and $W \in [n]$. If $\Omega_{\f, =W}$ is non-empty, with probability at least $1-2^{-2n}$, it outputs $z^* \in \Om$ such that $(1-\delta) W\leq |z^*| \leq (1+\delta) W$ and $\mind(z^*,S) \geq \max_{z' \in \Omega_{\f, =W}} \mind(z',S)$. The algorithm runs in time $s^2 \cdot n^{O(1)} \cdot \frac{2^n c^{R}}{\binom{n}{R} }$, where $R=\floor{\frac{\delta n}{2(1+\alpha+\delta)}}$.
 \end{lemma}
 \begin{proof}
     Consider the following algorithm.
     
     \begin{algorithm}[H] \label{alg:schFarthestMinHeavy}
 \KwIn{A $k$-CNF formula $\f$, $S \subseteq \{0,1\}^n, |S|=s, W \in [n]$}
 \KwOut{$z^* \in \Om$ with $(1-\delta) W \leq |z^*| \leq (1+\delta) W$ $\mind(S,z^*) \geq \left(1-\delta \right) \max_{z \in \Omega_{\f,=W}} \mind(z, S)$ if $\Omega_{\f,=W}$ is non-empty, $\perp$ otherwise.}
 \SetKwFor{RepeatTimes}{repeat}{times:}{endfor}
 Set $z^*=\perp, D=0$. \\
 \For{$r \in [n]$}{
    \For{$z \in S \bigcup \{\mathbf{0}\}$}{
        Let $t:=\min\left\{ \floor{\frac{\delta r}{1+\alpha}}, R\right\}$ \\
        \RepeatTimes{$n^{O(1)} \cdot |A_{r-t,r+t} (z)|$}{
        $u:=\ALS_{\alpha, c ,\delta}(\f,z,r)$\\
    \If{$u$ satisfies $\f$, $\mind(S,u) > D$, and $(1-\delta) W\leq |u| \leq (1+\delta) W$}{
        $z^* \gets u,D \gets \mind(S,u)$.
    }
    
    }}}
 \caption{\textsf{\sch-Farthest-Weighted}}
\end{algorithm}
Suppose there exists $z_0 \in \Om$ such that $\mind(z_0, S) \geq r$ and $|z_0|= W$. This implies that there exists $z \in S$, such that $d_H(z,z_0)=r$ and $d_H(z',z_0) \geq r$, for all $z' \in S \setminus \{z\}$. First, consider the case that $W \geq r$, and that $\ALS_{\alpha, c ,\delta}(\f, z,r)$ outputs $z^*$, such that $d_H(z^*, z_0) \leq \delta r$. Then, we can use the triangle inequality to show that for all $z' \in S$,
$$d_H(z^*,z') \geq d_H(z',z_0)  - d_H(z*,z_0) \geq r  - \delta r\;,$$ which implies that $\mind(z^*, S) \geq  \left(1  - \delta \right)r$. Further, $W-\delta r \leq |z^*| \leq W+\delta r$, and because $W \geq r$, this implies that $(1-\delta) W\leq |z^*| \leq (1+\delta) W$. Now, note that $\ALS_{\alpha,\delta}(\f, z, r)$ outputs such a $z^*$ if $y$, the starting assignment it samples, is within distance $t$ of $z_0$. Note that $t$, chosen in line 4 of $\ALS_{\alpha, c ,\delta}$ depends on $r$. This event occurs with probability at least $\frac{\binom{n}{t}}{|A_{r-t , r+t}(z)|}$. 

Now, suppose that $r > W$, and $\ALS_{\alpha, c ,\delta}(\f, \mathbf{0},W)$ outputs $z^*$, such that $d_H(z^*, z_0) \leq \delta W$. As in the previous case, the triangle inequality will imply that $d_H(z^*,z') \geq (1-\delta)r$, for all $z' \in S$, and $(1+\delta) W \leq |z^*| \leq (1+\delta) W$, and just as in the previous case, $\ALS_{\alpha,\delta}(\f, \mathbf{0}, r')$ outputs such a $z^*$ if $y$, the starting assignment it samples, $y$ is within distance $t'$ (where $t'$ is the value chosen corresponding to $r'$ in $\ALS_{\alpha, c ,\delta}$) of $z_0$, which happens with probability at least $\frac{\binom{n}{t}}{|A_{r-t , r+t}(z)|}$.

The rest of the proof (bounding the running time) is identical the proof of \cref{lem:anchor:sch}, with the dependence on $s$ coming from the number of nested loops.
\end{proof}

\subsection{\sch-based algorithms for dispersion: Proofs of \Cref{thm:sch-for-dia} and \Cref{thm:sch-heavy-full}} \label{sec:schfinal}
\subsubsection*{Proof of \Cref{thm:sch-for-dia}: Diameter}
The proof of \Cref{thm:sch-for-dia} is similar to that of \Cref{thm:ppz-for-dia}. \Cref{lem:anchor:sch} implies  that there exists a $1-\delta$-approximate farthest point oracle that takes as input a $k$-CNF formula $\f$ and $z \in \Q{n}$, and with probability $1-2^{-2n}$, outputs $z^* \in \Om$ such that $d_H(z,z^*) \geq (1-\delta) \max_{z' \in \Om} d_H(z,z')$. We first use \sch's algorithm for $k$-SAT to find one satisfying assignment $z_1^*$ to $\f$. Let $z_2^*$ be the satisfying assignment output by the $1-\delta$-approximate farthest point oracle with $z^*_1$ and $\f$ as input. The triangle inequality then implies that $z_1^*$ and $z_2^*$, will satisfy $d_H(z_1^*, z_2^*) \geq \frac{1}{2} (1-\delta)$. The running time guarantee for the first and second algorithms come from using $c=k-1, \alpha=1+\frac{2}{k-2}$ that we described in \Cref{sec:schprelims}. 

\subsubsection*{Proof of \Cref{thm:sch-heavy-full}: Weighted min-dispersion}
Firstly, it is easy to observe that for any set $S \subseteq \Q{n}$ and $W \in [n]$, we can use the algorithm in \Cref{lem:schFarthestHeavy} to output $z^* \in \Omega_{\f, \geq (1-\delta )W}$ such that $\mind( z^*, S) \geq (1-\delta) \max_{z \in \Omega_{\f, \geq W}} \mind(z, S)$. We do so by iterating over all $W' \in \{W, W+1, \dots, n\}$, using $\textsf{\sch-Farthest-Weighted}(\f, S, W')$, and returning $z^*$ with maximum value of $\mind(z^*, S)$. This can be used along with \Cref{lem:mindispersion} to prove \Cref{thm:sch-heavy-full}.
\subsubsection*{Proof of \Cref{thm:sch-for-sumdisp}: Sum-dispersion} We refer the reader to \Cref{app:schsum} for the proof. 

%% file: applications-old.tex
\section{Applications and generalisations} \label{sec:applications}

In Sections~\ref{sec:isometric} and~\ref{sec:lfs} we demonstrate that the techniques we developed in \Cref{sec:sch} are fairly general and can be also used to obtain diverse solutions to several NP-complete optimisation problems. Following this, Section~\ref{sec:fast} shows how an improvement in runtime of Sch\"{o}ning's and PPZ algorithms (for finding one solution) can be obtained if $\Om$ has many dispersed solutions. Finally, Section~\ref{sec:csp} shows how to extend our Sch\"{o}ning result to finding diverse solutions to CSPs.

For simplicity, we focus on the $\optm$ diversity measure in this section. It is easy to generalize the results to the $\opts$ diversity measure as well. 

\subsection*{Optimization Problems and Bi-Approximations}
\label{sec:app2}

We show that our techniques can be used in a black-box as well as white-box manner for a broad class of optimization problems called \emph{subset problems}. A subset problem consists of an implicitly defined family $\calf$ of subsets of $[n]$, and the problem is to find $A \in \calf$ of minimum (or maximum) size. We start with describing a framework that captures all these problems. This framework will also help us to abstract the notion of \emph{an isometric reduction}, which we will define formally in \Cref{sec:isometric}.

\paragraph{Implicit set systems: }An implicit set system $\Phi$ is a function that takes a string $I \in \{0,1\}^*$ (called the instance) and outputs an integer $n \in \IN$ and $\mathcal{F}_I \subseteq \Q{n}$, called the feasible set of $\Phi$. Elements in $\{0,1\}^n$ outside $\calf_I$ are called infeasible. Many natural computational problems we consider can be defined using implicit set systems. For an implicit set system $\Phi$, we define the computational problem $\Phi\textsc{-Subset}$.
\begin{problem}[$\Phi\textsc{-Subset}$]
    \textbf{Input:} An instance $I \in \{0,1\}^*$ to $\Phi$. 
    
    \textbf{Output:} $A \in \mathcal{F}_I$, if $\mathcal{F}_I$ is non-empty.
\end{problem}
An example of an implicit set system is one generated by $k$-CNF formulas. If the input instance $I$ is a $k$-CNF formula $\f$ over $n$ variables (using some canonical encoding of formulas as strings), then $\Phi(\f)=\left(n, \Om \right)$ ($\mathcal{F}_I$ is defined to be empty for all other strings for consistency). In this case, the problem $\Phi\textsc{-Subset}$ is $\np$-complete. Other examples of implicit set systems are those generated by graphs, where the input string $I$ encodes a graph $G$, $n_I=|V(G)|$, and $\mathcal{F}_I$ is the set of all independent sets of $G$, or the set of all vertex covers of $G$, etc. For such problems, sets are identified with the corresponding bit-vectors. Throughout this section, we will interchangeably use strings in $\Q{n}$ to denote subsets of $[n]$ and vice versa.

For the graph problems posed above, the problem $\Phi\textsc{-Subset}$ is in $\p$, and for an instance $I$ we are interested in finding the element (set) in $\mathcal{F}_I$ that has minimum (or maximum) weight (size). 

\begin{problem}[$\Phi\textsc{-Min}$] \textbf{Input:} An instance $I \in \{0,1\}^*$ to $\Phi$. 

    \textbf{Output:} $A \in \mathcal{F}_I$, of minimum weight if $\mathcal{F}_I$ is non-empty.
\end{problem}
\begin{problem}[$\Phi\textsc{-Max}$]
    \textbf{Input:} An instance $I \in \{0,1\}^*$ to $\Phi$. 
    
    \textbf{Output:} $A \in \mathcal{F}_I$, of maximum weight if $\mathcal{F}_I$ is non-empty.
\end{problem}

An example of $\Phi\textsc{-Min}$ is Minimum Vertex Cover and an example of $\Phi\textsc{-Max}$ is Maximum Independent Set. For an instance $I$ for these problems, we use $\OPT_{\phimax}(I)$ and $\OPT_{\phimin}(I)$ to denote the size of the sizes of the largest and smallest sets in $\mathcal{F}_I$ respectively (if $\mathcal{F}_I$ is non-empty). We also use $\calf_{I, \text{min}}$ and $\calf_{I, \text{max}}$ to denote the subsets of $\calf_{I}$ consisting of the elements of smallest and largest weight respectively.

Now, we are interested in finding approximately maximally diverse solutions to the $\Phi\textsc{-Min}$ and $\Phi\textsc{-Max}$ problems, that are also approximately optimal. In the following definition of bi-approximation, let $C_1 \geq 1$ and $C_2 \leq 1$.
\begin{problem}[$(C_1, C_2)\text{-}\divmin$]
    \textbf{Input:} An instance $I$ to the implicit set system $\Phi$, $s \in \IN$ \\
    \textbf{Output:} $S^* \subseteq \mathcal{F}_I$ with $s$ elements such that every $z \in S^*$ has weight at most $C_1 \cdot \OPT_{\phimin} (I)$, and $\PD(S^*)  \geq C_2 \cdot \max_{S \subseteq \mathcal{F_{I, \text{min}}}, |S|=s} \PD(S)$
\end{problem}

\noindent
For the next definition, let $C_1 \leq 1$ and $C_2 \leq 1$.

\begin{problem}[$(C_1, C_2)\text{-}\divmax$]
    \textbf{Input:} An instance $I$ to the implicit set system $\Phi$, $s \in \IN$ \\
    \textbf{Output:} $S^* \subseteq \mathcal{F}_I$ of $s$ elements such that every $z \in S^*$ has weight at least $C_1 \cdot  \OPT_{\phimax} (I)$, and $\PD(S^*)  \geq C_2 \cdot \max_{S \subseteq \mathcal{F_{I, \text{max}}}, |S|=s} \PD(S)$
\end{problem} 

\subsection{Isometric reductions}\label{sec:isometric}

Our first set of applications results from \Cref{thm:sch-heavy-full} on finding diverse satisfying assignments for a $k$-CNF formula that has Hamming weight at least (or at most) a prescribed value $W \in [n]$. Using ``isometric'' reductions between popular NP-complete optimization problems and SAT, we obtain bi-criteria approximation algorithms for diverse solutions of many NP-complete optimization problems. We formally define such reductions first.
\begin{definition}[Isometric Reduction]
    Consider two implicit set systems $\Phi_1$ and $\Phi_2$. A isometric reduction from $\Phi_1$ to $\Phi_2$ is given by a computable function $f$ and a family of computable functions $\{g_I \}$ for every instance $I$ of $\Phi_2$. The function $f$ takes as input an instance $I_1 \in \{0,1\}^*$ of $\Phi_1$ with $\Phi(I_1)=\left(n_1, \mathcal{F}_1\right)$ and outputs an instance $I_2$ of $\Phi_2$ with $\Phi(I_2)=\left(n_2,\calf_2\right)$ such that $n_2=n_1$ and $|\calf_2|=|\calf_1|$. The function $g_{I_2}$ is a bijective function $g_{I_2}: \calf_2 \to \calf_1$, that has the following properties.
    \begin{itemize}
        \item For each $A \in \mathcal{F}_2$, $|A|=|g_{I_2}(A)|$.
        \item For any $A_1, A_2 \in \mathcal{F}_2$, $d_H(A_1, A_2)= d_H(g_{I_2}(A_1), g_{I_2}(A_2))$.
    \end{itemize}
\end{definition}

An isometric reduction preserves the geometry of the solution space. This implies the following theorem.

\begin{theorem} \label{thm:isometricreduction}
    Consider two implicit set systems $\Phi_1$ and $\Phi_2$ such that there exists an isometric reduction $\left( f, \{g_I\}\right)$ from $\Phi_1$ to $\Phi_2$. Suppose there exists an algorithm that solves the $(C_1, C_2)\textsc{-Diverse-}\Phi_2\textsc{-Min}$ problem with input instance $I$ and $s \in \IN$, running in time $\tau(n,s,|I|)$. Then, given an instance $I_1$ for $\Phi_1$, and $s \in \IN$, there exists an algorithm for $(C_1, C_2)\textsc{-Diverse-}\Phi_1\textsc{-Min}$ running in time $\tau_f+ \tau(n, s, |f(I_1)|)+\tau_{g_{I_{2}}}$. Here, $\tau_f$ and $\tau_{g_{I}}$ denote the time it takes to compute the functions $f$ and $g_I$.     
\end{theorem}

\noindent
Clearly, an analogous theorem holds for $(C_1, C_2)\textsc{-Diverse-}\Phi_2\textsc{-Max}$ also. We now demonstrate some simple examples of isometric reductions, which imply the results in the first three rows of \Cref{table: table}. We leave the task of finding more interesting isometric reductions to future work. 

\paragraph{Maximum Independent Set: }We begin by noting that an independent set instance can be written as a $2$-CNF formula $\f_{IS}$: for every $v\in V$, we let $x_v \in\set{0,1}$ such that $x_v=1$ if and only if $v$ is chosen in the independent set. For every edge $e = (u,v) \in E$, we define the constraint $\neg x_u \vee \neg x_v $. Note that this constraint is satisfied if and only if at most one vertex participating in the edge is chosen in the independent set. The formula $\f_{IS}$ is a conjunction of all the constraints corresponding to the edges in $E$. Then an independent set of $G$ corresponds to a satisfying assignment of $\f_{IS}$ and vice versa. Moreover, the Hamming weight of a satisfying assignment of $\f_{IS}$ is equal to the size of the corresponding independent set. Finding an independent set of maximum size is therefore equivalent to finding a satisfying assignment of $\f_{IS}$ of maximum Hamming weight. Moreover, the Hamming distance between two satisfying assignments $z_1,z_2$ corresponding to two independent sets $I_1$ and $I_2$ are preserved, in the sense that $d_H(z_1,z_2) = |I_1\Delta I_2|$, where $\Delta$ denotes the symmetric difference between sets. 

\paragraph{Minimum Vertex Cover: }Every vertex cover instance can be written as a $2$-CNF formula $\f_{VC}$: For every edge $e = (u,v) \in E$, we define the constraint $ x_u \vee x_v $. Note that this constraint is satisfied if and only if at least one vertex participating in the edge is chosen in the vertex cover. The formula $\f_{VC}$ is a conjunction of all the constraints corresponding to the edges in $E$, which implies that a vertex cover of $G$ corresponds to a satisfying assignment of $\f_{VC}$ and vice versa, and the Hamming weight of a satisfying assignment of $\f_{IS}$ is equal to the size of the vertex cover. Finding a vertex cover of minimum size  is therefore equivalent to finding a satisfying assignment of $\f_{VC}$ of minimum Hamming weight, and it can be seen that this reduction is isometric.

\paragraph{Minimum $d$-hitting set:} Recall that an instance of the $d$-hitting set problem consists of a family $\mathcal{S}$ of subsets of $[n]$ of size $d$, with the output being a subset of $[n]$ of minimum size that has a non-empty intersection with each subset in $\mathcal{S}$. This can easily be written as a $d$-CNF formula $\f$ as follows. For every set $S \in \mathcal{S}$, we define a clause $C_S$ which is a disjunction of all the non-negated literals corresponding to the elements in $S$, with the formula $\f$ being the conjunction of the clauses corresponding to each $S \in \mathcal{S}$. Finding hitting set of minimum size corresponds to finding a satisfying assignment to this formula of minimum Hamming weight and it can be seen that this reduction is isometric as well.  

\begin{remark}
    We note that the problems of diverse vertex cover and diverse hitting set have been studied in the setting of parameterized complexity by~\cite{baste2019fpt, baste2022diversity}. However, in these works the focus is on obtaining optimal solutions with optimal diversity and their results are not directly comparable to ours. A typical runtime from the existing results is of the type $2^{s\ell}$ where $s$ is the number of solutions required and $\ell$ is the size of a solution (e.g., the size of the minimum vertex cover). Note that in some settings, $s \ell = \Omega(n^{\alpha})$, for some $\alpha >1$, rendering the above running time of $2^{n^{\alpha}}$. Our results in \Cref{thm:isometricreduction} state that at the cost of relaxing both the quality of the solutions obtained and for approximating the maximum dispersion, the running time can be reduced to $\text{poly(s)} \cdot o\left(2^n\right)$.
\end{remark}

\subsection{Local feasibility search }\label{sec:lfs}
What about problems for which we cannot define an isometric reduction to $k$-SAT? For several of those problems, we point out that the techniques developed in \Cref{sec:sch} are very general and can be adapted to deal with several optimization problems. For the applications in this section, we restrict our attention to minimization problems. We start with defining a version of local search for subset problems similar to \sch's local search for $k$-SAT. 

\begin{definition}(Parameterized approximately-local feasibility search - $(\alpha,c)$-PLFS) An $(\alpha,c)$-PLFS algorithm for an implicit set system $\Phi$ takes as input an instance $I$ for $\Phi$, $A \in \Q{n}$, and $t \in \IN$, and if there exists a feasible solution $A' \in \mathcal{F}_I$ such that $d_H(A,A') \leq t$, outputs an $A^* \in \mathcal{F}_I$ such that $d_H(A, A^*) \leq \alpha t$ in time $c^t \cdot n^{O(1)}$. 
\end{definition}
When $\alpha=1$ we just call the algorithm a PLFS algorithm. Note that there are several examples of problems admitting PLFS algorithms. For example, the algorithms $\LS_{1,k}$ and $\LS_{3, k-1}$ described in \Cref{sec:sch} for $k$-SAT. We also note that this is the exact same definition of a local search used in \Cref{sec:sch}, generalized to subset problems. 

\begin{remark}
    Notice that a PLFS algorithm only searches for any feasible solution in $B(A,t)$, where $B(A,t)$ is the Hamming ball of radius $t$ around $A$. We note that this is potentially easier than searching for a solution of minimum weight in $B(A,t)$. Indeed, for several graph problems, the existence of an algorithm running in time $f(t) \cdot n^{O(1)}$ that finds a solution of minimum weight in $B(A,t)$ is unlikely~\cite{fellows2012local}.
\end{remark}
 
\begin{theorem}[From PLFS to Dispersion] \label{thm:PLFS}
    Let $\Phi$ be an implicit set system that admits an $(\alpha,c)$-PLFS. 
    Then, for every $0 < \delta \leq \frac{2(1+\alpha)}{c-1}$, there exists an algorithm that takes as input an instance $I$ to $\Phi$, $s \in \IN$, and, if $|\calf_{I, \text{min}}| \geq s$, outputs $S^* \subseteq \mathcal{F}_I$ of size $s$ such that $|A| \leq (1+\delta) \OPT_{\phimin}(I)$ for all $A \in S^*$, and $\PD(S^*) \geq \frac{1}{2}(1-\delta) \max_{S \subseteq \mathcal{F_{I, \textsc{min}}}, |S|=s} \PD(S)$. This algorithm runs in time $s^3 \cdot n^{O(1)} \cdot \frac{2^n c^{R}}{\binom{n}{R} }$, where $R=\floor{\frac{\delta n}{2(1+\alpha+\delta)}}$.
    In particular, when $\delta = \frac{2(1+\alpha)}{c-1}$, this algorithm runs in time $O^*\left( s^3 \cdot \left( \frac{2}{1+1/c} \right)^n \right)$.
\end{theorem}
\begin{proof} 
    We note that the $(\alpha,c)$-PLFS for $\Phi$ has the same guarantees that $\LS_{\alpha,c}$ has for $k$-SAT. This implies that all the theorems in \Cref{sec:sch}, and in particular \Cref{{thm:sch-heavy-full}} can be carried over to implicit set systems.
\end{proof}

The question now is, which problems admit a PLFS algorithm? In the field of parametrized complexity, there is a huge body of work on FPT algorithms parametrized by the solution size. While this does not directly imply PLFS algorithms, the framework of monotone local search by Fomin, Gaspers, Lokshtanov, and Saurabh~\cite{ConicSearch} provide a bridge connecting PLFS to FPT algorithms. 

\paragraph{Monotone local search:} For an implicit set system $\Phi$, the cone of length $t$ starting at a set $A \in \Q{n}$ is defined to be $C(A,t):= \{A' \in \Q{n}: A \subseteq A' \text{  and } |A \Delta A'| \leq t\}$. $\Phi$ admits a parameterized local monotone search algorithm if there exists an algorithm taking an instance $I$ of $\Phi$, a set $A \in \Q{n}$ and $t \in [n]$ as input, and if $C(A,t) \bigcap \calf_I$ is non-empty, outputs some $A^* \in C(A,t) \bigcap \calf_I$ in time $c^t\cdot n^{O(1)}$ for some constant $c>1$. 

Now, we prove that for the class of \emph{hereditary} problems, the concepts of parameterized local feasibility search and parameterized local monotone search are in fact, equivalent. We remind the reader that we are dealing with minimization problems only. 
\begin{definition}
    An implicit set system $\Phi$ is called hereditary if for all instances $I$ of $\Phi$ such that $\Phi(I)=\left(n ,\calf\right)$, $\calf$ satisfies the property that for any $A \subseteq B \subseteq [n]$, $A \in \mathcal{F}$ implies that $B \in \mathcal{F}$.
\end{definition}

\begin{lemma} \label{lem:heriditary}
    If a hereditary implicit set system $\Phi$ admits a parameterized monotone local search algorithm, then it also admits a PLFS algorithm that runs in the same time, and vice versa. 
\end{lemma}
\begin{proof}
    For any $A \in \Q{n}$, Suppose there exists $B \in \mathcal{F}_I \bigcap B(A,t)$. Now, consider the set $A \bigcup B$. Because $\Phi$ is a hereditary set system, $A \bigcup B \in \mathcal{F}_I$. Further, because $A \bigcup B=A \bigcup (A \Delta B)$, $d_H(A,A \bigcup B) \leq t$ and $A \bigcup B \in C(A,t)$. Hence, we can use the monotone local search algorithm to output some $A^* \in \mathcal{F}_I \bigcap C(A,t) \subseteq \mathcal{F}_I \bigcap B(A,t)$, which implies the existence of a PLFS for $\Phi$. On the other hand, Suppose $\Phi$ admits a PLFS algorithm. Suppose there exists some $B \in C(A,t) \bigcap \calf_I$. Because $C(A,t) \subseteq B(A,t)$, the PLFS algorithm is guaranteed to output $A^* \in \calf_I \bigcap B(A,t)$. Now consider the set $A \bigcup A^*$. Because $A \bigcup A^*=A \bigcup (A \Delta A^*)$, this implies that $A \bigcup A^* \in C(A,t) \bigcap \calf_I$, which implies that $\Phi$ admits a monotone local search algorithm. 
\end{proof}
\Cref{lem:heriditary}, along with \Cref{thm:PLFS} implies the existence of a $c$-PLFS for many combinatorial problems that were studied in~\cite{ConicSearch}. We select the same problems and present them in the table below (instantiated at $C_1=3/2, C_2=1/4$), along with our results on isometric reductions.

\begin{table}[htt]
\centering
\begin{tabular}{|l|l|l|l|}
\hline
Problem & Extension  & MinOnes \cite{ConicSearch} & $s$-Dispersion \\
 & \cite{ConicSearch} & One exact solution & Bi-approx \\
\hline
\textsc{$d$-Hitting Set $(d \geq 3)$} & $d^k$ & $(2 - \frac{1}{d})^n$ & \Cref{thm:isometricreduction} \\
\textsc{Vertex cover } & $2^k$ & $1.5^n$ & $s^3 \cdot 1.5486^n$ \\
\textsc{Maximum independent Set} & $2^k$ & $1.5^n$ & $s^3 \cdot1.5486^n$ \\
\hline
\textsc{Feedback Vertex Set} & $3.592^k$ & $1.7217^n$ &  $s^3 \cdot1.6420^n$\\
\textsc{Subset Feedback Vertex Set} & $4^k$ & $1.7500^n$ & $s^3 \cdot1.6598^n$ \\
\textsc{Feedback Vertex Set in Tournaments} & $1.6181^k$ & $1.3820^n$ &  $s^3 \cdot1.5162^n$\\
\textsc{Group Feedback Vertex Set} & $4^k$ & $1.7500^n$ &  $s^3 \cdot1.6598^n$\\
\textsc{Node Unique Label Cover} & $|\Sigma|^{2k}$ & $(2 - \frac{1}{|\Sigma|^2})^n$ & \Cref{thm:PLFS} \\
\textsc{Vertex $(r,\ell)$-Partization $(r,\ell \leq 2)$} & $3.3146^k$ & $1.6984^n$ &  $s^3 \cdot1.6289^n$ \\
\textsc{Interval Vertex Deletion} & $8^k$ & $1.8750^n$ &  $s^3 \cdot1.7789^n$\\
\textsc{Proper Interval Vertex Deletion} & $6^k$ & $1.8334^n$ & $s^3 \cdot1.7284^n$ \\
\textsc{Block Graph Vertex Deletion} & $4^k$ & $1.7500^n$ & $s^3 \cdot1.6598^n$ \\
\textsc{Cluster Vertex Deletion} & $1.9102^k$ & $1.4765^n$ & $s^3 \cdot1.5415^n$ \\
\textsc{Thread Graph Vertex Deletion} & $8^k$ & $1.8750^n$ &  $s^3 \cdot1.7789^n$\\
\textsc{Multicut on Trees} & $1.5538^k$ & $1.3565^n$ & $s^3 \cdot1.51^n$ \\
\textsc{3-Hitting Set} & $2.0755^k$ & $1.5182^n$ & $s^3 \cdot1.5544^n$ \\
\textsc{4-Hitting Set} & $3.0755^k$ & $1.6750^n$ & $s^3 \cdot1.6167^n$ \\
\textsc{$d$-Hitting Set $(d \geq 3)$} & $(d - 0.9245)^k$ & $(2 - \frac{1}{d-0.9245})^n$ & \Cref{thm:PLFS} \\
\textsc{Min-Ones 3-SAT} & $2.562^k$ & $s^3 \cdot1.6097^n$ & \Cref{thm:sch-heavy-full} \\
\textsc{Min-Ones $d$-SAT $(d \geq 4)$} & $d^k$ & $(2 - \frac{1}{d})^n$ & \Cref{thm:sch-heavy-full} \\
\textsc{Weighted $d$-SAT $(d \geq 3)$} & $d^k$ & $(2 - \frac{1}{d})^n$ & \Cref{thm:sch-heavy-full} \\
\textsc{Weighted Feedback Vertex Set} & $3.6181^k$ & $1.7237^n$ & $s^3 \cdot1.6432^n$ \\
\textsc{Weighted 3-Hitting Set} & $2.168^k$ & $1.5388^n$ &  $s^3 \cdot1.5612^n$\\
\textsc{Weighted $d$-Hitting Set $(d \geq 4)$} & $(d - 0.832)^k$ & $(2 - \frac{1}{d-0.832})^n$ & \Cref{thm:PLFS} \\ \hline
\end{tabular}
\caption{\scriptsize{The second column contains the time taken to obtain one exact solution using methods in~\cite{ConicSearch}. The third Column contains the time taken to solve $(3/2, 1/4)\text{-}\divmin$ (except for Maximum Independent Set, where $(1/2, 1/4)\text{-}\divmax$ is solved) 
}
}
\label{table: table2}
\end{table}
We remark that both the isometric reduction and the PLFS approaches give $s$-dispersion algorithms for $d$-Hitting Set. However, the second approach yields an algorithm with better guarantees because the monotone search for $d$-hitting set is faster than the local search for $d$-SAT~\cite{fomin2010iterative}.

\input{App_SchoningDispersed}

\subsection{Approximating $\optm$ for CSPs}
\label{sec:csp}
It is not hard to see that \sch's parametrized local search algorithm can be used to find diverse solutions for $k$-ary CSP's as well, that is, \Cref{obs:schmain} generalizes to CSPs~\cite{schoning1999probabilistic}[Section 3]. Formally, we prove the following theorem. 
\begin{theorem}[\sch approximating $\optm$ for CSPs] \label{thm:sch-for-mindisp-CSP}
    Let $\Psi$ be a any constraint satisfaction problem over the alphabet $\{0,1\}$, and $s \in \IN$. with the maximum arity of the constraints being $k$. For $0 < \delta \leq \min\{1, \frac{4(k-1)}{(k-2)^2} \}$, there exists an algorithm taking $\Psi$ and $s$ as input and, if $\Psi$ has at least $s$ distinct satisfying assignments, outputs a set $S^*$ of $s$ of satisfying assignments to $\Psi$ such that $\PD(S^*) \geq \frac{1}{2}\left(1-\delta \right) \optm(\Psi,s)$. It runs in time $O^*\left(s^3 \cdot  \frac{2^n (k-1)^{R}}{\binom{n}{R} }\right)$, where $R=\floor{\frac{\delta n}{2(2+\delta+ \frac{2}{k-2})}}$.  
    
\end{theorem}

%% file: App_SchoningDispersed.tex
\subsection{\sch's and PPZ algorithms run faster if $\Om$ contains \\dispersed solutions}\label{sec:fast}
In this section, we show that if $\Om$ contains a dispersed subset, then \sch's algorithm as well as the PPZ algorithm find a satisfying assignment to $\f$ faster. Let $\Om$ denote the set of satisfying assignments to $\f$. 
 
For every $r \in [n]$, we define 
\[ N_{r}:= \max \{|S| : S \subseteq \Om, \PD(S) \geq r \} \]
Note that from the definition of $N_r$, for every $r \in [n]$, there exists a set $S_r \subseteq \Om$ of size $N_r$, such that the balls of radius $\lfloor \frac{r}{2} \rfloor$ around each $z^* \in S_r$ are disjoint. We also note that $N_{0,\f} \geq N_{1,\f} \geq \cdots \geq N_{n,\f}$.
\begin{theorem}\label{thm:schfaster}
Let $\f$ be a $k$-CNF formula. If $\f$ is satisfiable, \sch's algorithm succeeds in finding a satisfying assignment within $O^*\left( \frac{2^n \left(1-1/k\right)^n}{N_{\floor{2n/k}}} \right)$ iterations.
    
\end{theorem}
If the solution space $\Om$ contains a code of minimum distance $2r=2n/k$, with $N_{2r} \geq 2^{n(1-H((2r-1)/n))}$ (using the Gilbert Varshamov bound), which is equal to $2^{n(1-H(2/k-1/n))}$. When $k \geq 6$, this gives an exponential improvement. 

\medskip \noindent

To prove this, recall \Cref{obs:schmain} and \sch's algorithm as described in \Cref{sec:schprelims}. 
\obssch*

\noindent
It consists of sampling $z$ uniformly at random from $\{0,1\}^n$ and performing a \sch walk for $3n$ steps starting from $z$. If $\f$ is satisfiable, for each $r \in [n]$, with probability $\frac{1}{2^n} \cdot \binom{n}{r}$, $z$ is at Hamming distance $\leq r$ from a satisfying assignment, and we can calculate the probability that the \sch walk ends in a satisfying assignment to be at least $\frac{1}{2^n} \binom{n}{r}\frac{1}{(k-1)^r}$. Hence, setting $r=\floor{n/k}$, we can lower bound this probability by $\left(\frac{1}{2}\left(1+\frac{1}{k-1}\right)\right)^n$, using~\Cref{schcalc}.  

\medskip \noindent
However, we now note that due to the definition of $N_r$, for each $0 \leq r \leq \floor{n/2}$, there exist $N_{2r}$ satisfying assignments to $\f$, with the Hamming balls of radius $r$ around then being disjoint. Hence, for each $r \in [\floor{n/2}]$, with probability at least $\frac{N_{2r}\binom{n}{r}}{2^n}$, $z$ is at distance $r$ from a satisfying assignment, when chosen uniformly at random from $\{0,1\}^n$. This means that the success probability of the \sch walk can be calculated to be at least $N_{\floor{2n/k}} \cdot \left(\frac{1}{2}\left(1+\frac{1}{k-1}\right)\right)^n$. This probability is clearly better than the probability of success for \sch's algorithm. Hence, we obtain that the running time of \sch's algorithm with a dispersion guarantee equals 

\[  \frac{2^n \left(1-1/k\right)^n}{N_{\floor{2n/k}}} \;.\]

Now we note that we can prove a similar statement for the PPZ algorithm. 
 
\begin{theorem}
    \label{thm:PPZfaster}
    Let $\f$ be a $k$-CNF formula. If $\f$ is satisfiable, the PPZ algorithm succeeds in finding a satisfiable assignment to $\f$ within $O^*\left(  \frac{ 2^{n-n/k}}{N_{\floor{2n/k}}}\right)$ iterations.
\end{theorem} 
\begin{proof}
 \Cref{lem:anchor:diam} implies that for any satisfying assignment $z \in \Om$, $\PPZMod$ outputs a satisfying assignment to $\f$ within distance $n/k$ of $z$ with probability at least $n^{-O(1)} \cdot 2^{-n+n/k}$. Now, due to the definition of $N_r$, there exists a set $S \subseteq \Om$ of size $N_{\floor{2n/k}}$ such that the balls of radius $\floor{n/k}$ around them being disjoint. Hence, the running time of the PPZ algorithm with the dispersion guarantee is $O^*\left(\frac{1}{N_{\floor{2n/k}}} \cdot 2^{n-n/k}\right)$. As before, if the solution space $\Om$ contains a code of minimum distance $2r=2n/k$, this leads to an exponential improvement. 
 \end{proof}

%% file: app_more-sch.tex
\section{More \sch-type algorithms}

\subsection{The case of small $k$ and small $\delta$.}
\label{app:more-sch}
In this section, we start by showing that by using $\LS_{1, k}$ instead of $\LS_{\left(1+2/(k-2)\right),k-1}$, we can handle the case of $k=2$, and beat \Cref{thm:sch-for-dia} for some $\delta$. We state the theorem to begin with. We define $\tau_1(\delta,k, n)$ to be $\frac{2^n k^{R_1}}{\binom{n}{R_1}} \text{ where } R_1=\floor{\frac{\delta n}{2(2+\delta)}}$ for each $\delta \in \left( 0, \min\{1, \frac{4}{k-1} \}\right]$
\begin{restatable}[\sch for \D]{theorem}{schdiam}
    \label{thm:sch-for-dia-two}
    Let $\f$ be a $k$-CNF formula on $n$ variables, for any $k \geq 2$. For each $0 < \delta \leq \min\{1, \frac{4}{k-1} \}$, there exists an algorithm taking $\f$ as input and running in time $O^*\left( \tau_1(\delta, k, n)\right)$ that outputs $z_1^*, z_2^* \in \Om$ such that $d_H(z_1^*, z_2^*) \geq \frac{1}{2}\cdot \left(1-\delta \right) \D(\f)$, if $\f$ is satisfiable.
\end{restatable}
\begin{proof}
    The proof is identical to the proof of~\Cref{thm:sch-for-dia} by using $\alpha=1, c=k$.  
\end{proof}
\begin{remark}
    We can also use $\LS_{1, k}$ instead of $\LS_{\left(1+2/(k-2)\right),k-1}$ in all the algorithms for dispersion as well, to get identical theorems, with the running times using $\tau_1(\delta,k, n)$ instead of $\tau(\delta, k, n)$ and $\delta$ can be tuned between $0$ and $\frac{4}{k-1}$ instead of $\frac{4}{k-1}\left(1+\frac{2}{k-2}\right)^2$. We do not restate all of them for the sake of brevity. 
\end{remark}
\paragraph{Comparison between \Cref{thm:sch-for-dia} and \Cref{thm:sch-for-dia-two}.} Not only does \Cref{thm:sch-for-dia-two} handle the case of $k=2$, it also outperforms \Cref{thm:sch-for-dia} for some cases of $\delta$. As before, we define $a_{k, \delta}$ such that $a_{k,\delta}^n=\tau(\delta,k,n)$ and $b_{k,\delta}^n=\tau_1(\delta,k,n)$. We plot $a_{k,\delta}$ and $b_{k,\delta}$ together, for different values of $k$ to illustrate the comparison. It can be seen that for $k=3$, the algorithm in \Cref{thm:sch-for-dia-two} always outperforms the algorithm in \Cref{thm:sch-for-dia}, and for larger values of $k$, it outperforms for smaller valuse of $\delta$. 
\begin{figure}[ht]
    \centering
    \includegraphics[scale=0.5]{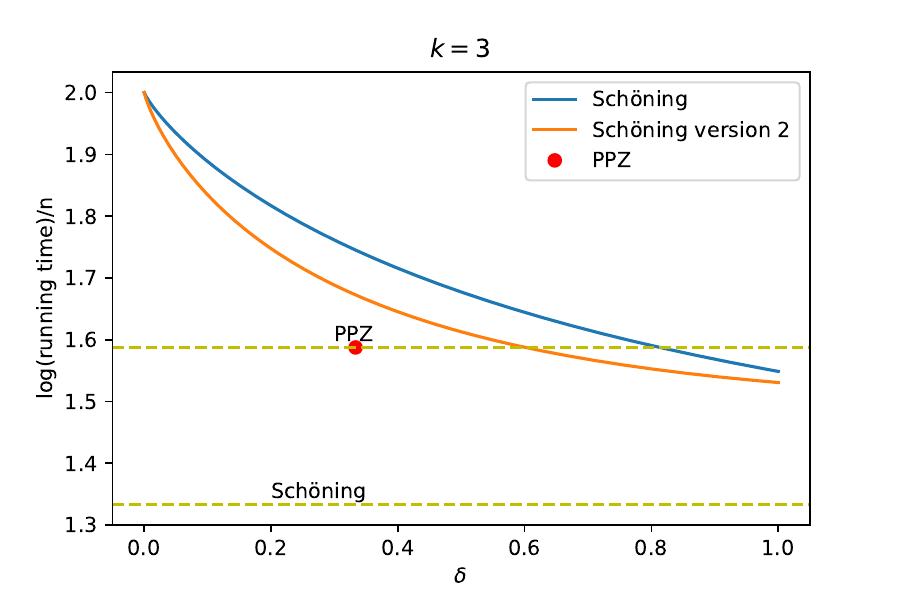}
    \includegraphics[scale=0.5]{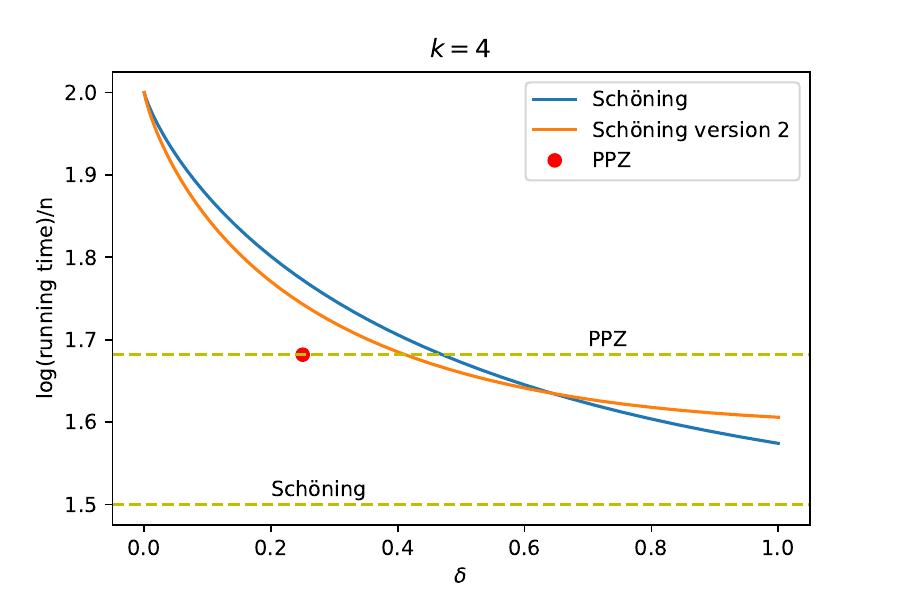}

    \includegraphics[scale=0.5]{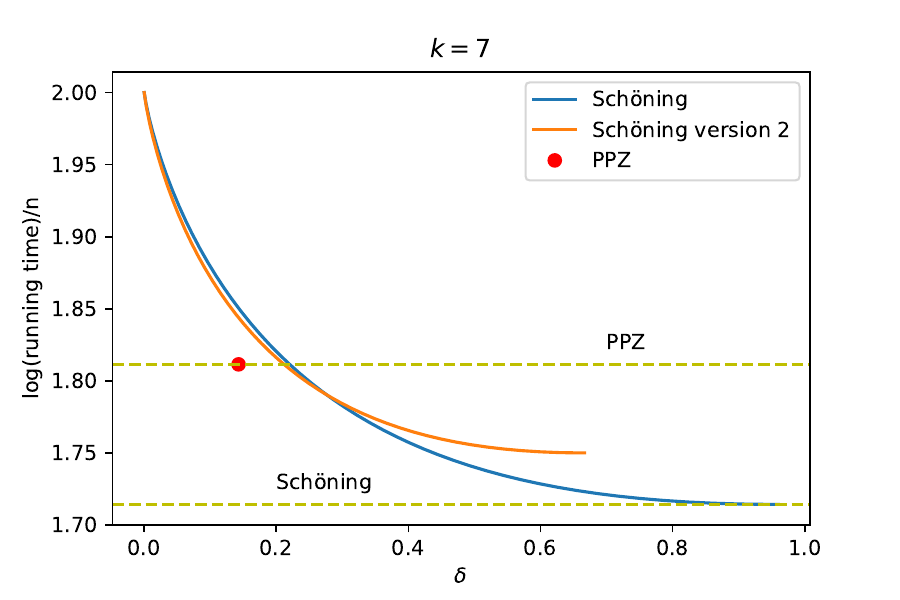}
    \includegraphics[scale=0.5]{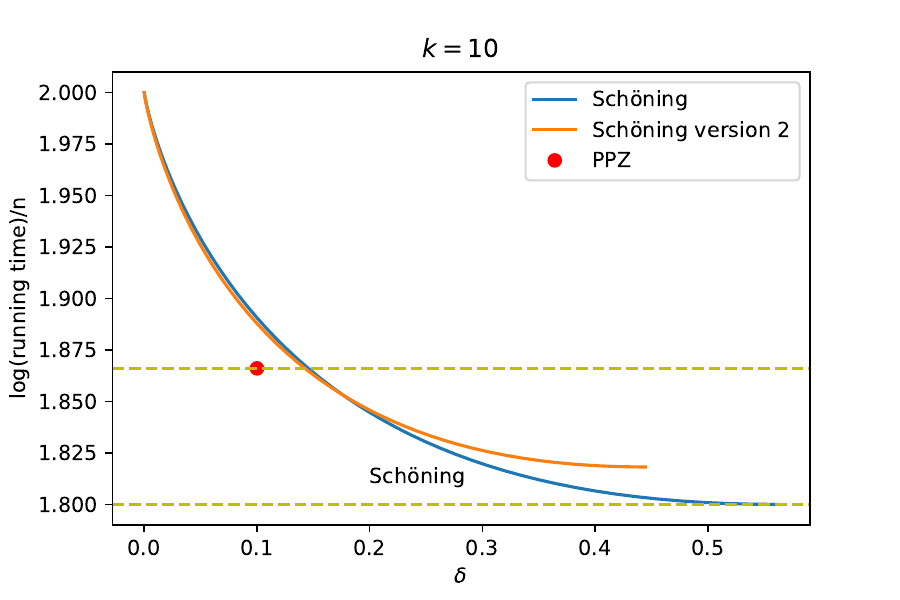}
    \caption{Plots of $a_{k,\delta}$ (labeled \sch) and $b_{k,\delta}$ (labeled \sch version 2)with respect to $\delta$ (on x-axis) for $k=3,4,5$.}
    \label{fig:enter-label}
\end{figure}

\subsection{\sch-based algorithm for the sum dispersion measure: Proof of \Cref{thm:sch-for-sumdisp}.}
\label{app:schsum}
To start with, we restate \Cref{thm:sch-for-sumdisp}. 
\schopts*
To prove this theorem, we show that \sch's algorithm can be modified to be a farthest point oracle for $\sumd$. 
\begin{lemma} \label{lem:schFarthestSum}
    Consider a local search algorithm $\LS_{\alpha, c}$. Then, for every $0 < \delta \leq \frac{2(1+\alpha)}{c-1}$, there exists an algorithm  running in time $s^2 \cdot n^{O(1)} \cdot \frac{2^n c^{R}}{\binom{n}{R} }$, where $R=\floor{\frac{\delta n}{2(1+\alpha+\delta)}}$ that takes as input $\f$ and a multi-set $S \subseteq \{0,1\}^n$ of size $s$ and if $\f$ is satisfiable, outputs $z^* \in \Om$ such that $\sumd(z^*, S) \geq \left(1-\delta\right) \cdot \max_{z' \in \Om} \sumd(z', S)$ with probability at least $1-2^{-2n}$. 
 \end{lemma}
 
 \begin{proof}
    Consider the following algorithm, that in line 6, uses the $\ALS_{\alpha, c, \delta}$ subroutine defined in the proof of \Cref{lem:anchor:sch}. 
    
    \begin{algorithm}[H] \label{alg:schFarthestSum}
 \KwIn{A $k$-CNF formula $\f$, $S \subseteq \{0,1\}^n, |S|=s$}
 \KwOut{$z^* \in \Om$ with $\sumd(S,z^*) > \left(1-\delta\right) \max_{z \in \Om}\sumd(S,z)$ if $\f$ is satisfiable, $\perp$ otherwise.}
 Set $z^*=\perp, D=0$. \\
 \SetKwFor{RepeatTimes}{repeat}{times:}{endfor}
 \For{$r \in \{0,1,2,\dots, n\}$}{
    \For{$z \in S$}{
        Let $t:=\min\left\{ \floor{\frac{\delta r}{1+\alpha}}, R\right\}$ \\
        \RepeatTimes{$n^{O(1)} \cdot |A_{r-t,r+t} (z)|$}{
        $u:=\ALS_{\alpha, c ,\delta}(\f,z,r)$\\
    \If{$u$ satisfies $\f$ and $\sumd(S,z^*) > D$}{
        $z^* \gets u, D \gets \sumd(S,u)$.
    }
    
    }
    }
    }
    Output $z^*$
 \caption{\textsf{\sch-Farthest-Sum}}
\end{algorithm}
    
     Suppose that there exists $z_0 \in \Om$ with $\sumd(z_0, S) = r$. Then, there must exist $z \in S$ such that $r':=d_H(z, z_0) \leq r/s$. Now, suppose that $\ALS_{\alpha, c ,\delta}(\f,r',z)$ outputs $z^* \in \Om$ such that $d_H(z^*, z_0) \leq \delta r'$. The triangle inequality then implies that $$\sumd(z^*, S)= \sum_{z \in S} d_H(z^*,z) \geq \sum_{z \in S} \left( d_H(z_0,z)-d_H(z^*, z_0)\right) \geq r-\delta s r' \geq \left(1-\delta \right)r \; .$$
    
    $\ALS_{\alpha,\delta}(\f, z, r')$ outputs such a $z^*$ if $y$, the starting assignment it samples, $y$ is within distance $t$ of $z_0$, This happens with probability at least $\frac{\binom{n}{t}}{|A_{r-t , r+t}(z)|}$
    Hence, it is sufficient to call $\LS_{\alpha, \delta}(\f, z, r')$ $n^{O(1)} \cdot \frac{|A_{r'-t , r'+t}(z)| c^t}{\binom{n}{t}  }$ (where $t$ is the value chosen corresponding to $r'$ by $\ALS_{\alpha, c ,\delta}$) times for each $r' \in \{0,1,2,\dots, n\}$ and $z \in S$, to ensure that the algorithm outputs $z^* \in \Om$ with $\sumd(S,z^*) \geq \left(1-\delta\right) \cdot \max_{z \in \Om} \sumd(z,S)$ with probability $1-2^{-2n}$, and the reminder of this proof (bounding the running time) is identical to the proof of \Cref{lem:anchor:sch}. The factor of $s^2$ in the running time bound arises from the fact that computing $\sumd(\cdot, \cdot)$ takes at most $s n$ time and due to the fact that we iterate over all $s \in S$ in line 3.  
 \end{proof}

\subsubsection*{Proof of \Cref{thm:sch-for-sumdisp}}
The proof of \Cref{thm:sch-for-sumdisp} is similar to that of \Cref{thm:ppz-for-sumdisp}. We use the algorithm $\textsf{\sch-Farthest-Sum}$ (\Cref{lem:schFarthestSum}) as a $(1-\delta)$-farthest point oracle, in the algorithm defined by \Cref{lem:sumdispersion}. The approximation, running time guarantees and the range of $\delta$ that the two algorithms handle follows from the bounds stated in~\Cref{lem:schFarthestSum} for $c=k, \alpha=1$ and $c=k-1, \alpha=1+\frac{2}{k-2}$. 

%% file: App_Dispersion.tex
\section{Technical lemmas using approximate farthest point oracles}\label{sec:approxfarthest}
\label{sec:dispersion}
In this section, we design approximation algorithms for computing $\opts(\f,s)$ and $\optm(\f,s)$, proving \Cref{thm:ppz-for-sumdisp}, \Cref{thm:ppz-for-mindisp}, \Cref{thm:sch-for-sumdisp} and \Cref{thm:sch-heavy-full}. 

\subsection{Sum dispersion: the proof of \Cref{lem:sumdispersion}}

\sumdispersion*
To prove this lemma, consider the following algorithm, which is the same as the algorithm studied in~\cite{cevallos2019improved}, with the small difference being that we deal with multi-sets instead of sets.

\begin{algorithm}[H] \label{alg:sumdispersion}
 \KwIn{A $k$-CNF formula $\f$, a number $s$, the oracle $\mathcal{O}$}
 \KwOut{$S \in \Om^s$ with $\SPD(S) \geq \max \{\frac{1}{2}(1-\delta),\frac{(1-\delta)(s-1)}{(1+\delta)s+(1-\delta)} \}\cdot \opts(\f,s)$ if $\f$ is satisfiable, $\perp$ otherwise.}
 \SetKwFor{RepeatTimes}{repeat}{times:}{endfor}
 Use the PPZ algorithm (or \sch's algorithm) to find a satisfying assignment $z_1^*$ to $\f$. \;
 Set $S \gets \{z_1^*\}$ \;
 \For{$i \in \{2,3,\dots, s\}$}{
    $z^*:=\mathcal{O}(\f,S)$\\
    $S \gets S \bigcup \{z^*\}$
    
 }
 \RepeatTimes{$s^2n$}{
    \For{$z \in S$}{
    $z^*:=\mathcal{O}(\f,S \setminus \{z\})$ \\
    \If{$\sumd(S \setminus \{z\},z^*) > \sumd( S \setminus \{z\},z)$}{
        $S \gets S \setminus \{z\} \bigcup \{z^*\}$
        }
    
    }
 }
 Output $S$
 \caption{Algorithm for Sum Dispersion}
 \end{algorithm}
Because $\mathcal{O}$ is invoked at most $s^2 n$ times during the whole duration of the algorithm, the union bound implies that with probability at most $1-o(1)$, $\mathcal{O}$ behaves as a $1-\delta$-approximate farthest point oracle in every iteration (because $\mathcal{O}$ behaves as a $1-\delta$-approximate farthest point oracle in every iteration with probability $1-2^{-2n}$).

\medskip \noindent
The algorithm described above combines $\mathcal{O}$ with the well-known farthest point insertion algorithm~\cite{ravi1994heuristic} for dispersion in metric spaces to get an algorithm that outputs a multiset $S \subseteq \Om$ with $|S|=s$ with the property that $\SPD(S) \geq \frac{1 - \delta}{2} \cdot \opts(\f,s)$.

\medskip \noindent
If $s$ is large, we can further improve the approximation factor by repeatedly employing the following natural local search procedure on the set $S$. For each $z \in S$, we use the farthest point oracle with $\f$ and $S \setminus \{z\}$ as input. If $z^*$, the satisfying assignment output by the farthest point oracle satisfies $\sumd(z^*, S\setminus\{z\}) > \sumd(z, S\setminus\{z\})$ (which is equivalent to the condition that $\SPD(S \setminus \{z\} \bigcup \{z^*\})) > \SPD(S)$), we replace $z$ by $z^*$ in $S$. We show that at the end of $s^2n$ iterations, $\SPD(S) \geq \frac{(1-\delta)(s-1)}{(1+\delta)s+(1-\delta)} \cdot \opts(\f,s)$. Because this local search procedure only increases the value of $\SPD(S)$, this would complete the proof of \Cref{lem:sumdispersion}.

 \medskip \noindent
 We start with lower bounding $\SPD(S)$ at the end of the farthest point insertion procedure. We start with proving the following lemma. For a multiset $S$, denote $|S|$ to be its cardinality counting multiplicities, and for two multisets $A$ and $B$, we use $d_H(A,B)=\sum_{a \in A, b \in B} d_H(a,b)$
 \begin{observation} \label{obs:multisettriangle}
     Let $A,B \subseteq \{0,1\}^n$ be two multisets. There exists $b \in B$ such that $\sumd(A,b) \geq \frac{|A|}{|B|(|B|-1)} \cdot \SPD(B)$.
 \end{observation}
 \begin{proof}
     Suppose not. This implies that for every $b \in B$,  $\frac{1}{|A|} \cdot \sumd(A,b) < \frac{1}{|B|(|B|-1)}\SPD(B)$. We now use the triangle inequality and the definition of $\SPD(B)$ to claim that
     \begin{align*}
         \SPD(B) = \frac{1}{2}\sum_{b_1, b_2 \in B} d_H(b_1, b_2) & \leq \frac{1}{2|A|} \cdot \sum_{a \in A} \sum_{b_1, b_2 \in B, b_1 \neq b_2} d_H(b_1, a) + d_H(b_2, a) \\
         &=  \sum_{b_1, b_2 \in B, b_1 \neq b_2}\frac{1}{2|A|} \cdot \sumd(A, b_1) +\frac{1}{2|A|} \cdot \sumd(A, b_2) \\
         &< |B|(|B|-1) \cdot \left( \frac{1}{2|B|(|B|-1)}\SPD(B) +\frac{1}{2|B|(|B|-1)}\SPD(B)\right) \\ &=\SPD(B) \; ,
     \end{align*}
     which is a contradiction.
 \end{proof}
 \noindent
 Now, let $\Sopt \subseteq \Om$ be a multiset of size $s$ with $\SPD(\Sopt)= \opts(\f,s)$. \Cref{obs:multisettriangle} implies that the step when $|S|=i$, there exists $z \in \Sopt$ with $\sumd(S, z) \geq \frac{i}{s(s-1)} \cdot \opts(\f,s)$. Hence, the point $z^*$ added to $S$ at step $i$ by $\mathcal{O}$ satisfies $\sumd(S,z^*) \geq \frac{i(1-\delta)}{s(s-1)} \cdot \opts(\f,s)$. We now show by induction that once the $i$-th point $z^*$ is added by algorithm, $\SPD(S) \geq \frac{i(i-1)(1-\delta)}{2s(s-1)} \cdot \opts(\f,s)$. This is trivially true when $|S|=1$. Assume that when $|S|=i-1$, $\SPD(S) \geq \frac{(i-1)(i-2)(1-\delta)}{2s(s-1)} \cdot \opts(\f,s)$. Because the point $z^*$ added to $S$ next satisfies $\sumd(S,z^*) \geq \frac{(i-1)(1-\delta)}{s(s-1)} \cdot \opts(\f,s)$, the value of $\SPD(S)$ at the end of round $i$ is at least $\left(\frac{(i-1)(i-2)(1-\delta)}{2s(s-1)} +  \frac{(i-1)(1-\delta)}{s(s-1)} \right) \cdot \opts(\f,s) = \frac{i(i-1)(1-\delta)}{2s(s-1)} \cdot \opts(\f,s)$. Since $i=s$, at the end of the farthest point insertion procedure, $\SPD(S) \geq \frac{(1-\delta)}{2} \cdot \opts(\f,s)$. 
 
 \medskip \noindent
 We now show that at the end of the local search procedure, $\SPD(S) \geq \frac{(1-\delta)(s-1)}{(1+\delta)s+(1-\delta)}\cdot \opts(\f,s)$. At each step of the procedure, either $\SPD(S)$ increases by at least $1$, or $\SPD(S)$ remains unchanged (such an $S$ is called a `local optimum'). Observe that at any iteration, if the value of $\SPD(S)$ is unchanged at the end of it, it also does not change during any of the later iterations. Because $\SPD(A) \leq s^2n$ for any multiset $A \subseteq \{0,1\}^n$ of size $s$, the algorithm reaches a local optimum within $s^2 n$ iterations.

 \medskip \noindent
 Consider any set $S$ which is a local optimum, and a set $\Sopt$, such that $\SPD(\Sopt)=\opts(\f,s)$. Because the local search employed on $S$ does not improve $\SPD(S)$, the property of $\mathcal{O}$ implies that
\[ \sumd(S \setminus \{x\},x) \geq (1 - \delta) \cdot \sumd(S \setminus \{x\},y) \text{ for all } x \in S, y \in \Om \; .\]
Specifically, this holds for all $y \in \Sopt$. Hence, we can sum over all $x \in S, y \in \Sopt$ to obtain that
\begin{equation} \label{eqn:swapping:termination}
     s \cdot \SPD(S) \geq\frac{(1-\delta)(s-1)}{2}  \cdot  d_H(S, \Sopt) \;,
\end{equation}
where $d(S, \Sopt)= \sum_{x \in S, y \in \Sopt}d_H(x,y)$. We now use the inequality that
\begin{equation}
    d_H(S, \Sopt) \geq \SPD(S) +  \SPD(\Sopt) \;.
\end{equation}
This follows from the fact that the Hamming metric is of negative type~\cite[Lemma 1]{cevallos2019improved} We now use this in \Cref{eqn:swapping:termination} to obtain that
\begin{equation*}
    s \cdot \SPD(S) \geq \frac{(1-\delta)(s-1)}{2} \left(\SPD(S) +\SPD(\Sopt)\right)
\end{equation*}
Rearranging, this implies that $\SPD(S) \geq \frac{(1-\delta)(s-1)}{(1+\delta)s+(1-\delta)}\cdot \opts(\f,s)$. 
\subsection{Min Dispersion: the proof of \Cref{lem:mindispersion}}

\mindispersion*
To prove this lemma, consider the following farthest point insertion algorithm, originally studied by Gonzales~\cite{gonzalez1985clustering}. 

\begin{algorithm}[H]
 \KwIn{A $k$-CNF formula $\f$, a number $s$}
 \KwOut{$S \in \Om^s$ with $\PD(S) \geq \frac{1}{2}\left(1-\delta \right) \cdot \optm(\f,s)$ if $\f$ is satisfiable, $\perp$ otherwise.}
 \SetKwFor{RepeatTimes}{repeat}{times:}{endfor}

 Use the PPZ algorithm (or \sch's algorithm) to find a satisfying assignment $z_1^*$ to $\f$. \;
 Set $S \gets \{z_1^*\}$ \;
 \For{$i \in \{2,3,\dots, s\}$}{
    $z^*:=\mathcal{O}(\f,S)$\\
    $S \gets S \bigcup \{z^*\}$
    
 }

\caption{Min Dispersion}
\end{algorithm}
\noindent
Because $\mathcal{O}$ is invoked at most $s$ times during the whole duration of the algorithm, the union bound implies that with probability at most $1-o(1)$, behaves as approximate farthest point oracle each time it is invoked. 
\medskip\noindent
Next, we show that at the end of the algorithm, $\PD(S) \geq \frac{1}{2}\left(1-\delta \right) \cdot \optm(\f,s)$ using induction. First, observe that $\PD(\{z_1^*, z_2^*\}) \geq \frac{1}{2} (1-\delta) \optm(\f,2)$ using the triangle inequality. Suppose that before the $i$-th iteration of the algorithm, $|S|=i-1$ and $\PD(S) \geq \frac{1}{2}\left(1-\delta\right) \cdot \optm(\f,i-1)$. Let $\Sopt \subseteq \Om$ be a set of size $i$ with $\PD(\Sopt)=\optm(\f,i)$. \Cref{obs:farthest-first} (stated and proved below) implies that there exists $x \in \Sopt$ such that $\mind(x,S) \geq 1/2 \cdot \optm(\f,i)$. Hence, the assignment added to $S$ at step $i$, $z^*$ satisfies $\mind(S,z^*) \geq \frac{1}{2}\left(1-\delta\right) \cdot \optm(\f,i)$, which implies that $\PD(S)\geq \frac{1}{2}\left(1-\delta\right) \cdot \optm(\f,i)$ at the end of the $i$-th iteration.

\begin{restatable}{observation}{gonzales}[Farthest Point insertion]
\label{obs:farthest-first} 
Let $A, B \subseteq \set{0,1}^n$ be two subsets with $\size{A} < \size{B}$. Then there exists $b\in B$ such that $\mind(b,A) \geq 1/2 \cdot \PD(B)$. 
\end{restatable}
\begin{proof}
    The proof is by contradiction.
We assume that $d_H(b,A) < 1/2 \cdot \PD(B)$ for all $b\in B$. Since  $\size{A} < \size{B}$ then, by pigeonhole principle, it must mean that there are exists an assignment $a \in A$ and two distinct assignments $b, b'\in B$ such that $d_H(b,A) = d_H(b,a)$ and $d_H(b',A) = d_H(b', a)$. Then, by triangle inequality and our assumption,we have that:
$$d_H(b,b') \leq d_H(b,a) + d_H(b',a) < \PD(B) \;.$$
However, by definition, we have that $\PD(B) \leq d_H(b,b')$ and so we obtain a contradiction.
\end{proof}

%% file: App_DistinctSum.tex
\section{On returning sets instead of multisets}
\label{app:distinctsum}
In this section, we extend our results for $\opts$ to $\opts_{\neq}$. Recall that the algorithm for $\opts$ returned a multiset of size $s$ that is an approximation of $\opts(\f,s)$. 

We showed that given any multi-set $T \subseteq \{0,1\}^n$, the sequence obtained from repeatedly sampling from $\{0,1\}^n \times S_n$ and running $\PPZMod$ contains $z^*$ such that $\sumd(z^*, T) \geq \frac{k-1}{k+1} \cdot r_{sum}$, where $r_{sum}=\max_{z \in \Om} \sumd(z, T)$. In this section, we extend that result to the $\opts_{\neq}$ problem. 

\begin{lemma}\label{lem:anchor:sumdistinct}
    Let $\f$ be a satisfiable $k$-SAT formula, $T \subseteq \{0,1\}^n$ be a set of size $t = o \left(\frac{n}{\log(n)}\right)$, and $r_{\text{sum}}= \max_{z \in \Om \setminus T}\sumd (z,T)$. Let $y$ and $\pi$ be chosen uniformly at random from $\{0,1\}^n$ and $S_n$ respectively. The probability that $\PPZMod(\f, y, \pi)$ outputs $z^* \in \Om \setminus T$ with $\sumd(z^*,T) \geq \frac{k-1}{k+1} \cdot r_{\text{sum}}$ is at least $ 2^{-n+n/k-o(1)}$
\end{lemma}
\begin{proof}
        Let $i^*= \lfloor \frac{k-1}{k+1} \cdot r_{sum}\rfloor$. In \Cref{lem:anchor:sum}, we showed that $\tau(\f,U_{i^*}) \geq \frac{1}{2n} \cdot 2^{-n(1-1/k)}$ outputs a satisfying assignment in $U_{i*}$. What we need to prove however is a lower bound on $\tau(\f, U_{i^*} \setminus T)$.

        \noindent
        Note that we can expand $\tau(\f, U_{i^*})$ to 
     \[ \tau(\f, U_{i^*}) = 2^{-n(1-1/k)} \sum_{z \in U_{i^*}} 2^{-\deg(z)/k} \geq \frac{1}{2n} \cdot 2^{-n(1-1/k)} \; . \]
     But notice that in the proof of \Cref{lem:anchor:sum}, we actually proved something stronger. We proved that
     \[ |U_{i^*}| 2^{- \frac{1}{k|U_{i^*}|}\sum_{z \in U_{i^*}} \deg(z) } \geq \frac{1}{2n} \]
     Now, we need to lower bound $\tau(\f, U_{i^*} \setminus T)$. From now on, we use $U$ to refer to $U_{i^*}$, $U_1$ to refer to $U_{i^*} \setminus T$ and $U_2$ to refer to $U_{i^*} \bigcap T$. We use $S_1$ to denote the set of edges between $U_1$ and $\Om \setminus U$, $S_2$ to refer to the set of edges between $U_2$ and $\Om \setminus U$, $S_3$ to denote the edges between $U_1$ and $U_2$ and $E(U_1)$ and $E(U_2)$ to refer to edges between with both endpoints in $U_1$ and $U_2$ respectively. What we need to lower bound is the quantity 
     \[ |U_1| 2^{-  \frac{2|E(U_1)|+|S_1|+|S_3|}{k|U_1|}} \;,  \]
     assuming the lower bound
     \[  |U| 2^{-  \frac{2|E(U_1)|+2|E(U_2)|+2|S_3|+|S_1|+|S_2|}{k|U|}} \geq \frac{1}{2n}\]
     Now, we let $|U|= \alpha |U_1|$, with $\alpha$ being well defined because $U_1$ is non-empty. We note that
     \begin{align*}
         |U_1| 2^{-  \frac{2|E(U_1)|+|S_1|+|S_3|}{k|U_1|}} &= \frac{|U|}{\alpha} \cdot2^{-  \frac{2|E(U_1)|+|S_1|+|S_3|}{k|U_1|}} \\
         & \geq |U_1| \cdot \left(2^{- \frac{2|E(U_1)|+2|E(U_2)|+2|S_3|+|S_1|+|S_2|}{k |U|}}\right)^{\alpha}\\
         & \geq |U_1| \cdot \left(\frac{1}{2n|U|}\right)^\alpha = |U_1| \cdot \left(\frac{1}{2n\alpha |U_1|}\right)^\alpha\ = |U_1|^{1-\alpha} (2n \alpha)^{-\alpha} 
     \end{align*}
     We now note that $\alpha |U_1|=|U| \leq |U_1| + t$, which implies that $|U_1| \leq \frac{t}{\alpha -1}$. Hence (because $1-\alpha$ is negative),
    \begin{equation*}
         |U_1| 2^{-  \frac{2|E(U_1)|+|S_1|+|S_3|}{k|U_1|}}\geq \left(\frac{t}{\alpha -1}\right)^{1-\alpha} (2n \alpha )^{-\alpha} \geq t^{1-\alpha} (2n)^{-\alpha} \alpha^{-1}= (2tn)^{-\alpha} (\alpha t)^{-1}
     \end{equation*}
     Further, $\alpha \leq t-1$, which means that this quantity can be lower bounded by $(2tn)^{-5t}$. Now, using the fact that $t=o\left( \frac{n}{\log(n)}\right)$, we get that this quantity is bounded below by $2^{-o(n)}$ which implies that $\tau(\f, U_1) \geq 2^{-n(1-1/k) - o(1)}$. 
\end{proof}

Now, we have shown that there exists an approximate farthest point oracle for computing 
     $\max_{z \in \Om \setminus T}\sumd (z,T)$, as long as $|T|=o\left(\frac{n}{\log(n)}\right)$. Hence, we can now use this approximate farthest point oracle in the algorithms by Gonzales and Cevallos, Eisenbrand, and Zenklusen, proving the following theorem.  
     \begin{restatable}{theorem}{ppzsumdispneq}[PPZ approximating $\opts_{\neq}(\f,s)$] \label{thm:ppz-for-sumdispneq}
    Let $\f$ be a $k$-CNF formula on $n$ variables. There exists a randomized algorithm  that takes $\f$ and an integer $s \geq 1$ as input and if $\f$ is satisfiable and has at least $s$ satisfying assignments, with probability at least $1-o(1)$, outputs a set $S^* \subseteq \Om$ of size $s$ such that:
    \begin{enumerate}
        \item $\SPD(S^*) \geq \frac{1}{2}\cdot \left(1-\frac{2}{k+1}\right) \cdot \opts(\f,s)$ if $s \leq \floor{\frac{3k+1}{k-1}}$.
        \item $\SPD(S^*) \geq \frac{k-1}{k+3}\left(\frac{1-\frac{1}{s}}{1+\frac{k-1}{(k+3)}\cdot \frac{1}{s} }\right) \cdot \opts(\f,s)$ if $s \geq \ceil{\frac{3k+1}{k-1}}$. 
    \end{enumerate}
    The algorithm runs in time $O^*\left(2^{n-n/k + o(n)}\right)$, as long as $s=o\left(\frac{n}{\log(n)}\right)$
\end{restatable}

%% file: App_MinOnes.tex
\section{Relationship between \textsc{Min-Ones} and \textsc{Farthest-Point}}
\label{app:minones}
In this section, we point out that a farthest point oracle can be derived from an algorithm that outputs a satisfying assignment to $\f$ with minimum weight. This problem, formally called $\textsc{Min-Ones}-k-\text{SAT}$ has an exact algorithm that runs in time $O^*\left( \left(2-\frac{1}{k}\right)^n\right)$. For simplicity, we define the decision versions of these problems. 
\begin{problem}[\textsc{Min-Ones}]
    \textbf{Input: }A $k$-CNF formula $\f$, $r \in [n]$. \\
\textbf{Output: }Yes, if there exists $z^* \in \Om$ such that $|z^*| \leq r$, No otherwise.
\end{problem}
\begin{problem}[\textsc{Farthest-Point}]
    \textbf{Input: }A $k$-CNF formula $\f$, $z \in \Q{n}$, $r \in [n]$. \\
\textbf{Output: }Yes, if there exists $z^* \in \Om$ such that $d_H(z^*, z) \geq r$, No otherwise.
\end{problem}

 We now show that the problems \textsc{Min-Ones} and \textsc{Farthest-Point} are equivalent to each other. 
 \begin{lemma}
     There exists a reduction, running in $n^{O(1)}$ time, from \textsc{Min-Ones} to \textsc{Farthest-Point} and vice versa
 \end{lemma}
 \begin{proof}
     We first show that there exists a polynomial time reduction from \textsc{Min-Ones} to \textsc{Farthest-Point}. Let $(\f, r)$ be an instance of $\textsc{Min-Ones}$. For any satisfying assignment $z^* \in \Om$, $|z^*| \leq r$ if and only if $d_H(z^*, \mathbf{1}) \geq n-r$. Hence, the instance $(\f, r)$ of $\textsc{Min-Ones}$ can be reduced to the instance $(\f,\mathbf{1}, n-r)$ of \textsc{Farthest point} (where $\mathbf{1}$ is the all $1$'s vector).

     Now, consider any instance $(\f, z, r)$ of \textsc{Farthest point}. Now, we create a new $k$-CNF formula $\f_{z}$, by ``rotating'' the formula $\f$. To be precise, we define $\f_{z}$ as follows. For any $j \in [n]$ such that $z_j=0$, we replace every occurrence of the literal $z_j$ in $\f$ with $\Bar{z}_j$ and every occurrence of $\Bar{z}_j$ in $\f$ with $z_j$. Hence, if $z^*$ is a satisfying assignment to $\f$, the assignment $z^* \oplus \Bar{z}$ is a satisfying assignment to $\f_{z}$, where $\Bar{z}$ is the antipode of $z$. Hence, there exists $z^* \in \Om$ with $d_H(z,z^*) \geq r$, if and only if $z^* \oplus \Bar{z} \in \Omega_{\f_z}$, and $|z^* \oplus z|\leq n-r$, i.e. if $\textsc{Min-Ones}(\f_z, n-r)$ returns yes. 
 \end{proof}

%% file: App_UniformSampling.tex
\section{Using Uniform Sampling to generate diverse satisfying assignments}\label{app:uniformv2}

Let $\alg$ be an algorithm that takes in $\f$ as input, and in $O^*(a^n)$ running time, outputs a satisfying assignment to $\Om$ such that each $z \in \Om$ is output with probability $1/|\Om|$ (in other words, it uniformly samples over the space of satisfying assignments). Note that because $k$-SAT is a self reducible problem\footnote{For the class of problems that are `self reducible', counting and sampling are equivalent, and approximate counting and approximate sampling are equivalent as well~\cite{sinclair1989approximate}}, an algorithm for $\# k$-SAT, that counts the number of satisfying assignments can be used to also uniformly sample from the space of satisfying assignments. We define the following algorithm that approximates the diameter of $\Om$, using the uniform sampler $\alg$ as a black box. It runs in time $O^*(b^n)$, where $b^n $ is some time budget that we choose.

\begin{theorem}
    Let $\f$ be a $k$-SAT formula with at least $2$ satisfying assignments. Let $\alg$ be an algorithm that uniformly samples satisfying assignments to $k$-SAT instances that runs in time $O^*(a^n)$. Consider any $b > a$. There exists an algorithm that runs in time $O^*(b^n)$ and with probability $1-o(1)$, and outputs two satisfying assignments $z_1, z_2 \in \Om$, with $d_H(z_1, z_2) \geq \min\{ \frac{1}{2}, H^{-1}(\log(b/a)) \} \cdot \D(\f)$.
\end{theorem}

\begin{proof}
Consider the following algorithm. 

\begin{algorithm}[H]
 \label{alg:uniformsampling}
 \KwIn{A $k$-CNF formula $\f$}
 \KwOut{$z_1, z_2 \in \Om$, with $d_H(z_1, z_2) \geq \min\{ \frac{1}{2}, H^{-1}(\log(b/a)) \} \cdot \D(\f)$}
 \SetKwFor{RepeatTimes}{repeat}{times:}{endfor}
    Find a satisfying assignment $z_1 \in \Om$ using any $k$-SAT solver.\\
    Let $z_2 \gets z_1$, $D \gets 0$\\
    \RepeatTimes{$n^{O(1)} \cdot (b/a)^n$}{
    Run $\alg$ to output $z' \in \Om$.
    \If{$d_H(z',z_1) > D$}{
        Set $z_2 \gets z'$, $D \gets d_H(z',z_1)$
    
    }
    
    }
 \caption{Using uniform samplers to approximate $\D(\f)$}
 \end{algorithm}
We consider two cases, based on the size of $\Om$. The first case is when $(b/a)^n \geq |\Om| \log(|\Om|)$. Let $z^*$ be a satisfying assignment that maximizes the hamming distance from $z_1$. In each iteration of the loop, $z^*$ is sampled with probability $\frac{1}{|\Om|}$, and hence the probability that the algorithm never encounters $z^*$ is upper bounded by 
$$(1-1/|\Om|)^{|\Om| \log(|\Om|)} \leq e^{-\log(|\Om|)} \leq \frac{1}{|\Om|}.$$ 
Hence, with probability $1-o(1)$, the algorithm finds $z^*$ and outputs $1/2$-approximation for $\D(\f)$. 

The second case is when $|\Om| \log(|\Om|) > (b/a)^n$. In this case, consider the ball of radius $r=n \cdot  H^{-1}\big(\log\big(\frac{b(1-2\log(n)/n)}{a}\big)\big)$ around $z_1$. We show that, in each iteration, $\alg$ finds a point $z_2$ outside the ball with probability at least $1/2$. This is because if $|\Om| > \frac{1}{n} \cdot (b/a)^n$, and the volume of the ball of radius $r$ around $z_1$ is at most $2^{H(r/n)n}=\big( (b/a) (1-2 \log(n)/n)\big)^n$. The ratio of these quantities is at most $n(1-2\log(n)/n)^n \approx n \cdot e^{2 \log(n)} <1/2 $, for sufficiently large $n$. Hence, for sufficiently large $n$, at least half of the points in $\Om$ have to be located outside this ball, and in each iteration of the loop, a satisfying assignment at distance at least $r$ from $z_1$ is found with probability at least $1/2$, and hence the loop finds at least one of these assignments with probability $1-o(1)$. This proves that the algorithm finds, with probability at least $1-o(1)$, two satisfying assignments at distance at least $r$ from each other, and since the diameter of $\f$ is at most $n$, the approximation factor achieved is at least $r/n \geq H^{-1}((b/a) (1-2 \log(n)/n)) = H^{-1}((b/a))-o(1)$.
\end{proof}

\paragraph{Comparison to our results:} We now perform some calculations assuming an approximation guarantee of $1/H^{-1}(\log(b/a))$ for the above algorithm. We use the state of the art existing algorithms for $ \#k $ -SAT to come up with bounds for the run-time and approximation factors and compare them with our more 'geometry-based' sampling algorithms we propose. 

For $3$-SAT, the best known approximate counting algorithms are by Schmitt and Wanka~\cite{schmitt2013exploiting}, running in time $O^*(1.51426^n)$. 

Hence, we can calculate the approximation factor this algorithm achieves for $k=3$, where the budget $b=2^{2/3}$. To do that, we plug in $b=2^{2/3}$, and $a=1.51426$ in $\frac{1}{H^{-1}(b/a)}$, which is $1/123$. This means that the sampling algorithm gives a $1/123$-approximation factor for the diameter of $3$-SAT. On the other hand, our \Cref{thm:ppz-for-dia} gives a $1/3$-approximation ratio in the same running time. We remark that this gap widens as $k$ increases.

%% file: App_schcalc.tex
\section{\sch run time calculation} \label{schcalc}
\begin{lemma} \label{lem:binommax}
    For every $t \in [n]$, $\frac{2^n}{\binom{n}{t} c^{-t}} \geq \frac{1}{n^{O(1)}} \cdot \frac{2^n}{ \binom{n }{ \floor{\frac{n}{c+1}} } c^{-\floor{\frac{n}{c+1}}}} = \frac{1}{n^{O(1)}}\cdot \left(\frac{2}{1+\frac{1}{c}}\right)^n$.
\end{lemma}
\begin{proof}
    Let $t= \mu n$. let $f(\mu)=\mu^{-\mu} (1-\mu)^{\mu-1}$ We use \Cref{obs:binomapprox} to show that 
    \[ \frac{2^n}{\binom{n}{r} c^{-r}} \geq \frac{1}{n^{O(1)}} \cdot \left( \frac{2}{c^{-\mu} f(\mu)} \right)^n \; .\]
    Using \Cref{obs:binomderivative}, we can see that the derivative of $g(\mu)=c^{-\mu} f(\mu)$ is $ g'(\mu)= \left( - \ln(c) + \ln(1-\mu)- \ln(\mu)\right)g(\mu)$. Because $g(\mu)$ is always positive, we can see that the derivative is a decreasing function of $\mu$, with $g'\left(\frac{1}{c+1}\right)=0$. Hence, the minimum value of $g(\mu)$ is attained when $\mu=\frac{1}{c+1}$. Substituting $\mu=\frac{1}{c+1}$ in $g(\mu)$, we get $1+\frac{1}{c}$. This implies that out of all $t \in [n]$, $t=\floor{\frac{n}{c+1}}$ (up to a $n^{O(1)}$ factor) minimizes the value of $\frac{2^n}{\binom{n}{t} c^{-t}}$. 
\end{proof}